\documentclass[a4paper,twocolumn,11pt,unpublished,nopdfoutputerror]{quantumarticle}
\usepackage[T1]{fontenc}
\usepackage{amsmath,amssymb,amsthm,mathtools,booktabs,array,microtype,graphicx,placeins}
\usepackage[numbers,sort&compress]{natbib}
\let\label\relax
\usepackage{xr-hyper}
\usepackage[hidelinks]{hyperref}
\AtBeginEnvironment{thebibliography}{\raggedright}
\title{Quantum change point metrology}
\author{Xu~Chen}
\orcid{0009-0003-8499-4080}
\email{cnc@hebust.edu.cn}
\affiliation{School of Sciences, Hebei University of Science and Technology,
Shijiazhuang, Hebei 050018, People's Republic of China}
\author{Xue~Ma}
\email{maxue@he.chinamobile.com}
\affiliation{Network Management Center, China Mobile Communications Group
Hebei Co., Ltd., Shijiazhuang, Hebei 050000, People's Republic of China}
\newtheorem{qcpMtheorem}{Theorem}
\expandafter\def\csname theHqcpMtheorem\endcsname{qcp-M.\thesection.\arabic{qcpMtheorem}}

\newtheorem{qcpMstructure}{Proposition}
\expandafter\def\csname theHqcpMstructure\endcsname{qcp-M.structure.\arabic{qcpMstructure}}
\newtheorem{qcpCtheorem}{Theorem}[section]
\expandafter\def\csname theHqcpCtheorem\endcsname{qcp-C.\thesection.\arabic{qcpCtheorem}}
\newtheorem{qcpCproposition}[qcpCtheorem]{Proposition}
\newtheorem{qcpClemma}[qcpCtheorem]{Lemma}

\newtheorem{qcpDtheorem}{Theorem}[section]
\expandafter\def\csname theHqcpDtheorem\endcsname{qcp-D.\thesection.\arabic{qcpDtheorem}}

\newtheorem{qcpDcorollary}[qcpDtheorem]{Corollary}
\newtheorem{qcpStheorem}{Theorem}
\expandafter\def\csname theHqcpStheorem\endcsname{qcp-S.\thesection.\arabic{qcpStheorem}}
\newtheorem{qcpSproposition}{Proposition}
\expandafter\def\csname theHqcpSproposition\endcsname{qcp-S.\thesection.\arabic{qcpSproposition}}
\newtheorem{qcpSlemma}{Lemma}
\expandafter\def\csname theHqcpSlemma\endcsname{qcp-S.\thesection.\arabic{qcpSlemma}}

\expandafter\def\csname theHqcpScorollary\endcsname{qcp-S.\thesection.\arabic{qcpScorollary}}

\expandafter\def\csname theHqcpPtheorem\endcsname{qcp-P.\thesection.\arabic{qcpPtheorem}}
\newtheorem{qcpPproposition}{Proposition}
\expandafter\def\csname theHqcpPproposition\endcsname{qcp-P.\thesection.\arabic{qcpPproposition}}
\newtheorem{qcpPlemma}{Lemma}
\expandafter\def\csname theHqcpPlemma\endcsname{qcp-P.\thesection.\arabic{qcpPlemma}}

\expandafter\def\csname theHqcpPcorollary\endcsname{qcp-P.\thesection.\arabic{qcpPcorollary}}
\newtheorem{qcpUtheorem}{Theorem}
\expandafter\def\csname theHqcpUtheorem\endcsname{qcp-U.\thesection.\arabic{qcpUtheorem}}
\newtheorem{qcpUlemma}[qcpUtheorem]{Lemma}
\newtheorem{qcpUcorollary}[qcpUtheorem]{Corollary}

\newtheorem{qcpBtheorem}{Theorem}
\expandafter\def\csname theHqcpBtheorem\endcsname{qcp-B.\thesection.\arabic{qcpBtheorem}}

\newtheorem{qcpVtheorem}{Theorem}
\expandafter\def\csname theHqcpVtheorem\endcsname{qcp-V.\thesection.\arabic{qcpVtheorem}}
\newtheorem{qcpVlemma}{Lemma}
\expandafter\def\csname theHqcpVlemma\endcsname{qcp-V.\thesection.\arabic{qcpVlemma}}

\newcommand{\loc}{\mathrm{loc}}
\newcommand{\amp}{\mathrm{amp}}
\newcommand{\G}{\mathcal G}
\newcommand{\tr}{\operatorname{Tr}}
\newcommand{\Unif}{\operatorname{Unif}}
\newcommand{\ket}[1]{|#1\rangle}
\newcommand{\bra}[1]{\langle#1|}
\newcommand{\proj}[1]{\ket{#1}\bra{#1}}
\newcommand{\HS}{\mathrm{HS}}
\newcommand{\exprel}{\operatorname{exprel}}

\makeatletter
\newcommand{\qcpresetvolume}[1]{%
  \@removefromreset{equation}{section}%
  \setcounter{section}{0}\setcounter{subsection}{0}\setcounter{subsubsection}{0}%
  \setcounter{equation}{0}\setcounter{figure}{0}\setcounter{table}{0}\setcounter{footnote}{0}%
  \renewcommand{\thesection}{\arabic{section}}%
  \renewcommand{\thesubsection}{\thesection.\arabic{subsection}}%
  \renewcommand{\theequation}{\arabic{equation}}%
  \def\theHsection{qcp-#1.\thesection}%
  \def\theHsubsection{qcp-#1.\thesubsection}%
  \def\theHsubsubsection{qcp-#1.\thesubsubsection}%
  \def\theHequation{qcp-#1.\thesection.\arabic{equation}}%
  \def\theHfigure{qcp-#1.\arabic{figure}}%
  \def\theHtable{qcp-#1.\arabic{table}}%
}
\newcommand{\qcparticlelayout}{%
  \renewcommand\section{\@startsection{section}{1}{\z@}{-3.5ex\@plus -1ex\@minus -.2ex}{2.3ex\@plus .2ex}{\normalfont\Large\bfseries}}%
  \renewcommand\subsection{\@startsection{subsection}{2}{\z@}{-3.25ex\@plus -1ex\@minus -.2ex}{1.5ex\@plus .2ex}{\normalfont\large\bfseries}}%
  \renewcommand\subsubsection{\@startsection{subsubsection}{3}{\z@}{-3.25ex\@plus -1ex\@minus -.2ex}{1.5ex\@plus .2ex}{\normalfont\normalsize\bfseries}}%
  \renewenvironment{abstract}{\small\begin{center}\bfseries\abstractname\end{center}\quotation}{\endquotation}%
}
\newcommand{\qcpvolumetitle}[1]{\par\vspace*{1.5em}\begin{center}{\LARGE #1\par}\end{center}\vspace{1em}}
\makeatother

\begin{document}
\begingroup
\qcpresetvolume{M}
\setlength{\emergencystretch}{3em}

\maketitle
\begin{abstract}
We quantify the extra Bayes risk of locating a quantum change point
and estimating its amplitude with one measurement. The joint and
marginal optimizations use the full record and the same joint prior,
with both parameters unknown. For product records on a known pure state curve, position has
clipped absolute loss in resolution units, and amplitude offset has
squared loss in fixed calibrated units. Ideal steps exhibit certified conflict at finite signal strength, including in phase records of sixteen qubits, and recover both informative marginal optima at strong signal. For finite offset prior width,
a retained quantum boundary region serves position, and the remaining
suffix serves amplitude. Physical recovery holds when strength grows
more slowly than the square root of record length. 
Known smooth rises also admit leading recovery in the coherent experiment as strength
grows and width shrinks. For the quintic profile, we establish two residual cost scaling laws. The residual cost is $\Theta(s^{-2})$ as $s\to\infty$. Taking the strong signal limit at fixed $w$ yields a limiting cost of $\Theta(w^3)$ as $w\to0$. In the resolved local model, quadratic position loss admits a common optimum within the covariant class. These results link onset geometry and decision loss to Bayesian measurement compatibility in quantum change point estimation.
\end{abstract}
\section{Introduction}
\label{qcp-M-sec:introduction}

A quantum sensor record can encode both when a signal begins and how
strongly it acts. Estimating these quantities connects change point
detection with quantum sensing and metrology
\cite{qcp-M-degen2017,qcp-M-giovannetti2004,qcp-M-giovannetti2006,qcp-M-giovannetti2011}.
For an abrupt onset, moving the boundary changes a short segment of
the record, while changing the amplitude affects the whole suffix.
This spatial structure offers a way to allocate information between
the two tasks. The question is how much precision a common readout can
retain for each.

Quantum change point studies have established identification with
minimum error \cite{qcp-M-sentis2016}, exact identification with
inconclusive outcomes \cite{qcp-M-sentis2017}, and sequential local
strategies \cite{qcp-M-sentis2018}. Related settings include boundaries
between unknown domain states \cite{qcp-M-llorens2025}, multiple
anomalies \cite{qcp-M-llorens2024}, and sequential detection delay
\cite{qcp-M-fanizza2023,qcp-M-grootveld2026}.
Here, the amplitude that makes the boundary distinguishable is also an
estimation target. Each unknown parameter affects the information
available for estimating the other, connecting the problem to nuisance
parameters and semiparametric estimation
\cite{qcp-M-suzuki2020,qcp-M-tsang2020}. Measurement compatibility determines whether local precision limits for
several parameters can be attained by one readout
\cite{qcp-M-gill2000,qcp-M-matsumoto2002,qcp-M-ragy2016}.
Quantum Fisher information and multiparameter bounds relate local
distinguishability to estimation precision
\cite{qcp-M-braunstein1994,qcp-M-paris2009,qcp-M-sidhu2020,qcp-M-szczykulska2016,qcp-M-albarelli2020,qcp-M-demkowicz2020,qcp-M-hayashi2023}.
Bayesian studies characterize incompatibility
\cite{qcp-M-albarelli2026}, derive bounds from monotone metrics
\cite{qcp-M-zhang2026}, and optimize projective readouts
\cite{qcp-M-bernad2026}.
Classical change point theory relates localization to jump size and
transition width \cite{qcp-M-dachian2015,qcp-M-amiri2021}.
For quantum change point records, the central question is how onset
geometry and decision loss determine the common readout cost and its
asymptotic scaling.

Within Bayesian decision theory
\cite{qcp-M-holevo1973,qcp-M-personick1971,qcp-M-tsang2012,qcp-M-rubio2019,qcp-M-rubio2020},
we define the common readout cost as the weighted joint optimum minus
the weighted marginal optima. All three optimizations use the full
record, the same joint prior and complete measurement class,
with both parameters unknown. The model consists of product records on a known smooth curve
of pure states, with independent uniform priors on position and a
calibrated amplitude offset
(Figure~\ref{qcp-M-fig:model-joint-estimation}).
Position has clipped absolute loss in resolution units, and the offset
has squared loss in fixed calibrated units.
A common optimum exists precisely when the positive operator valued measure (POVM) that is optimal for position commutes with the optimal amplitude observable.
Quantum local asymptotic normality and Gaussian equivalence provide
a framework for comparisons with coherent state models
\cite{qcp-M-guta2006,qcp-M-kahn2009,qcp-M-butucea2018}.
We use these comparisons to control the optimal joint risk and both
optimal marginal risks.

We certify a positive common readout cost for the ideal step at a reference signal strength, including for phase records of sixteen qubits. As signal strength grows, a common readout recovers both marginal optima.
Recovery assigns a retained quantum boundary region to position and
the remaining suffix to amplitude. With the offset prior fixed at any
finite positive width, physical recovery holds along $s_n\to\infty$ with
$s_n=o(\sqrt n)$, where $n$ is the record length. Both marginal optimal
risks remain bounded away from zero, as do the risk reductions relative
to using the prior alone.
For known smooth rise profiles, the coherent experiment also admits recovery to leading order as the signal strengthens and the rise width shrinks.
For a quintic rise of known width $w$, we obtain two scaling laws for the residual cost.
The residual cost is $\Theta(s^{-2})$ as $s\to\infty$, uniformly for $s^2w\in[1/2,2]$. Taking the strong signal limit at fixed $w$ yields a limiting cost of $\Theta(w^3)$ as $w\to0$.
For quadratic position loss in the resolved local experiment, a readout based on amplitude prediction attains both marginal optima within the complete covariant class. With clipped absolute position loss, the optimal readout requires a cubic correction whose squared contribution accounts for the cubic width law under the full position prior.

Section~\ref{qcp-M-sec:model} defines the model and estimation task, and Section~\ref{qcp-M-sec:main-results} presents the main results and their mechanisms. Section~\ref{qcp-M-sec:certified-evidence} presents numerical conflict certificates and comparisons with physical records. The Supplementary Material contains the proofs and certificate error budgets.

\section{Model and estimation task}
\label{qcp-M-sec:model}

\subsection{Quantum record and prior information}

The quantum record contains $n$ systems with the same finite dimensional
Hilbert space $\mathcal H$. An unknown boundary $u$ marks the onset
of a response controlled by an unknown calibrated amplitude offset
$z$. The known response curve $\phi(\theta)$ consists of normalized
pure states and is $C^3$ in norm on $[-r_0,r_0]$. Define the state and the squared projective speed at $\theta=0$ by
\begin{equation}
 \begin{gathered}
 \rho(\theta)=\proj{\phi(\theta)},\\
 q_0=\|(I-\rho(0))\phi'(0)\|^2>0.
 \end{gathered}
 \label{qcp-M-eq:curve}
\end{equation}
The positive constant $q_0$ sets the local scale of the curve parameter.

Let $L\in(0,\infty)$ be the known half width of the offset prior.
Position and amplitude offset have independent uniform priors,
\begin{equation}
 (u,z)\sim\Unif[0,1]\otimes\Unif[-L,L].
 \label{qcp-M-eq:prior}
\end{equation}
We fix the calibration $a=2\sqrt{\log2}$,
$s_{\mathrm c}=3a/2$ and $K=s_{\mathrm c}^2$.
All logarithms are natural. The signal center $s>0$ is known.
The offset $z$ determines the calibrated strength $h=s+az$,
which is the coherent field amplitude in the comparison below.
The physical curve parameter is $\theta=h/\sqrt{nq_0}$.

For an ideal step, define $k_n(u)=\lfloor nu\rfloor+1$.
The state on $\mathcal H^{\otimes n}$ is
\begin{equation}
 \rho_n(u,z;s)
 =\rho(0)^{\otimes(k_n-1)}
  \otimes\rho(\theta)^{\otimes(n-k_n+1)}.
 \label{qcp-M-eq:record}
\end{equation}
When $u<1$, the unchanged sites $1,\ldots,k_n-1$ precede the first changed
site $k_n$. Under the position prior, $k_n$ is uniform on
$\{1,\ldots,n\}$; the zero-mass endpoint $u=1$ represents the unchanged
auxiliary state. The position target for a finite record is $u_n=(k_n-1)/n$.
Reports lie in $\mathcal X_n=\{0,1/n,\ldots,(n-1)/n\}$, and $x$ gives
the site estimate $\widehat k=nx+1$.

The condition $(s+aL)/\sqrt{nq_0}\le r_0$ ensures that the full physical parameter
interval lies within the domain of the known curve. Where a construction
requires a strictly positive amplitude $h$, the additional
condition $s>aL$ is imposed.

The coherent comparison experiment is defined on bosonic Fock space
over $L^2[0,1]$ by
\begin{equation}
 \sigma_s(u,z)=\proj{f_{u,h}}_{\rm coh},\qquad
 f_{u,h}=h\mathbf1_{[u,1]}.
 \label{qcp-M-eq:coherent}
\end{equation}
It retains the same prior on $(u,z)$, but has continuous location
target $u$ and report space $\mathcal X_\infty=[0,1]$. For normalized coherent
vectors, the squared fidelity is $\operatorname{Fid}(|f\rangle,|g\rangle)=e^{-\|f-g\|^2}$.
Distance is measured by the half trace distance $D(\rho,\sigma)=\|\rho-\sigma\|_1/2$, which for
these pure states satisfies $D=\sqrt{1-\operatorname{Fid}}$.

\begin{figure*}[t]
\centering
\includegraphics[width=\linewidth]{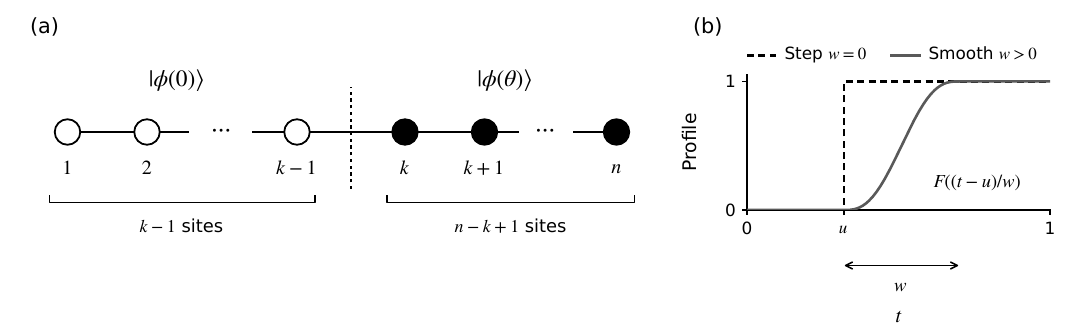}
\caption[Quantum change point model]{Quantum record and onset profiles.
(a) The illustrative phase family
$\ket{\phi(\theta)}=(\ket{0}+e^{i\theta}\ket{1})/\sqrt{2}$.
Sites $1,\ldots,k-1$ are in $\rho(0)$ and sites $k,\ldots,n$ in
$\rho(\theta)$, with first changed site $k=\lfloor nu\rfloor+1$.
(b) Normalized profiles for an ideal step and a smooth rise $F((t-u)/w)$, 
with unknown onset $u$ and known width $w$.}
\label{qcp-M-fig:model-joint-estimation}
\end{figure*}

\subsection{Common measurement and calibrated losses}

A complete POVM $M(dx,dy)$ produces
paired reports $x\in\mathcal X$ for position and $y\in[-L,L]$ for the calibrated
offset $z$, where $\mathcal X$ denotes the report space of the
chosen experiment. Completeness is expressed by
\[
 M(dx,dy)\ge0,\qquad
 \int_{\mathcal X\times[-L,L]}M(dx,dy)=I.
\]
The measurement design depends only on the known model, prior and task weight. The POVM description covers adaptive strategies based on observed outcomes and retained quantum systems, together with randomized and ancillary implementations. Every outcome $(x,y)$ is scored for both tasks, both risks are averaged over the full outcome distribution.

The resolution in the location loss is set by the known signal center.
Define $b(\theta)=1-|\langle\phi(0),\phi(\theta)\rangle|^2$ and
\begin{equation}
 b_n(s)=b\bigl(s/\sqrt{nq_0}\bigr),\qquad \tau_n(s)=nb_n(s).
 \label{qcp-M-eq:distinguishability}
\end{equation}
Position error is scored by clipped absolute loss in resolution units,
and amplitude offset error by squared loss in fixed units:
\begin{equation}
 \begin{aligned}
 \ell_{\loc,n}(x,u)
   &=2\min\{1,\tau_n(s)|x-u_n|/K\},\\
 \ell_{\loc,\infty}(x,u)
   &=2\min\{1,s^2|x-u|/K\},\\
 \ell_{\amp}(y,z)&=(y-z)^2,\quad y\in[-L,L].
 \end{aligned}
 \label{qcp-M-eq:losses}
\end{equation}
Amplitude reports may be clipped to $[-L,L]$ without increasing the
loss. The bounds for location and amplitude are $B_\loc=2$ and
$B_\amp=4L^2$, respectively.
At $s=s_{\mathrm c}$, the coherent location loss becomes $2|x-u|$.
For finite records, $b_n(s)|\widehat k-k_n|$ expresses position error in the same
resolution units, whereas the physical parameter estimate $\widehat\theta=(s+ay)/\sqrt{nq_0}$
has squared error $a^2(y-z)^2/(nq_0)$ under a physical prior of half width
$aL/\sqrt{nq_0}$.

We compare prior widths while keeping $a,s,q_0,K$ and the loss units fixed. Along physical recovery paths, we take $s_n\to\infty$ with $s_n=o(\sqrt n)$ and keep the state curve, offset prior and amplitude units fixed. The physical displacement $\theta$ then vanishes at each site, while the prior spread in $z$ remains fixed and position errors are scored in resolution units.

\subsection{Bayesian risks and the cost of a common readout}

For either experiment, write $\omega_{u,z}$ for the input state and $\ell_j$
for the loss associated with $j\in\{\loc,\amp\}$. The Bayes risk of a paired
POVM, averaged over the joint prior in \eqref{qcp-M-eq:prior}, is
\[
 r_j(M)=\mathbb E_{u,z}\!\int_{\mathcal X\times[-L,L]}
 \ell_j\,\tr[\omega_{u,z}M(dx,dy)].
\]
Here $\ell_\loc$ depends on $(x,u)$ and $\ell_\amp$ on $(y,z)$.
The minimum Bayes risk for task $r_j^\star=\inf_M r_j(M)$ optimizes task $j$
over the full record and joint prior, averaging over unknown $z$ for
position and unknown $u$ for amplitude. Appending a fixed report for
the other task embeds each single task measurement in the paired class,
so all three optimizations use the same complete measurements.

With $e_j(M)=r_j(M)-r_j^\star$ denoting the marginal excess risk, the joint optimum
and cost of a common readout at fixed weight $0\le\lambda\le1$ are
\begin{equation}
 \begin{aligned}
 R_\lambda^\star&=\inf_M\{\lambda r_\loc(M)+(1-\lambda)r_\amp(M)\},
 \\
 \G_\lambda&=R_\lambda^\star-\lambda r_\loc^\star-(1-\lambda)r_\amp^\star
 \\
 &=\inf_M\{\lambda e_\loc(M)+(1-\lambda)e_\amp(M)\}.
 \end{aligned}
 \label{qcp-M-eq:gap}
\end{equation}
The cost $\G_\lambda$ is the additional Bayes risk incurred by serving
both tasks from the same record. It is nonnegative and satisfies
$\G_0=\G_1=0$. For the weights $\alpha+\beta>0$,
the corresponding unnormalized cost is $(\alpha+\beta)\G_{\alpha/(\alpha+\beta)}$.


An exact operator characterization of the gap is obtained from
$\Omega=\mathbb E_{u,z}\omega_{u,z}$ and $\Omega_1=\mathbb E_{u,z}[z\omega_{u,z}]$, both averaged under the full joint prior.
The amplitude observable thus incorporates uncertainty in position.
All operators in the following statement act on the average state
support $\mathcal H_\Omega=\operatorname{supp}\Omega$.

\begin{qcpMstructure}[Exact common optimality]
\label{qcp-M-prop:common-optimality}
Assume that the report spaces are compact and the losses are bounded
and continuous in the reports.
For squared amplitude loss, the spectral measurement of the Personick
observable $A$ minimizes the Bayes risk
\cite{qcp-M-personick1971}.
On $\mathcal H_\Omega$, $A$ is the unique bounded self-adjoint solution of
\begin{equation}
 \tfrac12(\Omega A+A\Omega)=\Omega_1,\qquad \|A\|\le L.
 \label{qcp-M-eq:amplitude-personick}
\end{equation}
For every $0<\lambda<1$, the gap $\G_\lambda$ is zero if and only if
there exists a POVM $N$ that is optimal for position and satisfies
\begin{equation}
 [N(B),A]=0,
 \label{qcp-M-eq:common-optimality}
\end{equation}
for every Borel set $B\subseteq\mathcal X$.
\end{qcpMstructure}

Every amplitude optimal readout has the spectral measure of $A$
as its amplitude marginal. A position POVM is jointly measurable
with this sharp marginal if and only if all its effects commute
with $A$ \cite{qcp-M-heinosaari2008}.
The gap vanishes at one interior task weight if and only if it
vanishes at all.
The \hyperref[qcp-SM-sec:common-optimality]{supplementary proof}
establishes that the optimal risks are attained.

The prior and loss shape the optimal decisions, connecting the
criterion to Bayesian measurement incompatibility
\cite{qcp-M-albarelli2026}.

In the smooth onset models of Section~\ref{qcp-M-sec:joint-rise}, the position parameter shifts a response profile whose shape and width are known. The prior and paired
measurement class are unchanged, while position remains a continuous
target even in finite records. Section~\ref{qcp-M-sec:proof-route}
provides the physical comparison and its resolution units.
Throughout, $\dagger$ denotes an adjoint and $\star$ an optimum.
In comparisons of physical and coherent records, the subscript $n$
denotes a finite physical record; an omitted record size subscript
denotes the coherent experiment. For the ideal step, we write
$\G_\lambda(s;L)=\G_\lambda(s,0;L)$.
Known design arguments are suppressed within a fixed experiment.

\section{Main results}
\label{qcp-M-sec:main-results}
Figures~\ref{qcp-M-fig:spatial-allocation}--\ref{qcp-M-fig:decision-loss} summarize the spatial allocation of information, the two limiting procedures, and the role of decision loss.

\subsection{Ideal steps: conflict and spatial recovery}
\label{qcp-M-sec:ideal-steps}
An ideal step concentrates position information near the boundary and
amplitude information across the suffix. This spatial separation gives
strong signal recovery for fixed priors with finite strength conflict.

\begin{qcpMtheorem}[Conflict and recovery]
\label{qcp-M-thm:main}
Consider the prior and calibrated losses in
\eqref{qcp-M-eq:prior} and \eqref{qcp-M-eq:losses}.
\begin{enumerate}
\renewcommand{\labelenumi}{(\roman{enumi})}
\item At $s=s_{\mathrm c}$, there are constants $c,\delta>0$
such that, for every $0<\lambda<1$, the coherent ideal step experiment satisfies
\[
 \G_\lambda(s_{\mathrm c};L)
 >c\min\{\lambda,1-\lambda\},
 \qquad
 |L-\tfrac12|\le\delta.
\]
The gap is also positive throughout this interval in all sufficiently long physical records.
\item For fixed $0<L<\infty$ and 
$s_n\to\infty$ with $s_n=o(\sqrt n)$, there are complete common
measurements $M_n$, such that
\[
 \max_{j} e_{j,n}(M_n)\longrightarrow0.
\]
Thus $\G_{\lambda,n}(L)\to0$.
Both marginal optimal risks and their improvements over the prior alone remain bounded away from zero.
\end{enumerate}
\end{qcpMtheorem}

The offset prior remains fixed in calibrated units, while the known
center sets the position resolution. Conflict and recovery hold under
the same prior and loss rule at each fixed prior width in the interval
in part (i).

For $u\le v$, the coherent fields satisfy
\begin{equation}
 \begin{aligned}
 \|f_{u,h}-f_{v,g}\|_2^2
 &=h^2(v-u)\\
 &\quad+(h-g)^2(1-v).
 \end{aligned}
 \label{qcp-M-eq:geometry}
\end{equation}
The first term comes from the interval between the candidate boundaries, and the second from their shared suffix.
A complete prelocalization instrument uses the observed coarse label to retain
a quantum boundary block for the position decision and the remaining
suffix for amplitude (Figure~\ref{qcp-M-fig:spatial-allocation}).
The block size balances coarse localization error against the
comparison error from replacing the block amplitude by the known center.
The risk comparison covers all localization outcomes and the full
position prior.

Supplementary Material Sections T1, O1 and D1 establish conflict
and its robustness in prior width. Sections A2 and A3 give the
recovery proof and the marginal risk bounds, respectively.

\begin{figure}[t]
\centering
\includegraphics[width=\linewidth]{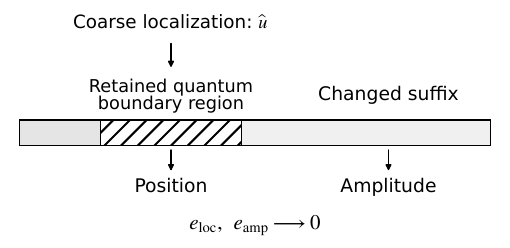}
\caption[Spatial allocation and recovery]{Spatial allocation and recovery.
The observed coarse localization label $\widehat u$ guides the allocation:
the retained quantum boundary region serves position estimation, and
the remaining changed suffix serves amplitude estimation.
For nondecreasing $C^2$ profiles,
both excess risks vanish in the coherent experiment as $s\to\infty$
and $w\to0$.
Risks are averaged over all localization outcomes.}
\label{qcp-M-fig:spatial-allocation}
\end{figure}

\subsection{Shrinking rises: leading recovery and inverse square cost}
\label{qcp-M-sec:joint-rise}
For a smooth onset, the overlap between location and amplitude
information is controlled by the rise width. Let the known profile
$F\in C^2(\mathbb R)$ be nondecreasing, equal to zero for $x\le0$ and to one for $x\ge1$. The coherent field is $f_{u,h}^{(w)}(t)=hF((t-u)/w)$,
evaluated at record coordinate $t\in[0,1]$. Let $\chi=s^2w$ be the relative rise width.
The coherent state and location resolution are
\begin{equation}
 \begin{aligned}
 \sigma_{s,w}(u,z)
   &=\proj{f_{u,h}^{(w)}}_{\mathrm{coh}},\\
 \delta_{s,w}
   &=\frac{\sqrt{1+\chi}}{s^2}.
 \end{aligned}
 \label{qcp-M-eq:rise-family}
\end{equation}
$\chi$ compares the rise width with the ideal step resolution. The location scale $\delta_{s,w}$ changes from $s^{-2}$ to $\sqrt{w}/s$ as the rise becomes resolved. The signal center $s$, rise width $w$ and profile $F$ are known.
The location loss is defined by
\begin{equation}
 \ell_\loc^{(w)}(x,u)=2\min\{1,|x-u|/(K\delta_{s,w})\}.
 \label{qcp-M-eq:rise-loss}
\end{equation}
The amplitude loss remains $(y-z)^2$. At $w=0$, the convention
$f_{u,h}^{(0)}=f_{u,h}$ recovers the coherent ideal step location loss.
For a smooth rise, position remains a continuous target in both the
coherent experiment and finite records.

\begin{qcpMtheorem}[Leading recovery and inverse square cost]
\label{qcp-M-thm:joint-scale}
\label{qcp-M-thm:uniform-first-cost}
\leavevmode\par
\begin{enumerate}
\renewcommand{\labelenumi}{(\roman{enumi})}
\item As $s\to\infty$ and $w\to0$, there are
complete common readouts whose two marginal excess risks vanish.
Consequently, $\G_\lambda(s,w)\to0$ for $0\le\lambda\le1$, with arbitrary $\chi$. Each marginal optimal risk and its
reduction from the prior only risk stay bounded away from zero.

\item Take $F(x)=10x^3-15x^4+6x^5$ on $[0,1]$ and $L=1/2$. There exist constants $0<C_-\le C_+<\infty$ and
$s_0>0$ such that, for all $s\ge s_0$ and $\chi\in[1/2,2]$,
\begin{equation}
 C_-\le s^2\G_\lambda(s,\chi/s^2)\le C_+.
 \label{qcp-M-eq:uniform-first-cost}
\end{equation}
\end{enumerate}
\end{qcpMtheorem}

Coarse localization guides the spatial allocation of the retained field. A shrinking quantum boundary region serves position estimation, and the remaining suffix serves amplitude estimation. Supplementary Material
Section J compares this labeled
product experiment with the original record in both directions,
including late changes, and tracks each marginal optimum separately.
At inverse square precision, optimal localization constrains a mode
used by the amplitude decision. Its linear response gives the residual
cost in the gray region $\chi\in[1/2,2]$ of
Figure~\ref{qcp-M-fig:limit-paths}.
Matching lower and upper bounds establish the inverse square law under the full prior. Comparisons accurate to $o(s^{-2})$ for the joint optimum and both marginal optima preserve the leading residual cost.

\begin{figure}[!b]
\centering
\includegraphics[width=\linewidth]{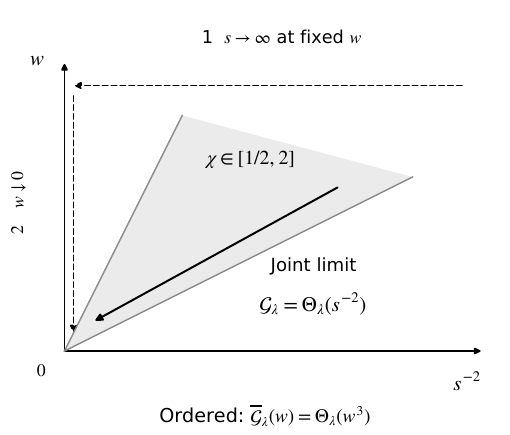}
\caption[Two distinct limiting paths]{Schematic limiting paths for the coherent quintic experiment with $L=1/2$ and fixed $0<\lambda<1$.
In the gray region $\chi\in[1/2,2]$,
$\G_\lambda=\Theta_\lambda(s^{-2})$ uniformly in $\chi$.
The dashed ordered path first takes $s\to\infty$ at fixed $w$,
then $w\downarrow0$, giving
$\overline{\G}_\lambda(w)=\Theta_\lambda(w^3)$.}
\label{qcp-M-fig:limit-paths}
\end{figure}

\subsection{Resolved rises}
\label{qcp-M-sec:resolved-rises}

\begin{qcpMtheorem}[Cubic width cost]
\label{qcp-M-thm:cubic-width}
For the quintic profile and $L=1/2$, $0<\lambda<1$. The strong signal limit exists.
As $w\downarrow0$,
\begin{equation}
 \overline{\G}_\lambda(w)
 :=\lim_{s\to\infty}\G_\lambda(s,w)
 =\Theta_\lambda(w^3).
 \label{qcp-M-eq:main-cubic-width}
\end{equation}
\end{qcpMtheorem}

The loss dependence can be seen in the resolved local experiment
(Figure~\ref{qcp-M-fig:decision-loss}).
At $u<1$, express the position displacement $v$ in units
of $\sqrt w/s$. The local state is
\begin{equation}
 \ket{\kappa(u)v+d(u)z}\otimes\ket{c(u)z},
 \qquad \kappa(u),c(u)>0.
 \label{qcp-M-eq:main-resolved-fiber}
\end{equation}
The first mode has position sensitivity $\kappa$ and signed shared
amplitude displacement $d$; the second has residual amplitude
sensitivity $c$.

For quadratic local position loss, the residual mode provides an amplitude estimate that can be used to correct the first mode displacement. Let $A$ be the optimal amplitude observable for $\ket{cz}$. Measuring the first mode quadrature and $A$ gives outcomes $X$ and $Y$, with  $\widehat v=(X-dY)/\kappa$ and $\widehat z=Y$. This readout attains both marginal optima within the complete covariant class. 

For clipped absolute position loss, the optimal position readout
requires a leading cubic correction to linear prediction,
involving $d^3C_3$.
For fixed sufficiently small $c>0$, $[C_3,A]\ne0$, and the squared
correction gives a local cost of order $d^6$ as $d\to0$.
For complete quintic rises, $d^2=(7\log2/10)w$.
A complete upper readout controls truncated rises.
Comparison with the original coherent experiment and averaging
the local bounds over the full position prior then yield the
cubic width law.

Supplementary Material Section S gives the local model, the limit
at fixed width and the quadratic comparison.
Sections P and B establish the cubic response and positive local
lower bound; Section U proves the full prior width law.

\begin{figure*}[!t]
\centering
\includegraphics[width=\linewidth]{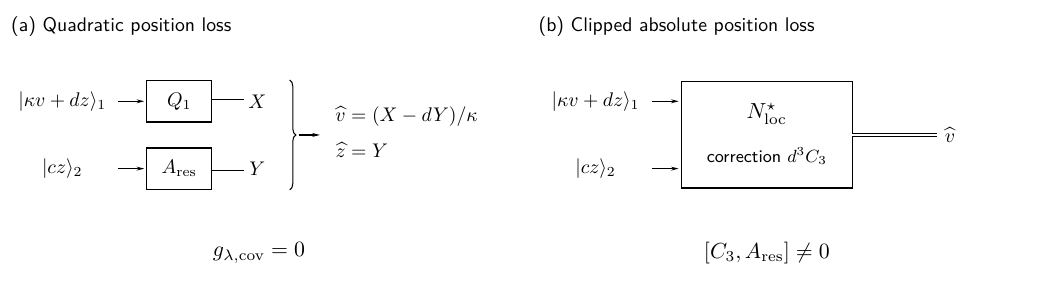}
\caption[Position loss and local readouts]{Position loss selects the local quantum readout.
The panels compare readouts for the same two-mode family, with complete
covariant measurements and quadratic amplitude loss.
(a) For quadratic position loss, measuring $Q_1$ and the optimal
residual amplitude observable $A_{\rm res}=A$ gives outcomes $X,Y$.
The reports $\widehat v=(X-dY)/\kappa$ and $\widehat z=Y$
attain both marginal optima within this class. (b) For the quintic local model with $L=1/2$, the position-optimal
POVM $N_\loc^\star$ for clipped absolute loss
requires a leading correction involving $d^3C_3$.
For sufficiently small $c>0$,
$[C_3,A_{\rm res}]\ne0$ produces a local cost of order $d^6$ as $d\to0$.
For complete quintic rises, $d^2=(7\log2/10)w$.
Averaging the transferred local bounds over the full position prior
yields the cubic width law.}
\label{qcp-M-fig:decision-loss}
\end{figure*}

\section{Numerical certification and transfer to physical records}
\label{qcp-M-sec:proof-route}
\label{qcp-M-sec:certified-evidence}
\subsection{Numerical certificates of conflict}
Two certificates establish finite strength conflict in the coherent
ideal step experiment and a physical phase qubit record of $16$ qubits. Each compares
the joint optimum with both marginal optima under unknown
position and amplitude.

At $s=s_{\mathrm c}$ and $L=1/2$, the coherent ideal step certificate for
equal weights yields
\begin{equation}
 1\times10^{-3}<\G_{1/2}<1.83\times10^{-3}.
 \label{qcp-M-eq:critical-gap-certificate}
\end{equation}
Continuity in prior width extends this bound to the conflict interval
in Theorem~\ref{qcp-M-thm:main}. For the original $n=16$ phase
qubit record at the same center, certification over the entire known
interval $L\in[7/16,9/16]$ gives
\begin{equation}
 4.17\times10^{-4}<\G_{1/2,16}(L)<2.242\times10^{-3}.
 \label{qcp-M-eq:main-finite-width-anchor}
\end{equation}
Relative to the equally weighted average of the actual marginal optima,
the gap lies between $0.58\%$ and $2.76\%$. Both tasks remain
informative throughout this interval, over which calibration and raw
loss units are held fixed.
Figure~\ref{qcp-M-fig:finite-certificate} displays deterministic certified
enclosures over eight prior width cells and tighter bounds at nine nodes.

\begin{figure}[t]
\centering
\includegraphics[width=\linewidth]{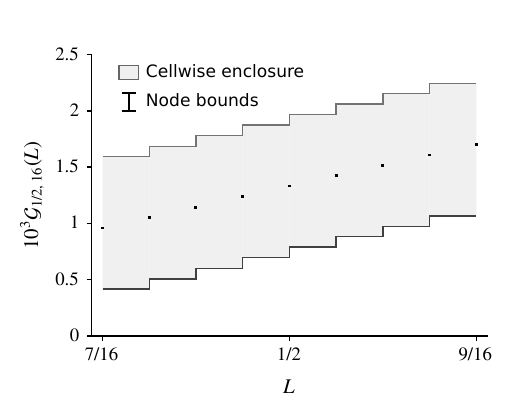}
\caption[Conflict in a finite physical record]{Certified conflict for
an ideal step in a phase record of $16$ qubits at $s=s_{\mathrm c}$
and $\lambda=1/2$.
The offset prior half width $L$ varies at fixed calibration
and raw loss units.
The gray band encloses $10^3\G_{1/2,16}(L)$ over eight width cells;
black capped segments give tighter bounds at nine nodes.
Both are deterministic.
The cost is between $0.58\%$ and $2.76\%$ of $(r_{\loc,16}^\star+r_{\amp,16}^\star)/2$
throughout $L\in[7/16,9/16]$.
Tables and error budgets appear in Supplementary Material Section~V.}
\label{qcp-M-fig:finite-certificate}
\label{qcp-V-fig:certificate-band}
\end{figure}

Complete POVMs and feasible decision duals give upper and lower
risk bounds, respectively \cite{qcp-M-yuen1975}.
Continuous bounds in prior width cover the intervals between nodes.
Supplementary Material Sections T1, O1, D1 and V give the constructions
and error budgets.

\subsection{Transfer to physical records}
Comparisons with physical records resolve each conclusion at its risk
scale. Table~\ref{qcp-M-tab:physical-paths} lists sufficient paths for recovery
and inverse square cost in physical records.
For smooth rises, the site parameters are
$\theta_i=(s_n+az)F((i/n-u)/w_n)/\sqrt{nq_0}$.
The target $u$ remains continuous, and reports lie in $\mathcal X_n$.
With $\delta_n=\sqrt{1+w_n\tau_n(s_n)}/\tau_n(s_n)$,
the physical position loss is
\begin{equation}
 \ell_{\loc,n}^{(w_n)}(x,u)
 =2\min\{1,|x-u|/(K\delta_n)\}.
 \label{qcp-M-eq:physical-rise-loss}
\end{equation}
The amplitude loss remains $(y-z)^2$.
Comparisons in both directions transfer the joint optimum and both
marginal optima.
Leading recovery uses a vanishing risk error, the inverse square cost
uses $o(s_n^{-2})$ precision. Supplementary Material Sections A2 and C11 establish these paths.

\begin{table*}[t]\centering\small
\begin{tabular}{p{.28\linewidth}p{.25\linewidth}p{.38\linewidth}}
\toprule
Conclusion & Sufficient signal path & Known design and comparison\\
\midrule
Ideal step recovery & $s_n\to\infty$, $s_n=o(\sqrt n)$ &
Fixed finite $L$; direct construction with a retained boundary block.\\[3pt]
Leading recovery for shrinking rises & $s_n\to\infty$, $s_n=o(n^{1/4})$ &
Known $F$, $w_n\to0$; risk errors $o(1)$.\\[3pt]
Inverse square crossover cost & $s_n\to\infty$, $s_n=o(n^{1/8})$ &
Quintic, $L=1/2$, $\chi_n\in[1/2,2]$, $w_n=\chi_n/s_n^2$;
risk errors $o(s_n^{-2})$.\\
\bottomrule
\end{tabular}
\caption{Sufficient paths to finite product records. The comparisons
preserve the unknown parameters, calibrated losses and complete
readouts to the precision required by each conclusion.}
\label{qcp-M-tab:physical-paths}
\end{table*}

\section{Discussion}\label{qcp-M-app:limit-scope}

The common readout cost quantifies how onset geometry and decision loss shape simultaneous position and amplitude estimation. The joint and marginal optimizations use the full record and the same joint prior, with both parameters unknown. For ideal steps, we certify a positive cost at a reference signal strength and recover both marginal optima at strong signal under the same prior. These results show how finite signal conflict and asymptotic compatibility arise within the same estimation task.

Recovery exploits the spatial distribution of information in the record. Coarse localization preserves a quantum boundary region for the position decision and leaves the remaining suffix for amplitude estimation. Known nondecreasing $C^2$ profiles also admit leading recovery in the coherent experiment. This holds along every path $s\to\infty$, $w\to0$, with arbitrary variation of the relative width $\chi=s^2w$. Both marginal optimal risks and their improvements over the prior alone remain bounded away from zero. The common readout, therefore approaches two informative marginal optima at their respective loss scales.
The residual laws quantify the coupling that remains beyond leading recovery. For the coherent quintic experiment with $L=1/2$, two limiting procedures give distinct scaling laws. The cost is $\Theta_\lambda(s^{-2})$. Taking $s\to\infty$ at fixed $w$ before letting $w\downarrow0$ gives $\overline{\G}_\lambda(w)=\Theta_\lambda(w^3)$. These laws describe how the relation between rise width and position resolution controls the decay of the common readout cost.
The resolved local experiment explains the role of the loss function. With quadratic position and amplitude losses, amplitude prediction from the residual mode supplies a linear correction for position. The resulting complete readout attains both marginal optima within the local covariant class. With clipped absolute position loss, a cubic correction produces a local cost of order $d^6$ at each fixed sufficiently small positive residual sensitivity. Transferring the local bounds to the full position prior yields the cubic width law. The loss thus determines how shared information enters the optimal decisions and their residual incompatibility.

Proposition~\ref{qcp-M-prop:common-optimality} gives an exact criterion:
compatibility holds if and only if some optimal position POVM commutes
with the optimal amplitude observable. Leading recovery describes the approach to the marginal optima, and the residual laws quantify its rate. Certified bounds establish positive conflict in finite phase records (Figure~\ref{qcp-M-fig:finite-certificate}). Comparisons at the required risk precision transfer recovery and the inverse square crossover cost to physical records along Table~\ref{qcp-M-tab:physical-paths}. The spatial allocation of information and the decisions selected by the loss connect quantum change point metrology with Bayesian measurement compatibility.

Future work can examine how the onset profile and position loss determine the first correction to linear amplitude prediction. Comparing these corrections across models would clarify which features set the residual exponent and which scaling laws persist across different designs. Treating rise width or profile shape as additional unknowns would reveal how uncertainty in the onset affects spatial recovery. Finite record implementations could also assess the tradeoff between estimation precision and the quantum memory required to retain the boundary region.

\begin{samepage}
\section*{Author contributions}

Xu Chen conceived the initial idea, developed the physical model, and carried out the formal analysis. Xue Ma performed the numerical computations and prepared the data visualizations. Both authors contributed writing and revision of the manuscript. LLM Models were used for language polishing and cross checking numerical results. Both authors reviewed the  outputs and take full responsibility for the scientific content and the final manuscript.

\section*{Acknowledgements}

This work was supported by Hebei University of Science and Technology under
Starting Grant No.~81/1181298.

\end{samepage}

\FloatBarrier

\clearpage
\endgroup

\newcounter{qcpMainReferenceCount}
\setcounter{qcpMainReferenceCount}{\value{NAT@ctr}}
\clearpage
\onecolumn
\newgeometry{left=1in,top=1in,textwidth=6.5in,textheight=9in}
\makeatletter
\def\toclevel@section{2}
\def\toclevel@subsection{3}
\def\toclevel@subsubsection{4}
\makeatother
\pdfbookmark[0]{Supplementary Material}{qcp-SM-bookmark}
\qcpvolumetitle{Supplementary Material}
\label{qcp-SM-start}
This supplement supplies the proofs and numerical error budgets for the
main article. The following map identifies the argument supporting each
main conclusion.
\begin{center}\small
\begin{tabular}{p{.31\linewidth}p{.61\linewidth}}
\toprule
Main-text conclusion & Supporting sections\\
\midrule
Proposition~\ref{qcp-M-prop:common-optimality} &
The proof immediately below: Personick optimum, attainment and commuting effects.\\[3pt]
Theorem~\ref{qcp-M-thm:main}(i), coherent certificate &
T1 and O1: complete upper rules and feasible duals; D1: robustness in prior width and physical return.\\[3pt]
Theorem~\ref{qcp-M-thm:main}(ii) &
A1--A3: risk comparison, retained boundary construction and informative marginals.\\[3pt]
Theorem~\ref{qcp-M-thm:joint-scale}(i) &
J1: two-way comparison with a labeled product experiment, including late changes.\\[3pt]
Theorem~\ref{qcp-M-thm:joint-scale}(ii), Table~\ref{qcp-M-tab:physical-paths} &
C1--C10: uniform inverse square bounds; C11: comparison at the required physical risk scale.\\[3pt]
Theorem~\ref{qcp-M-thm:cubic-width}, Figure~\ref{qcp-M-fig:decision-loss} &
S: fixed width limit and quadratic compatibility; P and B: cubic response and positive local lower bound; U: complete upper rule and full prior integration.\\[3pt]
Finite certificate, Figure~\ref{qcp-M-fig:finite-certificate} &
V: the 16-qubit experiment, node risks and continuous width-cell enclosures.\\
\bottomrule
\end{tabular}
\end{center}

\paragraph{Conventions.}
The priors, calibrated losses, reports and state curve are those of the
main text. Trace distance is half the trace norm and fidelity is squared.
For a synthesis $T$, let $\Omega=TT^\dagger$ and
$\|D\|_\Omega=\|D\Omega^{1/2}\|_{\HS}$; finite unnormalized masses
$\tr\Omega$ are retained. A Naimark isometry is denoted by $J$,
with commuting report observables $X,Y$.
Strictly positive gap constants concern a fixed $0<\lambda<1$;
the endpoint weights have zero gap.
The coherent field amplitude is $h=s+az$, the physical curve parameter
is $\theta$, and $t$ is a record or integration coordinate.

For the residual ensemble $\ket{cz}$ with $z\sim\Unif[-L,L]$, write its
amplitude optimum as $\mathcal A_L(c)$ and its suffix average as
$\overline{\mathcal A}_L=\int_0^1\mathcal A_L(a\sqrt e)\,de$.
The local covariant gap is $g_\lambda(\kappa,d,c)$;
$g_{\lambda,w}(u)$ uses the actual sensitivities, including truncated rises.
For complete quintic rises, $\kappa_F=\sqrt{10/7}$ and
$g_\lambda(d;c)=g_\lambda(\kappa_F,d,c)$.
The full coherent gap is $\G_\lambda(s,w)$, the physical gap is
$\G_{\lambda,n}$, and
$\overline{\G}_\lambda(w)=\lim_{s\to\infty}\G_\lambda(s,w)$.

\pdfbookmark[1]{Proof of Proposition 1}{qcp-SM-common}
\section*{Proof of exact common optimality}
\label{qcp-SM-sec:common-optimality}
Let $\omega_{u,z}$ be trace norm measurable input states on a separable Hilbert space,
with the joint prior of Proposition~\ref{qcp-M-prop:common-optimality}.
The position report space $\mathcal X$ is compact metric. The location
loss is measurable, nonnegative, uniformly bounded and continuous in
the report for almost every parameter. Amplitude reports lie in $[-L,L]$.
Define
\[
 \Omega_j=\mathbb E[z^j\omega_{u,z}],\quad j=0,1,2,
 \qquad \Omega=\Omega_0,\qquad
 C_x=\mathbb E[\ell_\loc(x,u)\omega_{u,z}].
\]
Dominated convergence makes $C_x$ continuous in trace norm.
Compression to $\mathcal H_\Omega=\operatorname{supp}\Omega$ preserves
all risks. On this support, $\Omega$ is faithful and trace class.

\begin{proof}[Proof of Proposition~\ref{qcp-M-prop:common-optimality}]
\textit{Amplitude optimum.}
Since $-L\Omega\preceq\Omega_1\preceq L\Omega$, the operators
\[
 A_R=2\int_0^R e^{-t\Omega}\Omega_1e^{-t\Omega}\,dt
 \quad\hbox{satisfy}\quad
 -L(I-e^{-2R\Omega})\preceq A_R\preceq L(I-e^{-2R\Omega}).
\]
The integrals with positive source $\Omega_1+L\Omega$ increase to a
bounded operator. As $e^{-2R\Omega}\to0$ strongly, $A_R$ converges
strongly to a self-adjoint $A$ with $\|A\|\le L$. Integration gives
\[
 \tfrac12(\Omega A_R+A_R\Omega)
 =\Omega_1-e^{-R\Omega}\Omega_1e^{-R\Omega}.
\]
The remainder tends to zero in trace norm, proving the Personick
equation. In an eigenbasis $\Omega e_i=\omega_i e_i$, $\omega_i>0$,
its solution has entries $A_{ij}=2(\Omega_1)_{ij}/(\omega_i+\omega_j)$,
which proves uniqueness among bounded operators.

Dilate a paired POVM as $M(Q)=J^\dagger E(Q)J$ and write
$Y=\int y\,E(dx,dy)$. Completing the amplitude square yields
\[
 r_\amp^\star=\tr\Omega_2-\tr(\Omega A^2),\qquad
 e_\amp(M)=\|YJ-JA\|_\Omega^2.
\]
The spectral measurement of $A$ makes this square zero. Its outcomes
lie in $[-L,L]$, so it is an admissible amplitude optimum.

\par\medskip\noindent\textit{Attainment.}
For a compact metric report space $\mathcal K$, complete POVMs are
positive unital maps $\Phi:C(\mathcal K)\to\mathcal B(\mathcal H_\Omega)$.
The bounds $\|\Phi(f)\|\le\|f\|_\infty$ give compactness in the
pointwise ultraweak topology. Positivity and unitality pass to the
limit, and scalar Riesz representation recovers a POVM.

Choose finite-rank projections $P_N\uparrow I$, put
$m=\tr\Omega$ and $\eta_N=\tr[(I-P_N)\Omega]$, and let $r_j^{(N)}$
use the compressed states $P_N\omega_{u,z}P_N$. The projection
trace-distance bound and Cauchy--Schwarz give, uniformly in $M$,
\[
 |r_j(M)-r_j^{(N)}(M)|
 \le B_j\mathbb E\|\omega_{u,z}-P_N\omega_{u,z}P_N\|_1
 \le2B_j\sqrt{m\eta_N}\longrightarrow0.
\]
Each compressed cost field is a finite sum of fixed matrix units with
continuous coefficients in the report. Its risk is therefore
continuous in the preceding topology. The full risk is a uniform
limit of these continuous functions. Compactness gives minimizers for
the two marginal risks and every weighted joint risk.

\par\medskip\noindent\textit{Commuting position effects.}
For $0<\lambda<1$, an attained joint optimum with $\G_\lambda=0$
has $e_\loc=e_\amp=0$. The amplitude square and the dense range of
$\Omega^{1/2}$ imply $YJ=JA$. With
$E_\loc(B)=E(B\times[-L,L])$ and $N(B)=J^\dagger E_\loc(B)J$,
\[
 [N(B),A]=J^\dagger[E_\loc(B),Y]J=0.
\]
The marginal $N$ is optimal for position. Functional calculus also
gives $E_Y(D)J=JE_A(D)$, so the amplitude marginal is the spectral
measure $E_A$ of $A$.

Conversely, a position optimum $N$ commuting with $A$ commutes
with its spectral measure $E_A$. On finite report partitions,
$N(B_i)E_A(D_j)$ are positive complete effects. Refining the partitions
and applying the preceding compactness argument gives a paired POVM
with marginals $N,E_A$. Its rectangle effects are
$M(B\times D)=N(B)E_A(D)$, since the second marginal is sharp
\cite[Proposition~8]{qcp-M-heinosaari2008}.
Both marginal risks are therefore optimal at every weight.
A fixed paired report on $\mathcal H_\Omega^\perp$ extends the rule
to the original Hilbert space.
\end{proof}

\pdfbookmark[1]{I. Ideal step conflict and recovery}{qcp-SM-part1}
\section*{I. Ideal step conflict, recovery and physical comparison}
\begingroup
\qcpresetvolume{C}
\qcparticlelayout
\setlength{\emergencystretch}{3em}

\numberwithin{equation}{section}

\newcommand{\LocLower}{0.124616894327471}
\newcommand{\LocUpper}{0.125468298728954}
\newcommand{\AmpLower}{0.060612066766632}
\newcommand{\AmpUpper}{0.060807446886014}
\newcommand{\JointLower}{0.094140512957477936}
\newcommand{\JointUpper}{0.094449096017483}
\newcommand{\GapLower}{0.001002640149993936}
\newcommand{\GapUpper}{0.001834615470431500}

\renewcommand{\thesection}{A\arabic{section}}\setcounter{section}{0}
\section{Risk transfer and finite records}
\begin{qcpClemma}[Comparison of the three risks]\label{qcp-C-lem:transfer}
Let $\mathsf E,\mathsf F$ share the same prior, targets, actions and losses
$0\le\ell_j\le B_j$. Let $\mathcal T:\mathsf E\to\mathsf F$ and
$\mathcal S:\mathsf F\to\mathsf E$ be channels independent of the parameters
with trace distance errors, uniform or averaged under the prior, $\varepsilon_T,\varepsilon_S$.
For $B_\lambda=\lambda B_\loc+(1-\lambda)B_\amp$, the marginal risks,
weighted joint risk and gap satisfy
\begin{align}
 -B_j\varepsilon_S&\le r_j^\star(\mathsf E)-r_j^\star(\mathsf F)\le B_j\varepsilon_T,\nonumber\\
 -B_\lambda\varepsilon_S&\le R_\lambda^\star(\mathsf E)-R_\lambda^\star(\mathsf F)\le B_\lambda\varepsilon_T,\label{qcp-C-eq:transfer}\\
 |\G_\lambda(\mathsf E)-\G_\lambda(\mathsf F)|&\le B_\lambda(\varepsilon_T+\varepsilon_S).\nonumber
\end{align}
A rule of excess $\delta_j$ pulls back through either channel with excess
at most $\delta_j+B_j(\varepsilon_T+\varepsilon_S)$.
Loss changes bounded by $d_j$ change the gap by at most
$2[\lambda d_\loc+(1-\lambda)d_\amp]$. An ordered change to the position loss
alone changes it by at most $\lambda d_\loc$.
\end{qcpClemma}
\begin{proof}
Pullback through $\mathcal T$ gives
$r_j^\star(\mathsf E)\le r_j^\star(\mathsf F)+B_j\varepsilon_T$.
The reverse channel gives the other inequality. The joint and weighted
marginal increments lie in one interval of length
$B_\lambda(\varepsilon_T+\varepsilon_S)$, which bounds their difference.
Comparison with the corresponding infimum gives the excess bound.

A loss change bounds each infimum by its uniform norm. For an ordered
location change, both increments lie in $[0,\lambda d_\loc]$.
\end{proof}

\begin{qcpCtheorem}[Compact coherent comparison]\label{qcp-C-thm:channels}
Let $M_j=\sup_{|\theta|\le r_0}\|\phi^{(j)}(\theta)\|$,
$\gamma_0=-i\langle\phi(0),\phi'(0)\rangle$,
$B_2=M_2+2|\gamma_0|M_1+\gamma_0^2$ and
$C_\phi=(B_2/q_0+\sqrt3)/2$. For finite $H\ge0$, put
\begin{equation}
 C_H=H+C_\phi H^2,\qquad
 \varepsilon_{n,H}=\min\{1,C_H/\sqrt n\}.\label{qcp-C-eq:critical-error}
\end{equation}
For $H/\sqrt{nq_0}\le r_0$, CPTP channels compare the finite and coherent
records in both directions with error at most $\varepsilon_{n,H}$, uniformly
in $u\in[0,1]$, $|h|\le H$. The channels depend only on $n$ and the known curve.
For $|s|+aL\le H$, set $d_{\loc,n}=2(|\tau_n-s^2|+\tau_n/n)/K$.
The corresponding marginal risks, joint risk and gap obey
\begin{align}
 |r_{\loc,n}^\star-r_{\loc,\infty}^\star|&\le2\varepsilon_{n,H}+d_{\loc,n},&
 |r_{\amp,n}^\star-r_{\amp,\infty}^\star|&\le4L^2\varepsilon_{n,H},\nonumber\\
 |R_{\lambda,n}^\star-R_{\lambda,\infty}^\star|&\le B_\lambda\varepsilon_{n,H}+\lambda d_{\loc,n},&
 |\G_{\lambda,n}-\G_{\lambda}|&\le2B_\lambda\varepsilon_{n,H}+2\lambda d_{\loc,n},\label{qcp-C-eq:finite-return}
\end{align}
where $B_\lambda=2\lambda+4L^2(1-\lambda)$.
\end{qcpCtheorem}
\begin{proof}
Use the gauge $\psi(\theta)=e^{-i\gamma_0\theta}\phi(\theta)$ and a basis
starting with $e_0=\phi(0)$, $e_1=\psi'(0)/\sqrt{q_0}$. Mapping $e_0$ to
vacuum and the transverse vectors to single photons of distinct colors
defines a product isometry $V_n$.
The cell coherent vector $c_\theta=|\theta\sqrt{q_0}e_1\rangle$ agrees with the mapped
curve to first order and satisfies
$\|V_{\rm cell}\psi(\theta)-c_\theta\|\le C_\phi q_0\theta^2$.

Product fidelity accumulates the cell error with a factor $\sqrt n$;
replacing $u_n$ by $u$ costs $H/\sqrt n$.
With $\Pi_n=V_nV_n^\dagger$ and any fixed density operator
$\rho_{\rm fix}$ on the finite record Hilbert space, define
\[
 \mathsf C_n(\omega)=V_n^\dagger\omega V_n+
 \tr[(I-\Pi_n)\omega]\rho_{\rm fix},\quad
 T_n(\rho)=\tr_{\rm other}(V_n\rho V_n^\dagger),\quad
 S_n(\sigma)=\mathsf C_n(\sigma\otimes\proj0_{\rm other}).
\]
These maps are completely positive and trace preserving, with
$\mathsf C_n(V_n\rho V_n^\dagger)=\rho$. Contractivity gives both comparisons.

Round $x$ to $\min\{n-1,\lfloor nx\rfloor\}/n$. The report rounding and $u-u_n$ both lie in $[0,1/n]$, so the
displacement changes by at most $1/n$. The Lipschitz bound for the loss
gives $d_{\loc,n}$; Lemma~\ref{qcp-C-lem:transfer} then transfers the optima.
\end{proof}

\section{Recovery by retaining the boundary region}\label{qcp-C-sec:recovery}
\begin{qcpCtheorem}[Recovery with aligned blocks]\label{qcp-C-thm:recovery}
Use the prior and losses of \eqref{qcp-M-eq:prior}, \eqref{qcp-M-eq:losses},
with $B_\loc=2$, $B_\amp=4L^2$. Fix $0<L<\infty$ so that the physical interval
lies in the curve domain.
Partition the record into $N_b=\lceil n/m\rceil$ blocks of length at most
$m$, including the short last block, and assume
$p_-=\inf_{|z|\le L}b((s+az)/\sqrt{nq_0})>0$.
For any $\delta_j>0$, one complete paired rule gives
\begin{equation}
 e_j\le B_jE_{n,m}+\delta_j,\qquad
 E_{n,m}=\min\left\{1,\frac{5N_b}{2np_-}
 +\frac{M_1aL}{\sqrt{2q_0}}\sqrt{\frac{m+1}{n}}\right\}.
 \label{qcp-C-eq:block-budget}
\end{equation}
The bounds hold for both tasks, so $\G_\lambda\le B_\lambda E_{n,m}$ at every weight.
\end{qcpCtheorem}
\begin{proof}
Test successive blocks against their initial state projections, stopping
at the first complementary outcome and retaining its label and output
record, with no click as
a separate outcome. Only blocks at or after the boundary can click.
For $l$ changed sites in the boundary block, the missed mass is
$a_l=(1-p)^l$. The distance from the untouched record with the correct
label is $(\sqrt{4a_l-3a_l^2}+a_l)/2\le\sqrt{a_l}+a_l/2$:
the lost vacuum amplitude and all other labels contribute these two
terms. Summing over $l$ and blocks under the uniform position prior
gives $e\le5N_b/(2np_-)$.  

Given the correct label, replace the boundary block amplitude by
$s/\sqrt{nq_0}$. The conditional prior factors into position in the
block and amplitude in the suffix. Pair rules approaching the two
infima, using the prior alone for an empty suffix. The replacement error is at most
$M_1a|z|\sqrt{l/(nq_0)}$. Averaging with $\mathbb E|z|=L/2$ and
$\mathbb E l\le(m+1)/2$ gives
$f\le M_1aL\sqrt{(m+1)/(nq_0)}/(2\sqrt2)$.
The forward error is $e+f$; discarding the label defines the reverse
channel with error $f$.
Lemma~\ref{qcp-C-lem:transfer} gives $B_j(e+2f)+\delta_j$, capped by $B_j$.

The curve bound gives $\sqrt{p_-}\ge\sqrt{b_n}-M_1aL/\sqrt{nq_0}$.
For $s_n\to\infty$, $s_n=o(\sqrt n)$, $b_n\sim s_n^2/n$ and eventually
$p_-\ge b_n/4$. Taking $m=\lceil n\tau_n^{-2/3}\rceil$ makes
$5N_b/(2np_-)\le10\tau_n^{-1/3}+10/\tau_n$, while the replacement term
is $O_L(\tau_n^{-1/3}+n^{-1/2})$. Letting $\delta_j\downarrow0$ proves Theorem~\ref{qcp-M-thm:main}(ii).
\end{proof}

\section{Positive and informative marginal risks}\label{qcp-C-proof:nondegeneracy}
\begin{qcpClemma}[Information bound for a uniform amplitude prior]
\label{qcp-C-lem:uniform-amplitude-information}
Let $z\sim\Unif[-L,L]$, $L>0$, and let the trace norm differentiable state
model have quantum Fisher information at most $\mathcal I$ uniformly in $z$.
Every measurement and amplitude report $Y$ satisfy
\[
 \mathbb E(Y-z)^2\ge\{3/L^2+(6/5)\mathcal I\}^{-1}.
\]
\end{qcpClemma}
\begin{proof}
After the report is clipped to $[-L,L]$, integration by parts with
$g=L^2-z^2$ gives
$\mathbb E[(Y-z)(g'+g\partial_z\log p_z)]=2L^2/3$; the endpoint term
vanishes. Cauchy--Schwarz then proves the bound using
$\mathbb E(g')^2=4L^2/3$, $\mathbb E g^2=8L^4/15$, the zero conditional
mean of the score and the classical Fisher information bound
$\le\mathcal I$.

For zero likelihoods, choose a faithful reference law and use
$(1-\eta)p_z+\eta$. Trace norm differentiability and the information
bound justify the score identity; bounded reports give risk convergence
as $\eta\downarrow0$.
\end{proof}

\begin{proof}[Informative marginals in Theorem~\ref{qcp-M-thm:main}]
Let $\theta_{0,n}=s_n/\sqrt{nq_0}$ denote the known physical center. Set
$\kappa_n=M_1aL/\sqrt{q_0\tau_n}$,
$p_\pm=b_n(1\pm\kappa_n)^2$, $d_n=\lfloor b_n^{-1}\rfloor$.
Eventually $0<p_-\le b(\theta)\le p_+<1$, $d_n<n$,
$p_\pm/b_n\to1$, $b_nd_n\to1$ and $d_n/n\to0$.
With the amplitude revealed, pair the positions $k,k+d_n$. The corresponding
states have squared fidelity at least $(1-p_+)^{d_n}$. Applying Helstrom's
bound \cite{qcp-C-helstrom1969} and the loss incurred by a wrong
classification gives
\[
 r_{\loc,n}^\star\ge\frac{1-d_n/n}{2}\min\{2,b_nd_n/K\}
 [1-\sqrt{1-(1-p_+)^{d_n}}].
\]
The factor $1/2$ counts each pair once. Testing sites until the first
click gives risk at most
$2b_n(1-p_-)/(Kp_-)\to2/K$, whereas a fixed report has loss two outside at most
$2K/b_n+1$ sites, which gives the prior only limit.

For the amplitude lower bound, reveal $k$. For
$q(\theta)=\|(I-\rho(\theta))\phi'(\theta)\|^2$, the derivative bound
$|q'|\le C_q=4M_1M_2+2M_1^3$ bounds the conditional QFI by
$4a^2\ell_k(1+\delta_{q,n})$, where $\ell_k=(n-k+1)/n$ and
$\delta_{q,n}=C_q(s_n+aL)/(q_0\sqrt{nq_0})=o(1)$.
The preceding lemma gives
\[
 r_{\amp,n}^\star\ge\frac1n\sum_{k=1}^n
 \{3/L^2+(24a^2/5)\ell_k(1+\delta_{q,n})\}^{-1}.
\]
Its Riemann limit is $(24a^2/5)^{-1}\log(1+8a^2L^2/5)>0$.

For the upper bound, test $m=o(n)$ sites ending at
$j_0=\lfloor n/2\rfloor$, with $mp_-\to\infty$. A click places the last
$N=\lfloor n/4\rfloor$ sites after the change. Use
$O_\theta=|\phi(\theta)\rangle\langle\chi_\theta|+|\chi_\theta\rangle\langle\phi(\theta)|$,
$\chi_\theta=(I-\rho(\theta))\phi'(\theta)/\sqrt{q(\theta)}$.
The known observable $O_{\theta_{0,n}}$ is measured on those sites.
Outcomes $X\in\{-1,0,1\}$ have variance at most one, mean derivative
$2\sqrt{q(\theta_{0,n})}$ at the known center and second derivative bounded by
$2(M_2+M_1^2)$. Thus $Z=n^{-1/2}\sum_{i=1}^NX_i$ obeys
$\mathbb E_zZ=az/2+o(1)$, $\operatorname{Var}_zZ\le N/n$, uniformly.
The clipped report $2aZ/(3/L^2+a^2)$ has risk at most
$1/(3/L^2+a^2)+o(1)$, by expansion and $\mathbb Ez^2=L^2/3$. Report zero when the pilot
has no click.

The ideal branch $k\le j_0$ preserves the amplitude prior, and the pilot
disagrees with that branch with probability at most
$m/n+e^{-mp_-}=o(1)$. The pilot and tail act on disjoint sites, giving total risk
$1/[2(3/L^2+a^2)]+L^2/6+o(1)<L^2/3$.
Both marginal risks and their gains over the prior alone are therefore
positive in the limit; the prior risks are $2+o(1)$ and $L^2/3$.
\end{proof}

\endgroup

\begingroup
\qcpresetvolume{Ccert}
\qcparticlelayout

\renewcommand{\thesection}{T\arabic{section}}\setcounter{section}{0}
\section{Certificate construction for continuous ideal step risks}\label{qcp-C-sec:certificate}
At $L=1/2$ and $s=s_{\mathrm c}$, partition the position prior into
$N$ cells. The following measurements and duals give the bounds in
Section O \cite{qcp-M-holevo1973,qcp-C-eldar2003}.

\subsection{Representation and complete upper bounds}
Write $u=J/N+r$, with independent uniform $J\in\{0,\ldots,N-1\}$
and $r\in[0,1/N]$. Delete the last $r$ units of the field at the left grid endpoint,
translate right and add vacuum to obtain the shifted field. Subtracting
$r$ from a clipped position report reduces radial loss and preserves
amplitude risk, so the grid optima are lower bounds.

For eight Gauss amplitude nodes, put $h_i=a(3/2+z_i)$,
$\mathsf G(v)_{ij}=e^{-(h_i-h_j)^2(1-v)/2}$ and
$\mathsf D_{ii}=e^{-h_i^2/(2N)}$.
Endpoint projections use $\mathsf D_{jj}$, and the innovation Gram
matrix is $\mathsf G(v)-\mathsf D\mathsf G(v+1/N)\mathsf D$.
Directed Cholesky factorization gives dimension $1+8N$, including vacuum.
For $u=l+t\in[l,l+1/N]$, the innovation coordinates are
\begin{equation}
 C_{\rm cell}^{-1}
 \left(\mathsf G(l)_{ij}e^{(h_j^2/2-h_ih_j)t}
       -(\mathsf D\mathsf G(l+1/N)\mathsf D)_{ij}e^{h_j^2t/2}\right)_{i=1}^8.
 \label{qcp-C-eq:innovation}
\end{equation}
Endpoint residuals in distinct cells occupy orthogonal spaces $Q_c$.
For dyadic frames $V,F$, certify $V^\dagger V\preceq n_VI$ and
$FF^\dagger\preceq n_FI$, so $VF/\sqrt{n_Vn_F}$ is contractive.
Add the $Q_c$ effects and the positive complement
$I-\sum_iE_i-\sum_cQ_c$, with report $(1/2,0)$.
Starting from fallback risks $1/2$ and $1/12$, integrate each effect's
risk change across every loss kink. Integrals of $e^{bt}$ and $te^{bt}$
use entire series remainders at small exponents.

\subsection{Prior integration and feasible lower bounds}
Complex amplitude continuation has trace one and trace norm
$e^{2a^2(1-u)(\operatorname{Im}\zeta)^2}$.
Completeness gives $\|\int y^jM(dx,dy)\|\le2^{-j}$ for $j=1,2$.
Integrate the quadratic term exactly. The remaining Cauchy bounds at
radius $R$ are $(R+3/4)e^{2a^2R^2}$ for amplitude and
$2e^{2a^2R^2}$ for position, giving the Gauss allowances
\begin{equation}
 Q_8(R)=\frac{(8!)^4}{17(16!)^2}\frac{e^{2a^2R^2}}{R^{16}},\qquad
 \varepsilon_\loc=2Q_8(R),\quad \varepsilon_\amp=(R+3/4)Q_8(R).
 \label{qcp-C-eq:gauss-fee}
\end{equation}
At $R=6/5$, both task allowances and their weighted joint allowance
are below $1.2\times10^{-7}$, charged once per endpoint.

For $K_{\rm cmp}=B^\dagger/\sqrt c$, $c\ge\|B\|^2$, retain the
compression failure flag. Return fidelity is at least
$\langle\psi,K_{\rm cmp}^\dagger K_{\rm cmp}\psi\rangle^2$;
omitted mass $\delta$ gives average distance at most $\sqrt{2\delta}$.
If $\|W-W_c\|_F\le e$, $\|W\|_F\le1$ and node error is $e_z$, then
\begin{equation}
 \kappa_\loc=2e(2+e),\qquad
 \kappa_\amp=e(2+e)+2e_z(1+e)^2,\qquad
 \kappa_R=(\kappa_\loc+\kappa_\amp)/2.
 \label{qcp-C-eq:coordinate-fee}
\end{equation}
Indeed, $\|WDW^\dagger-W_cDW_c^\dagger\|\le\|D\|e(2+e)$.
Shift the $d_c$-dimensional dual by $\kappa_jI$, at cost $d_c\kappa_j$,
and use the zero dual on failure, charging $B_j\sqrt{2\delta}$.
An entrywise residual bound $|Z-FF^\dagger|\le e_Z$ gives
$Z\succeq-d_ce_ZI$ and certifies the corrected slacks.

Let $S_0,S_1,S_2$ be the unnormalized amplitude moments.
For any self-adjoint $A_0$, set $E_A=S_0A_0+A_0S_0-2S_1$. The dual
\[
 Y_\amp=S_2-A_0S_0A_0-\|E_A\|I/2
\]
is feasible: for $|y|\le1/2$, its slack is
$(y-A_0)S_0(y-A_0)+yE_A+\|E_A\|I/2\succeq0$.
The entrywise estimate $\|E_A\|\le d_c\max_{ij}|(E_A)_{ij}|$
completes the bound.

\renewcommand{\thesection}{O\arabic{section}}\setcounter{section}{0}
\section{Positive gap certificate at the reference center}\label{qcp-C-sec:latest-certificate}
Set $s=s_{\mathrm c}$ and $L=\lambda=1/2$.
Section T supplies the measurement and integration interfaces.

\subsection{Purified comparison and feasible duals}
For $N=256$, $d=1/N$ and $u=j/N+r$, purify $r$ in
$L^2([0,d],dr/d;\mathcal F)$ by
$\Psi_{j,z}(r)=|(s_{\mathrm c}+az)\mathbf1_{[j/N+r,1]}\rangle$.
A decomposable POVM with fiber report
$\operatorname{clip}_{[0,U]}(x-r)$, $U=1-d$, reduces location loss
$2|x-u|$ for target $j/N$ and preserves amplitude risk. Thus
\begin{equation}
 R_{\rm pur}^{\star}\le R_{\rm orig}^{\star},\qquad
 r_{j,{\rm pur}}^{\star}\le r_{j,{\rm orig}}^{\star},\quad
 j\in\{{\rm loc},{\rm amp}\}.
 \label{qcp-C-O-eq:pur-direction}
\end{equation}
For $m_{ab}=(h_a-h_b)^2/2$ and
$\exprel(t)=(e^t-1)/t$, with $\exprel(0)=1$, integration over $r$ gives
\begin{align}
 \overline K(v)_{ab}&=e^{-m_{ab}(1-v)}\exprel(m_{ab}d),\\
 \langle\Psi_{j,a},\Psi_{k,b}\rangle
 &=e^{-h_a^2(k-j)/(2N)}\overline K(k/N)_{ab},\quad j\le k.
 \label{qcp-C-O-eq:pur-kernel}
\end{align}
The backward innovation is $\overline K(v)-D\overline K(v+d)D$,
$D_{aa}=e^{-h_a^2d/2}$, and the terminal block is
$\exprel(-m_{ab}d)$. Directed Cholesky factorization with eight nodes gives
dimension $8N$. The holomorphic density has trace one and norm at most
$e^{2a^2(\operatorname{Im}\zeta)^2}$, so Section T gives
\begin{equation}
 \varepsilon_{\rm prior}=3/25000000.
 \label{qcp-C-O-eq:prior-fee}
\end{equation}

For dyadic weighted coordinates $W$ of rank $d_r=1675$, let
$C_{x,y}=W\operatorname{diag}(|x-u_j|+(y-z_a)^2/2)W^\dagger$.
The stored dual $Y_0$, positive correction $D_{\rm final}$,
$p=2^{-36}$ and row repairs $d_i$ give
\begin{equation}
 Y_{\rm grid}=Y_0-D_{\rm final}-pI-\operatorname{diag}(d_i)
 \preceq C_{x,y}.
 \label{qcp-C-O-eq:grid-dual}
\end{equation}
For each grid action, the residual
$E_k=C_{x,y}-Y_0+D_k+pI-F_kF_k^\dagger$ has absolute row sums at most
$d_i$. The bound $2|v_iv_j|\le|v_i|^2+|v_j|^2$ gives
$E_k+\operatorname{diag}d_i\succeq0$; subsequent positive corrections
preserve earlier constraints. This covers $256\cdot129=33024$ actions.

The arithmetic uses binary64 products and sums, with FMA allowed,
without overflow, underflow or products at reduced precision.
For $u_f=2^{-53}$ and $\gamma_m=mu_f/(1-mu_f)$, a length-$k$ dot
product has error at most $\gamma_{2k}\sum|a_ib_i|$.
Row norm bounds in powers of two $w_i,f_i,b_i$ give entry errors
$2^{-38}w_iw_j$, $2^{-39}f_if_j$ and $2^{-35}b_ib_j$ for costs,
Gram factors and positive corrections, with remaining allowance
$2^{-44}$. Use $\gamma_{4096+2\cdot33024}<2^{-35}$ and double the
positive norm sums. If $A_i$ is the exact absolute computed residual
sum and $B_i$ its error allowance, then
$\widehat s_i\ge(1-\gamma_{d_r+16})(A_i+B_i)$ with
$\gamma_{d_r+16}<1/4$. Maximizing
$d_i^{\rm raw}=\lceil4\widehat s_i2^{64}\rceil$ over actions gives
\begin{equation}
 d_i=2^{-64}\left\lceil d_i^{\rm raw}/3\right\rceil
       \ge A_i+B_i.
 \label{qcp-C-O-eq:repair}
\end{equation}
Position costs are affine between nodes. Amplitude chords of spacing
$\Delta=1/128$ exceed the quadratic cost by at most
$\Delta^2WW^\dagger/8$. Subtract this matrix and
$\varepsilon_{\rm coord}I$ to extend the dual to every continuous report.

\subsection{Return to the original risk}
For the compression $V$, the exact bound
$\|V^\dagger V-I\|\le\sqrt{\|(V^\dagger V-I)^2\|_\infty}<1$
certifies full column rank and contraction after normalization;
$\|\cdot\|_\infty$ is the maximum absolute row sum. Put
$t=c_{\rm old}/c_{\rm new}\ge1$. The omitted mass and compressed
lower objective are
$\delta_{\rm new}=1-t(1-\delta_{\rm old})$ and
$q=t[\tr Y_{\rm grid}-\Delta^2\tr WW^\dagger/8-d_r\varepsilon_{\rm coord}]$.
\begin{qcpClemma}[Square root lower return]
For a contraction of full column rank $B=U|B|$ from the retained space
to the original space, unit vectors $\psi_\theta$, probability prior
$\pi$ and $0\le\ell(a,\theta)\le L_\theta$, a compressed lower bound
$q$ yields
\begin{equation}
 r_{\rm orig}^{\star}\ge(\sqrt{q_+}-\sqrt D)_+^2,
 \qquad D\ge\int L_\theta(1-\|B^\dagger\psi_\theta\|^2)d\pi.
 \label{qcp-C-O-eq:sqrt-return}
\end{equation}
\end{qcpClemma}
\begin{proof}
The polar isometry gives the positive contraction $S=UB^\dagger$.
Pulling an original POVM through $U$, the triangle inequality in the
direct sum gives
$\sqrt{q(M)}\le\sqrt{r(M)}+
[\int L_\theta\|(I-S)\psi_\theta\|^2d\pi]^{1/2}$.
Since $(I-S)^2\preceq I-S^2$, the last term is at most $\sqrt D$.
Take near minimizers of the original risk.
\end{proof}
The joint loss envelope is $383/256$. With
$D=(383/256)\overline\delta_{\rm new}$, the return gives
\begin{equation}
 R^{\star}_{\rm orig}\ge
 (\sqrt{q_+}-\sqrt D)_+^2-\varepsilon_{\rm prior}\ge\JointLower.
 \label{qcp-C-O-eq:final-lower}
\end{equation}
For $q>D$, enclose $\sqrt{qD}$ above by a rational number and evaluate
$q+D-2\sqrt{qD}$; for $q\le D$, use zero.

\subsection{Numerical bounds}
The saved coordinates, dual factors and upper measurements give the
following allowances. Displayed intervals are rounded outward;
calculations use the exact values in the numerical archive.
\begin{center}\small
\begin{tabular}{lr}
\toprule
Quantity & Bound\\
\midrule
Repaired grid dual trace & $[0.0941492445,\,0.0941492446]$\\
Report mesh allowance & $[7.62939453,\,7.62939454]\times10^{-6}$\\
Coordinate trace allowance & $<2.744\times10^{-10}$\\
Compressed objective $q$ & $[0.0941416148,\,0.0941416149]$\\
Weighted omitted mass $D$ & $<2.561\times10^{-12}$\\
Prior integration allowance & $1.2\times10^{-7}$\\
\bottomrule
\end{tabular}
\end{center}
Subtracting the marginal upper bounds from the returned joint lower
bound gives Table~\ref{qcp-C-tab:anchor}.
\begin{table}[htbp]
\centering
\begin{tabular}{lc}
\toprule
Quantity at $L=1/2$, $s=s_{\mathrm c}$ & Certified interval\\
\midrule
$r_\loc^\star$ & $[0.12461689,\,0.12546830]$\\
$r_\amp^\star$ & $[0.06061206,\,0.06080745]$\\
$R_{1/2}^\star$ & $[0.09414051,\,0.09444910]$\\
$10^3\G_{1/2}$ & $[1.0026,\,1.8347]$\\
\bottomrule
\end{tabular}
\caption{Bounds for the continuous prior \eqref{qcp-M-eq:prior} at
$L=1/2$, $h\in[a,2a]$ and position loss $2|x-u|$.
Each interval is rounded outward from the exact risk certificate.
Section~\ref{qcp-C-sec:certificate} gives the measurement interfaces.}
\label{qcp-C-tab:anchor}
\end{table}

\endgroup
\begingroup
\qcpresetvolume{D}
\qcparticlelayout
\setlength{\emergencystretch}{3em}

\numberwithin{equation}{section}

\newcommand{\GapLower}{0.001002640149993936}

\renewcommand{\thesection}{D\arabic{section}}\setcounter{section}{0}
\section{Robustness in prior width at the reference center}
\begin{qcpDcorollary}[A width neighborhood in fixed units]
\label{qcp-D-cor:width-neighborhood}
Fix $s=s_{\mathrm c}$, $a,q_0,K$ and both raw losses. Put
$I_L=\{L:|L-1/2|\le3\times10^{-4}\}$ and $g_{I_L}=2.525\times10^{-4}$.
For each known $L\in I_L$, amplitudes are positive and
\begin{equation}
 \G_{1/2}(s_{\mathrm c};L)>g_{I_L},\qquad
 \G_{\lambda}(s_{\mathrm c};L)
 >2\min\{\lambda,1-\lambda\}g_{I_L}\quad(0<\lambda<1).
 \label{qcp-D-eq:width-positive}
\end{equation}
With $H=s_{\mathrm c}+a(0.5003)$, $C_{I_L}=1+2(0.5003)^2$ and
$H/\sqrt{nq_0}\le r_0$, uniformly on $I_L$,
\begin{equation}
 \G_{1/2,n}(s_{\mathrm c};L)
 >g_{I_L}-2C_{I_L}\varepsilon_{n,H}-d_{\loc,n},
 \label{qcp-D-eq:width-finite-return}
\end{equation}
where $d_{\loc,n}$ is evaluated at $s_{\mathrm c}$ in
Theorem~\ref{qcp-C-thm:channels}. For a fixed known curve, one
sufficiently large $n$ gives positivity throughout $I_L$ and at every
fixed interior weight.
\end{qcpDcorollary}
\begin{proof}
Couple $z=Lr$ and $z_0=L_0r$, where $L_0=1/2$ and
$(u,r)\sim\Unif[0,1]\otimes\Unif[-1,1]$. Set
$\Delta_L=|L-L_0|$, $L_m=\max(L,L_0)$.
The average state distance is at most $\varepsilon=a\Delta_L/3$,
since $\mathbb E|r|=1/2$ and $\mathbb E\sqrt{1-u}=2/3$.
On reports $[-L_m,L_m]$, the amplitude loss changes by at most
$\Delta_L[2L_m|r|+(L+L_0)r^2]$, with mean
$\varepsilon_\ell=\Delta_L[L_m+(L+L_0)/3]$.
Comparing states at fixed targets, then changing the amplitude target,
Lemma~\ref{qcp-C-lem:transfer} gives, with
$B_\lambda(L_m)=2\lambda+4L_m^2(1-\lambda)$,
\begin{align}
 |\Delta r_\loc^\star|&\le2\varepsilon,&
 |\Delta r_\amp^\star|&\le4L_m^2\varepsilon+\varepsilon_\ell,\nonumber\\
 |\Delta R_\lambda^\star|&\le B_\lambda(L_m)\varepsilon+(1-\lambda)\varepsilon_\ell,&
 |\Delta\G_\lambda|&\le2B_\lambda(L_m)\varepsilon+2(1-\lambda)\varepsilon_\ell.
 \label{qcp-D-eq:width-continuity}
\end{align}
At equal weights, $\Delta_L\le3\times10^{-4}$, $L_m\le0.5003$,
$L+L_0\le1.0003$ and $a<1.6655$ give
$|\Delta\G_{1/2}|<7.5\times10^{-4}$.
The bound on $a$ follows from eight terms of
$2\operatorname{atanh}(1/3)=\log2$ and the geometric tail.
Since $\GapLower>1.0025\times10^{-3}$, the central certificate gives
$\G_{1/2}>g_{I_L}$. Nonnegative regrets give the factor
$2\min\{\lambda,1-\lambda\}$; \eqref{qcp-C-eq:finite-return}
gives the finite record bound with common $H$.
\end{proof}

\endgroup

\pdfbookmark[1]{II. Leading recovery for shrinking rises}{qcp-SM-part2}
\section*{II. Leading recovery for shrinking rises}
\begingroup
\qcpresetvolume{C}
\qcparticlelayout

\renewcommand{\thesection}{J\arabic{section}}\setcounter{section}{0}
\section{Leading recovery for smooth rises}\label{qcp-C-J-J-start}
Consider \eqref{qcp-M-eq:rise-family}--\eqref{qcp-M-eq:rise-loss} for
fixed $F,L$ and $0\le\lambda\le1$.
Put $c_\chi=\sqrt{1+\chi}$, $\delta=c_\chi/s^2$, and
$B_\loc=2$, $B_\amp=4L^2$, $B_\lambda=\lambda B_\loc+(1-\lambda)B_\amp$.
Define $J_F(r)=\int_{\mathbb R}[F(x-r)-F(x)]^2dx$ and
$\kappa_F^2=\int_0^1F'(x)^2dx$. The local location kernel is
\begin{equation}
 k_\chi(v)=\exp[-\chi J_F(c_\chi|v|/\chi)/2]\quad(0<\chi<\infty),
 \qquad k_0(v)=e^{-|v|/2},\quad k_\infty(v)=e^{-\kappa_F^2v^2/2}.
 \label{qcp-C-J-eq:kernel}
\end{equation}
The experiment $\mathsf P_\chi$ has this overlap. Its spectral density
and canonical outcome density are
\begin{equation}
 m_\chi(p)=\frac1{2\pi}\int e^{-ipv}k_\chi(v)\,dv,\qquad
 p_\chi(x)=|\mathcal F^{-1}\sqrt{m_\chi}(x)|^2,
 \label{qcp-C-J-eq:spectral}
\end{equation}
where $\mathcal F$ is unitary on $L^2(\mathbb R)$ and the inverse is in $L^2$.
Write
\begin{equation}
 r_F(\chi)=2\int\min\{1,|x|/K\}p_\chi(x)\,dx,\qquad
 \overline{\mathcal A}_L=\int_0^1\mathcal A_L(a\sqrt e)\,de,
 \label{qcp-C-J-eq:risks}
\end{equation}
where $\mathcal A_L(c)$ is the amplitude optimum for $\ket{cz}$ with
$z\sim\Unif[-L,L]$. On $e\in[0,1]$, coherent state comparison gives
the modulus $4L^2aL\sqrt{|e-e'|}$ for $\mathcal A_L(a\sqrt e)$.
The three optimal risks satisfy
\begin{equation}
 r_\loc^\star=r_F(\chi)+o(1),\quad
 r_\amp^\star=\overline{\mathcal A}_L+o(1),\quad
 R_\lambda^\star=\lambda r_F(\chi)+(1-\lambda)\overline{\mathcal A}_L+o(1).
 \label{qcp-C-J-eq:three-risks}
\end{equation}
as $s\to\infty,w\to0$, with arbitrary $\chi=s^2w$.

\subsection{Geometry and spectral continuity}
\begin{qcpClemma}[Geometry and uniform spectral comparison]
\label{qcp-C-J-lem:geometry}\label{qcp-C-J-lem:spectral}
For $r\ge0$, $r^2/(1+r)\le J_F(r)\le r$; for $r\ge1$,
$J_F(r)=r-C_F$, $C_F=2\int_0^1F(1-F)$; and
$J_F(r)/r^2\to\kappa_F^2$ as $r\downarrow0$.
The spectral densities $m_\chi$ are normalized. If $\chi_n\to\chi_0$ in
$[0,\infty]$, then $\|m_{\chi_n}-m_{\chi_0}\|_1\to0$ and
$\|\sqrt{m_{\chi_n}}-\sqrt{m_{\chi_0}}\|_2\to0$.
There are parameter independent channels in both directions whose errors
are bounded by this square-root distance, uniformly over the translation
$v\in\mathbb R$.
\end{qcpClemma}
\begin{proof}
The function $F(x)-F(x-r)\in[0,1]$ has integral $r$ and is supported on
an interval of length at most $1+r$. The lower bound follows from
Cauchy--Schwarz, and the upper bound from squaring this function.
Splitting the integral at $0,1,r,r+1$ yields the large-shift identity;
differentiability of translation in $L^2$ gives the small-shift limit.
For the quintic,
$\kappa_F^2=10/7$, $C_F=50/231$.
The exact coherent fields
\begin{equation}
 g_{\chi,v}(x)=F((x-c_\chi v)/\chi)-F(x/\chi)
 \quad\hbox{in }L^2(\mathbb R)
 \label{qcp-C-J-eq:fields}
\end{equation}
give \eqref{qcp-C-J-eq:kernel}, with indicator-difference and
one-mode $\ket{\kappa_Fv}$ endpoints. Geometry gives pointwise kernel
convergence and the integrable majorant
\[
 |k_\chi(v)|\le e^{-v^2/(2(1+|v|))},\qquad
 \chi J_F(c_\chi|v|/\chi)\ge
 \frac{(1+\chi)v^2}{\chi+c_\chi|v|}\ge\frac{v^2}{1+|v|}.
\]
Normalization follows from positive definiteness and Fourier inversion,
and the asserted norm convergence follows from dominated convergence and
Scheffe's theorem. Equality of the Gram kernels identifies the cyclic states $\ket{\psi_{\chi,v}}$
with $e^{-ipv}\sqrt{m_\chi(p)}$. Their vector distance is independent of $v$; the resulting isometries
and completed inverses give both channels.
\end{proof}

\subsection{Product comparison under the original prior}
Partition $[0,1]$ into $N_c$ cells of length $b=1/N_c$, centered at $u_j$.
A cell is called good if $(j+1)b+w\le1$. Its suffix length is
$\ell_j=1-(j+1)b-w$, and its local coordinate is $v=(u-u_j)/\delta$.
The mass of the remaining cells is
$\beta=1-b\lfloor(1-w)/b\rfloor\le w+b$.
The reference experiment $\mathsf Q_{s,w}$ emits a good label together with
\begin{equation}
 \ket{\psi_{\chi,v}}\otimes\ket{a\sqrt{\ell_j}\,z},
 \label{qcp-C-J-eq:comparison-product}
\end{equation}
or one bad label with vacuum, preserving the original targets and losses.
For $0<t<1$, put
\begin{align}
 P&=\min\left\{1,2(N_c+1)\left(\frac{aL}{s}
       +\frac{\sqrt{2/\pi}}{2s\sqrt t}\right)\right\},\nonumber\\
 H&=t(s\sqrt{b/2}+aL),\qquad D_w=aL\sqrt{b+w}.
 \label{qcp-C-J-eq:budgets}
\end{align}
Two-way channels have prior averaged errors
\begin{equation}
 \varepsilon_T\le\beta+P+H+D_w,\qquad
 \varepsilon_S\le\beta+D_w,
 \label{qcp-C-J-eq:channels}
\end{equation}
capped at one. The reference factors at each label. Lemma~\ref{qcp-C-lem:transfer} gives
$0\le\G_\lambda\le B_\lambda(\varepsilon_T+\varepsilon_S)$ and a complete rule
with $e_j\le B_j(\varepsilon_T+\varepsilon_S)+\eta$ for both tasks and any $\eta>0$.
\begin{proof}[Product comparison]
A fraction $t$ of the intensity is assigned to the pilot. On a good cell,
$A_w(u)=\int_0^1F((r-u)/w)dr=1-u-w(1-\mu_F)$, $\mu_F=\int_0^1F$.
Constant-mode homodyne gives $X=\sqrt t(s+az)A_w(u)+Z/2$, $Z\sim N(0,1)$;
the clipped report $1-w(1-\mu_F)-X/(s\sqrt t)$ has error at most
$aL/s+|Z|/(2s\sqrt t)$. Integration over the boundary strips, with the
pilot's selection probabilities, gives $P$.

For an observed good label, displace the retained field by
$-\sqrt{1-t}sF((r-u_j)/w)$. Before attenuation,
its residual is $s[F((r-u)/w)-F((r-u_j)/w)]+azF((r-u)/w)$.
The first term is supported on $W_j=[jb,(j+1)b+w]$ and has norm at most
$s\sqrt{b/2}$; the offset has norm at most $aL$. Thus attenuation costs
$H$, and omitting the offset on $W_j$ costs $D_w$. The displacement is
independent of $u,z$. Spatial factorization, followed by the change of
variable $x=s^2(r-u_j)$, gives \eqref{qcp-C-J-eq:fields} and the constant
suffix mode, which is vacuum when $\ell_j=0$.

Given $u,z$, the pilot and retained field are independent. Wrong
labels and terminal cells contribute $P+\beta$ to the forward error.
The reverse map embeds the location field in $W_j$, the amplitude
mode in the suffix, and vacuum elsewhere, then restores the known
reference. Its errors are $D_w+\beta$. Both maps are complete channels.

At each good label, the factored prior lets one pair rules approaching
both marginal infima. Averaging over the other factor gives the
matching lower bound. The vacuum label also has zero gap.

Choose $N_c=\lceil s^{1/3}\rceil$, $b=1/N_c$, $t=1/s$. Then
$P,H=O_L(s^{-1/6})$, $D_w=O_L(s^{-1/6}+\sqrt w)$, $\beta\le w+b$ and
$e_j=O_{F,L}(s^{-1/6}+\sqrt w)$, for every $w\to0$ and arbitrary $\chi$.
\end{proof}

\subsection{Covariant risk and global risk tracking}
Let $r_F(\chi;V)$ denote the optimum for $v\sim\Unif[-V,V]$, clipped
reports and loss $2\min(1,|x-v|/K)$.
\begin{qcpClemma}[Complete covariant risk]
\label{qcp-C-J-lem:canonical}
For $\mathsf P_\chi$, complete covariant measurements, with failure assigned
cost two, have optimal risk $r_F(\chi)$ in \eqref{qcp-C-J-eq:risks}.
The finite-prior optima satisfy $r_F(\chi;V)\to r_F(\chi)$ uniformly on
$[0,\infty]$ as $V\to\infty$.
\end{qcpClemma}
\begin{proof}
On the cyclic momentum space, the shift representation has Lebesgue spectral
measure. Its normalized covariant POVMs therefore have measurable isometric
kernel dilations \cite[Theorem~2.4.3, Corollary~2.4.5 and
Section~2.5.1]{qcp-S-toigo2005}. If the finite-report part has total effect $Q$, then $Q$ commutes
with shifts. Normalize on $\operatorname{supp}Q$ and multiply its dilation
by $Q^{1/2}$ to obtain kernels $\langle\eta(p),\eta(q)\rangle$ with $\|\eta(p)\|\le1$.
Assign the commuting deficit $I-Q$ to the failure report of cost two.
The triangular benefit has nonnegative transform
$\widehat b_K(t)=K\operatorname{sinc}^2(Kt/2)$, where
$\operatorname{sinc}(x)=\sin(x)/x$ and $\operatorname{sinc}(0)=1$.
Its expectation is
\[
 (2\pi)^{-1}\iint\sqrt{m_\chi(p)m_\chi(q)}\widehat b_K(p-q)
 \operatorname{Re}\langle\eta(p),\eta(q)\rangle\,dpdq.
\]
A constant unit vector maximizes this expectation and gives $p_\chi$
\cite{qcp-C-holevoShift2003}. Young's inequality justifies the integral
for $\sqrt m\in L^2$, $\widehat b_K\in L^1$.

Restricting the prior and clipping reports gives
$r_F(\chi;V)\le r_F(\chi)$. For the reverse inequality, extend finite-window
rules to the one-point compactification $\mathbb R\cup\{\infty\}$ and
average their translates over $[-V,V]$. The risk at zero is their
finite-prior risk, while translation by $b$ changes any continuous test
of norm at most one by at most $|b|/V$.
The compactness argument in the proof of
Proposition~\ref{qcp-M-prop:common-optimality}, applied to this compact report
space, gives a complete limit along a subnet realizing the risk liminf.
The covariance defect vanishes. The loss at zero is continuous, with value
two at infinity, so the limit risk is at least the preceding covariant
minimum, including failure mass. This proves pointwise convergence.
The spectral channels give equicontinuity independently of $V$; a finite
net on $[0,\infty]$ then gives uniformity.
\end{proof}
\begin{proof}[Theorem~\ref{qcp-M-thm:joint-scale}(i): global tracking and informative marginals]
On good cells, $V=b/(2\delta)=bs^2/(2\sqrt{1+s^2w})\to\infty$
uniformly for $w\le1$. The canonical lemma, the bad-cell cost $2\beta$
and the channels in both directions prove the location formula in
\eqref{qcp-C-J-eq:three-risks}.

The amplitude optima $\mathcal A_L(a\sqrt{\ell_j})$ converge after
averaging to $\overline{\mathcal A}_L$, by their continuity modulus,
$0\le1-u-\ell_j\le b+w$ and the terminal mass bound. Transferring the
amplitude and joint infima proves
\eqref{qcp-C-J-eq:three-risks}.

Continuity of $r_F$ follows from Plancherel's theorem. Since the outcome
distribution has no point masses, $r_F>0$. The benefit is also positive because $m_\chi$ is positive on
a set of positive measure, and $\widehat b_K$ is positive near zero.
Compactness makes $0<r_F<2$ uniform.

Lemma~\ref{qcp-C-lem:uniform-amplitude-information}, applied with QFI
$4a^2e$, gives the amplitude lower bound. For the upper bound, use
homodyne $X=a\sqrt e\,z+N$, $\operatorname{Var}N=1/4$, followed by its
optimal linear predictor and clipping. Integration yields
\begin{equation}
 \frac5{24a^2}\log(1+8a^2L^2/5)
 \le\overline{\mathcal A}_L\le\frac1{4a^2}\log(1+4a^2L^2/3)<L^2/3.
 \label{qcp-C-J-eq:A-bounds}
\end{equation}
The prior-only location benefit is confined to length $2K\delta\to0$.
\end{proof}

\endgroup

\pdfbookmark[1]{III. Inverse square cost and physical transfer}{qcp-SM-part3}
\section*{III. The inverse square residual cost and its physical transfer}
\begingroup
\qcpresetvolume{C}
\qcparticlelayout

\renewcommand{\thesection}{C\arabic{section}}\setcounter{section}{0}
\paragraph{Fixed assumptions for C1--C11.}
C1--C10 establish Theorem~\ref{qcp-M-thm:uniform-first-cost}(ii) for the full coherent record;
C11 transfers the result to the finite product record and compares the leading risks.
Throughout these sections, the quintic profile is used, with \(a=2\sqrt{\log2}\),
\(K=9\log2\), \(L=1/2\), the original independent uniform \(u,z\) prior
and a fixed \(0<\lambda<1\). The scaling parameters are
\(\epsilon=s^{-1}\), \(\chi=s^2w\), and
\(\delta=\sqrt{1+\chi}\epsilon^2\).
All optima use the losses in main-text \eqref{qcp-M-eq:rise-loss}.
The compact design interval \(I=[1/2,2]\) is known;
rules may depend on \(\chi\) but are independent of \(u,z\).
\paragraph{Proof route.}
C1 supplies the relative operator bounds and C2 the quintic packet estimates.
C3 removes endpoint contributions at $o(\epsilon^2)$ precision.
C4 centers comparisons at each model's own marginal optimum;
C5 and C6 give the scalar upper and lower bounds.
C7--C9 return that scalar decision problem to the original record.
C10 verifies uniformity on $I$, and C11 gives the physical transfer.

\section{Amplitude optima, response and source lifts}
\label{qcp-C-sec:relative-foundation}
For a Hilbert--Schmidt synthesis $V$ and $Z=Z^\dagger$, $\|Z\|\le L$,
the \hyperref[qcp-SM-sec:common-optimality]{common optimality proof}
applies to $\Omega=VV^\dagger$, $\Omega_1=VZV^\dagger$ on
$\operatorname{supp}\Omega$, retaining the finite source mass:
\begin{align}
 \tfrac12\{\Omega,A\}&=VZV^{\dagger},\qquad \|A\|\le L,
 \label{qcp-C-eq:foundation-personick}\\
 e_{\amp}(J;V)&=\|YJ-JA\|_{\Omega}^2. \label{qcp-C-eq:foundation-square}
\end{align}
The spectral measurement of $A$ is amplitude optimal.
The Schur norm is the maximum absolute row or column sum; it is
submultiplicative and dominates the operator norm.
\begin{qcpClemma}[A source defect controls the amplitude optimum]
\label{qcp-C-lem:source-defect}
Let \(\widetilde A=\widetilde A^{\dagger}\) be bounded on the record space and
let \(K^{\dagger}=-K\) be bounded on the source. The source defect
$E=\widetilde A V-VZ-VK$ satisfies
\[
 \|\widetilde A-A\|_{\Omega}\le\|E\|_{\HS}.
\]
\end{qcpClemma}
\begin{proof}
Skew-adjointness gives $\{\Omega,H\}=EV^\dagger+VE^\dagger$ for
$H=\widetilde A-A$. Trace-class products and Cauchy--Schwarz yield
$\|H\|_{\Omega}^2=\operatorname{Re}\tr(V^\dagger HE)
\le\|H\|_{\Omega}\|E\|_{\HS}$.
\end{proof}

\subsection{Relative bounds for the uniform prior residuals}
Use $V_cf=\int_{-L}^L|cz\rangle f(z)\nu(dz)$,
$\nu=\Unif[-L,L]$, $L=1/2$, $0\le c\le a$.
Set $\Omega_{j,c}=V_cZ^jV_c^\dagger$, $\Omega_c=\Omega_{0,c}$, and suppress
$c$ in the polynomial representation. For monic orthogonal polynomials
$P_n$ of $\mu_c=e^{-c^2z^2}\nu$ with squared norms $h_n$, the unitary
$f\mapsto e^{c^2z^2/2}f$ gives
\begin{equation}
 T_{mn}=\frac{c^m}{\sqrt{m!}}
  \int z^m\frac{P_n(z)}{\sqrt{h_n}}\,\mu_c(dz),
  \qquad \Omega_j=T\mathsf J^jT^{\dagger}.
 \label{qcp-C-eq:polynomial-synthesis}
\end{equation}
Here $\mathsf J$ is multiplication by $z$ in this polynomial basis.
\begin{qcpClemma}[Uniform prior recurrence bound]
\label{qcp-C-lem:recurrence-bound}
For the monic orthogonal polynomials of \(e^{-t x^2}dx\) on \([-1,1]\), write
\(xP_n=P_{n+1}+\beta_n(t)P_{n-1}\), \(\beta_0=0\). Their recurrence coefficients
satisfy, for \(0\le t\le\log2\),
\[
 \frac{e^{-t}}4\le \beta_n(t)\le\frac13\qquad(n\ge1).
\]
\end{qcpClemma}
\begin{proof}
Norm differentiation and orthogonality give
$\beta_n'=\beta_n(\beta_{n-1}-\beta_{n+1})$, $h_n'/h_n=-(\beta_n+\beta_{n+1})$
and $\beta_n(0)=n^2/(4n^2-1)$. Thus $0<\beta_n\le1$ implies
$\beta_n\ge e^{-t}/4$, while $\beta_1'=-\beta_1\beta_2<0$.
Before a first exit above $1/3$, each $n\ge2$ satisfies
$\beta_n(t)\le(4/15)\exp[t/3-(1-e^{-t})/4]<1/3$ on $[0,\log2]$.
The uniform derivative bound $|\beta_n'|\le1$ permits continuation of
the infinite family.
\end{proof}
After $z=Lx$, the recurrence coefficient is at most $L^2/3$. Put
\begin{equation}
 q_{\max}=\frac{\log2}{3},\qquad
 U_*=1+\frac{8q_{\max}}{1-q_{\max}},\qquad
 S_j=U_*^2(2L/\sqrt3)^j,
 \quad R_A=2S_1,\quad R_B=2(S_2+R_A S_1).
 \label{qcp-C-eq:relative-constants}
\end{equation}
\begin{qcpCproposition}[Relative amplitude and response bounds]
\label{qcp-C-prop:relative-residual}
Let \(A_c\) be the bounded amplitude optimum in
\eqref{qcp-C-eq:foundation-personick}, and let the response \(B_c\) solve
\begin{equation}
 \tfrac12\{\Omega_c,B_c\}
 =\Omega_{2,c}-\tfrac12\{A_c,\Omega_{1,c}\}.
 \label{qcp-C-eq:response-sylvester}
\end{equation}
On the average state support, uniformly for \(0\le c\le a\), the bounds are
\begin{equation}
 \|\Omega_c^{-1/2}A_c\Omega_c^{1/2}\|\le R_A,\qquad
 \|\Omega_c^{-1/2}B_c\Omega_c^{1/2}\|\le R_B,\qquad
 \|B_c\|\le L^2+LR_A.
 \label{qcp-C-eq:relative-residual-bounds}
\end{equation}
The same bounds apply to all unnormalized finite Fock sections.
\end{qcpCproposition}
\begin{proof}
For $c>0$, write $T=CD$, $d_n=c^n\sqrt{h_n}/\sqrt{n!}$,
$q=c^2L^2/3\le q_{\max}$. Counting downward recurrence steps and factorial
log-convexity give
\begin{equation}
 |C_{n+2j,n}|\le(2q)^j/j!,
 \qquad |(C^{-1})_{n+2j,n}|\le q^j/j!.
 \label{qcp-C-eq:triangular-bounds}
\end{equation}
The inverse entries have the representation
\begin{equation}
 (C^{-1})_{n+2j,n}
 =c^{2j}\sqrt{\frac{n!}{(n+2j)!}}\,[z^n]P_{n+2j}(z),
 \quad
 |[z^n]P_{n+2j}|\le \binom{n+j}{j}(L^2/3)^j.
 \label{qcp-C-eq:triangular-inverse}
\end{equation}
Thus $\|C\|_{\rm Sch}\le e^{2q}$, $\|C^{-1}\|_{\rm Sch}\le e^q$,
$e^{-2q}I\le C^\dagger C\le e^{4q}I$, and
$d_{n+1}/d_n\le\sqrt{q/(n+1)}$.

For $\mathsf K=DHD$, $mI\le H\le MI$, the min--max principle gives
$md_n^2\le\lambda_n\le Md_n^2$. Its eigenvectors satisfy
\begin{equation}
 |(v_n)_i|
 \le\min\left\{\frac{Md_i}{\sqrt{m\lambda_n}},
                    \frac{\sqrt{M\lambda_n}}{md_i}\right\}
 \le\frac Mm\min\{d_i/d_n,d_n/d_i\}.
 \label{qcp-C-eq:relative-eigenvectors}
\end{equation}
The inverse weight estimate follows from
$D^{-1}v_n=\lambda_n^{-1}HDv_n\in\ell^2$ and
$\langle D^{-1}v_n,H^{-1}D^{-1}v_n\rangle=\lambda_n^{-1}$.
Since $M/m=e^{6q}$ and $qe^{6q}\le4\log2/3<1$, the min--max bands are
disjoint. Parity and continuation from $c=0^+$ therefore determine the
parity of each eigenvector. On indices of the same parity,
$|U_{in}|\le e^{6q}q^{|i-n|/2}$, giving
$\|U\|_{\rm Sch}\le U(q):=1+2e^{6q}q/(1-q)\le U_*$.
Finite sections and the factorial tails give this continuation.

The polar factor of the injective, dense range $T$ is unitary.
Normalized moments are $Z_j=U^\dagger \mathsf J^jU$, with
\begin{equation}
 \|Z_j\|_{\rm Sch}\le U(q)^2(2L/\sqrt3)^j\le S_j,
 \qquad q=c^2L^2/3\le\log2/3.
 \label{qcp-C-eq:q-dependent-Schur}
\end{equation}
For $\mathsf B=\Omega^{-1/2}A\Omega^{1/2}$, the entries
$2\lambda_j(Z_1)_{ij}/(\lambda_i+\lambda_j)$ give $\|\mathsf B\|_{\rm Sch}\le2S_1$.
Applying this multiplier to the response source
$Z_2-(\mathsf B Z_1+Z_1\mathsf B^\dagger)/2$ gives $R_B$.
The positive Sylvester integral gives the ordinary operator bound.
For the projection $P_N$ onto the first $N$ Fock coordinates, lower
triangularity gives $T_N=P_NT=C_ND_NP_N$, where $C_N=P_NCP_N$ and
$D_N=P_NDP_N$ act on the retained coordinates. Hence
\[
 \Omega_{j,N}
 =C_ND_N(P_N\mathsf J^jP_N)D_NC_N^\dagger.
\]
The bounds for $C_N,C_N^{-1}$ and the separated min--max bands give
the same Schur estimates after diagonalizing $D_NC_N^\dagger C_ND_N$.
The moments use compressions $P_N\mathsf J^jP_N$ of the full powers and
retain the original mass.
At $c=0$, $A_0=0$, $B_0=(L^2/3)I$ on vacuum.
\end{proof}

\subsection{Source lifts and continuity}
\begin{qcpClemma}[Skew-adjoint source lift]
\label{qcp-C-lem:source-lift}
The relative amplitude bound \(R_A\) gives a bounded skew-adjoint source
operator \(K_V\) satisfying
\begin{equation}
 AV-VZ=VK_V,\qquad \|K_V\|\le\kappa_{\mathrm{src}}:=3(R_A+L).
 \label{qcp-C-eq:source-lift}
\end{equation}
For the coherent residual response, a bounded skew-adjoint source operator
\(K_{1,V}\) satisfies
\begin{equation}
 VK_{1,V}=2BV+V[K_V,Z],
 \qquad \|K_{1,V}\|\le6(R_B+L\kappa_{\mathrm{src}}).
 \label{qcp-C-eq:first-source-lift}
\end{equation}
\end{qcpClemma}
\begin{proof}
Write $V=\Omega^{1/2}U$, $UU^\dagger=I$, $P=U^\dagger U$ and
$D=\Omega^{-1/2}(AV-VZ)$. Since $\|D\|\le R_A+L$ and
$DU^\dagger+UD^\dagger=0$, the skew lift
$K_V=U^\dagger D-D^\dagger U-U^\dagger DP$ has $UK_V=D$ and norm at most $3\|D\|$.
The response equation gives $WV^\dagger+VW^\dagger=0$ for
$W=2BV+V[K_V,Z]$; the same construction with
$\|\Omega^{-1/2}W\|\le2R_B+2L\kappa_{\rm src}$ proves the first lift.
\end{proof}
Write $K_c=K_{V_c}$ and $K_{1,c}=K_{1,V_c}$.
For two syntheses with the same target, put $\Omega_V=VV^\dagger$.
Applying the defect lemma to $A_W(V-W)-(V-W)(Z+K_W)$ and subtracting
the lift identities gives
\begin{align}
 \|A_V-A_W\|_{\Omega_V}
   &\le(2L+\kappa_{\mathrm{src}})\|V-W\|_{\HS}, \label{qcp-C-eq:weighted-A}\\
 \|V(K_V-K_W)\|_{\HS}
   &\le2(2L+\kappa_{\mathrm{src}})\|V-W\|_{\HS}. \label{qcp-C-eq:weighted-K}
\end{align}
For the coherent family, the synthesis difference satisfies
\begin{equation}
 \|V_c-V_{c'}\|_{\HS}^2
 =\int2(1-e^{-(c-c')^2z^2/2})\,\nu(dz)
 \le\frac{L^2}{3}|c-c'|^2.
 \label{qcp-C-eq:coherent-source-continuity}
\end{equation}
On compact subsets of $c>0$, uniqueness, weighted convergence and the
bounded lift norms give strong convergence on the dense range of
$V_c^\dagger$, for both $K_c$ and $A_c$. Trace norm continuity of the
moments identifies each weak limit of $B_c$ through its Sylvester equation.
For its source $C_c$, $\|B_cV_c\|_{\HS}^2=\tr(B_cC_c)$ gives
Hilbert--Schmidt continuity. The first lift identity gives weighted
continuity, with endpoint strips bounded by their mass.

\section{Quintic packets and spectral bounds}
\label{qcp-C-proof-spectral-foundations}
\label{qcp-C-proof-palm-translation-experiment}
\label{qcp-C-proof-exact-local-expansion}
\label{qcp-C-sec:packet}
Integrating the translated polynomial pieces gives
\begin{equation}
 J_F(r)=
 \begin{cases}
 \dfrac{10}7r^2-\dfrac{10}7r^4+2r^6-\dfrac{10}7r^7
                   +\dfrac5{21}r^9-\dfrac2{77}r^{11},&0\le r\le1,\\
 r-\dfrac{50}{231},&r\ge1,
 \end{cases}
 .
 \label{qcp-C-eq:profile-polynomial}
\end{equation}
Also $\nu_F=\int_0^1(1-F^2)=281/462$ and $\int(F')^2=10/7$.
With $c_\chi=\sqrt{1+\chi}$, $\psi_\chi(v)=\chi J_F(c_\chi|v|/\chi)$, put
\begin{equation}
 m_{\chi,\tau}(p)=\frac1{2\pi}\int e^{-ipv}
                         e^{-\tau\psi_\chi(v)/2}\,dv .
 \label{qcp-C-eq:packet-spectrum}
\end{equation}
For $\varphi=F'$, Plancherel's theorem gives
$\psi_\chi(v)=(c_\chi/\pi)\int(1-\cos pv)
|\widehat\varphi(\chi p/c_\chi)|^2p^{-2}dp$.
The symmetric Levy density is therefore
$\tau c_\chi|\widehat\varphi(\chi p/c_\chi)|^2/(2\pi p^2)$.
Extract a finite uniform jump component on a symmetric interval where
$\widehat\varphi\ne0$. Its compound Poisson mixture has a positive
density component everywhere, preserved by convolution with the remaining
Levy law. Integrability of the characteristic function then gives a
continuous, strictly positive density.

The first six derivatives match at $0,\pm\chi/c_\chi$.
The polynomial gives seventh derivative jumps
$7200\tau c_\chi^7/\chi^6$ at zero and
$-3600\tau c_\chi^7\chi^{-6}e^{-\tau D_\chi}$ at each join.
The regular eighth derivative has bounded variation and exponential tails.
Eight distributional integrations by parts give the expansion below,
including any fixed finite number of momentum and intensity derivatives,
uniformly on $I=[1/2,2]$, $\mathcal T=[1/2,3/2]$:
\begin{equation}
 m_{\chi,\tau}(p)=
 \frac{3600\tau c_\chi^7}{\pi\chi^6p^8}
 \{1-e^{-\tau D_\chi}\cos(d_\chi p)\}
       +O_I(|p|^{-9}),
 \quad
 D_\chi=\frac{181\chi}{462},\quad
 d_\chi=\frac{\chi}{\sqrt{1+\chi}}.
 \label{qcp-C-eq:uniform-tail}
\end{equation}
The floor $1-e^{-\tau D_\chi}>0$, compact positivity and bounded
polynomial variation give
\begin{align}
 c_1(1+|p|)^{-8}&\le m_{\chi,\tau}(p)\le c_2(1+|p|)^{-8},
 \label{qcp-C-eq:uniform-spectral-comparison}\\
 \sup_{\chi\in I,\tau\in\mathcal T,p}
 \left|\frac{\partial_\tau^j\partial_p^k m_{\chi,\tau}(p)}
                 {m_{\chi,\tau}(p)}\right|&<\infty
 \quad(j,k\text{ fixed finite}).
 \label{qcp-C-eq:packet-relative-derivatives}
\end{align}
All intensities $(1+a\epsilon z)^2$ lie in $\mathcal T$ when
$aL\epsilon\le1/8$.

\subsection{The first arrival representation}
Sorting every photon sector by its first coordinate gives the exact isometry
\begin{equation}
 |f\rangle=
 e^{-\|f\|^2/2}|{\rm vac}\rangle
 \ \oplus\
 \int_0^1 f(t)e^{-\frac12\int_0^t f(r)^2dr}
       |t\rangle\otimes|f\mathbf1_{(t,1]}\rangle\,dt .
 \label{qcp-C-eq:first-arrival-isometry}
\end{equation}
The sector factor $\sqrt n$ gives $1/\sqrt{(n-1)!}$;
symmetrization inverts the isometry and retains vacuum coherence.
At the quantum coordinate $t\le1-w$, the known displacement $-s$ is applied,
followed by dilation of the window of width $w$ and compression of the
suffix to one residual mode. With $F_\chi(x)=F(x/\chi)$,
$A_\chi(x)=\int_0^xF_\chi^2$, $\zeta=a\epsilon z$, the normalized packet is
\begin{equation}
 B_\zeta(x)=\mathbf1_{x\ge0}(1+\zeta)F_\chi(x)
       e^{-(1+\zeta)^2A_\chi(x)/2}
       |(1+\zeta)F_\chi(x+\cdot)-1\rangle_{[0,\chi]} .
 \label{qcp-C-eq:arrival-packet}
\end{equation}
The survival probability $e^{-(1+\zeta)^2A_\chi(x)}$ normalizes the packet,
giving the arrival component
\begin{equation}
 \epsilon^{-1}B_{a\epsilon z}((t-u)/\epsilon^2)
                    \otimes|az\sqrt{1-t-w}\rangle .
 \label{qcp-C-eq:physical-arrival-packet}
\end{equation}
Vacuum and late arrivals are retained in C3.

\subsection{Cross overlaps and all vector directions}
Use $\widehat B_\zeta(p)=(2\pi c_\chi)^{-1/2}
\int_0^\infty e^{-ipx/c_\chi}B_\zeta(x)dx$.
At separation $d\ge0$, differentiate the exact first arrival packets
at equal amplitudes. With $\eta$ perturbing the ket amplitude at the
later position, their overlap has first-order expansion
\[
 K^P_{1+\eta,1}(d)
 =e^{-\chi J_F(d/\chi)/2}
  \left[1+\frac{\eta}{2}\{d-\chi J_F(d/\chi)\}\right]+O(\eta^2).
\]
Indeed, the common coherent overlap contributes
$-\eta\chi J_F/2$, the relative waiting normalization contributes
$\eta d/2$, and the terms involving the squared amplitude difference
start at order $\eta^2$. The exact packet derivative bounds below
justify differentiation and Fourier inversion.
At equal amplitudes,
\begin{equation}
 \|\widehat B_\zeta(p)\|^2=m_{\chi,(1+\zeta)^2}(p).
 \label{qcp-C-eq:packet-norm-spectrum}
\end{equation}
With $m=m_{\chi,1}$, the first ket-amplitude derivative is
$[c_\chi(v-v')-\psi_\chi(v-v')]k_\chi(v-v')/2$; Fourier inversion gives
\begin{equation}
 w_0=\widehat B_0/\sqrt m,\qquad
 \eta(p)=\left.\partial_\tau\log m_{\chi,\tau}(p)\right|_{\tau=1}
                    -\frac{ic_\chi}{2}\partial_p\log m(p),\qquad
 w_1=\frac{\partial_\zeta\widehat B_\zeta|_0}{\sqrt m}
                                  -\eta w_0,\quad w_1\perp w_0.
 \label{qcp-C-eq:packet-tangent}
\end{equation}
Thus $q=a\operatorname{Re}\eta$ and $r=a\operatorname{Im}\eta$ are the
scalar intensity and phase components; $w_1$ is the orthogonal tangent.
For every fixed finite $j,k$,
\begin{equation}
 \|\partial_\zeta^j\partial_p^k\widehat B_\zeta(p)\|
 \le C_{j,k}(1+|p|)^{-4},\qquad |\zeta|\le\zeta_0<1 .
 \label{qcp-C-eq:full-packet-estimate}
\end{equation}
Put $N_\zeta(x)=\|(1+\zeta)F_\chi(x+\cdot)-1\|_{L^2[0,\chi]}^2$.
In sector $n$, the integrand is the waiting amplitude times
$e^{-N_\zeta(x)/2}\prod_{l=1}^n[(1+\zeta)F_\chi(x+y_l)-1]/\sqrt{n!}$.
After $j$ derivatives and multiplication by $x^k$, this integrand is $C^2$ with
third derivative jumps at $0,\chi,\chi-y_l$. Their absolute sum and the
regular fourth derivative $L^1$ norm are bounded by
$C_{j,k}(n+1)^{N_{j,k}}R_1^n/\sqrt{n!}$.
Four integrations by parts give $|p|^{-4}$ decay, and the $L^1$ bound
controls small $p$. Since $\sum(n+1)^{2N}(R_1^2\chi)^n/n!<\infty$,
the squared sector bounds sum. The relative spectral estimates then
bound derivatives of the normalized packet.

On the enlarged window space, use
\begin{equation}
 U_0(p)=
 \begin{pmatrix}
 0&w_0(p)^{\dagger}\\
 w_0(p)&I-|w_0(p)\rangle\langle w_0(p)|
 \end{pmatrix}.
 \label{qcp-C-eq:known-window-frame}
\end{equation}
This self-adjoint unitary maps $(0,w_0)$ to $e_0=(1,0)$.
Set $\widetilde w_1=U_0(0,w_1)$; the derivative estimates give C7's
weighted expansion. Rescaling the window to $[0,1]$ makes the sector
bounds uniform on $I$.

\section{Endpoint estimates at quadratic precision}
\label{qcp-C-proof-physical-endpoint-removal}
\label{qcp-C-proof-boundary-gap-bounds}
\label{qcp-C-sec:endpoints}
Fix $L<\infty$, $w=\chi\epsilon^2$, $s=\epsilon^{-1}\ge2aL$,
$w<d<1/4$, and $B_\lambda=2\lambda+4L^2(1-\lambda)$.
For real filters $\sum_jf_j(X)^2=I$, put
$E_j=f_j(X)-f_j(u)$, $Q=\sum_jE_j^2$. Completeness gives
$\sum_jf_j(u)E_j=-Q/2$, hence
\begin{equation}
 \sum_j f_j(X)\rho f_j(X)-\rho
 =\sum_jE_j\rho E_j-\tfrac12(Q\rho+\rho Q).
 \label{qcp-C-eq:soft-instrument-identity}
\end{equation}
For a pure input its trace distance disturbance is at most
$(\langle Q\rangle+\sqrt{\langle Q^2\rangle})/2$.

Choose filters $(f_L,f_M,f_R)$ equal to the left, middle and right unit
vectors outside the transition intervals $[d,2d]$, $[1-2d,1-d]$.
Rotate linearly from zero to $\pi/2$ on each transition, giving speed
$\ell/d$, $\ell=\pi/2$, and $\sum_j\operatorname{TV}(f_j^2)=4$.
Apply the filters to the quantum arrival coordinate, assigning vacuum
value one and retaining labels and states. The delay is $\min(W_z,1-u)$, where
\begin{equation}
 W_z\le_{\rm st}\epsilon^2Y_\chi,\qquad
 Y_\chi=\chi+\operatorname{Exp}(1/4).
 \label{qcp-C-eq:waiting-envelope}
\end{equation}
because the post rise intensity is at least $s^2/4$.
All moments of $Y_\chi$ are uniformly finite on $I$. Hence
\begin{equation}
 \bar\delta_{\epsilon,d}\le C_I
 \left(\frac{\epsilon^4}{d}+\frac{\epsilon^6}{d^2}\right).
 \label{qcp-C-eq:endpoint-disturbance}
\end{equation}
A delay $W$ crosses transitions on length at most $2d+2W$.
Integrating $Q\le\ell^2W^2/d^2$ gives
$C_I(\epsilon^4/d+\epsilon^6/d^2)$. For the square root term, the
transition interiors have the same bound, while the exterior strips
contribute at most $\epsilon^6d^{-2}\int_0^\infty
[\mathbb E(Y_\chi^4\mathbf1_{Y_\chi\ge r})]^{1/2}dr$.
The exponential envelope bounds this integral uniformly, including vacuum.

\subsection{Posterior weights and endpoint costs}
Keep the normalized conditional states and compare their label weights
$p_j(u,z)$ with $p_j(u,0)$. Monotone coupling of the waiting survival
probabilities $e^{-\tau_zA_\chi(x)}$, $\tau_z=(1+a\epsilon z)^2$, gives
$\mathbb E|W_z-W_0|\le\epsilon^2|\tau_z-1|C_\chi$,
$C_\chi=\chi^2+16$. Estimate the initial segment of width $\chi$ directly, then use
$A_\chi(x)=x-\nu_F\chi$ for $\tau\ge1/4$.
The bounded variation translation inequality gives
\begin{equation}
 \int_0^1\frac12\sum_j|p_j(u,z)-p_j(u,0)|du
 \le\bar\eta_\epsilon
 :=2C_\chi(2aL\epsilon^3+a^2L^2\epsilon^4).
 \label{qcp-C-eq:endpoint-prior}
\end{equation}
For $\pi_j=\int p_j(u,0)du$, let $\mathsf F_j$ have position density
$p_j(u,0)/\pi_j$, the original amplitude prior and normalized states.
The forward labelled comparison and reverse label removal have errors
$\bar\eta$ and $\bar\delta+\bar\eta$. Lemma~\ref{qcp-C-lem:transfer} gives
\begin{equation}
 \left|\mathcal G_\lambda(\mathsf E)
       -\sum_j\pi_j\mathcal G_\lambda(\mathsf F_j)\right|
 \le B_\lambda(\bar\delta+2\bar\eta).
 \label{qcp-C-eq:actual-endpoint-disintegration}
\end{equation}
On the left branch, $u<2d$. Splitting the field at $2d+w$ and freezing the
amplitude in the window leaves a constant suffix mode. The contracted
vectors satisfy
\begin{equation}
 \|v_0\|
 \left\|\frac{v_z}{\|v_z\|}-\frac{v_0}{\|v_0\|}\right\|
 \le2\|v_z-v_0\|.
 \label{qcp-C-eq:conditional-vector-bound}
\end{equation}
The error weighted by the reference probability is at most
$4aLd\sqrt{2d+w}$; a zero reference probability contributes zero mass.
The product experiment has zero gap, so two comparisons through the
identity channel give
\begin{equation}
 \pi_L\mathcal G_\lambda(\mathsf F_L)
          \le8B_\lambda aLd\sqrt{2d+w}.
 \label{qcp-C-eq:left-endpoint}
\end{equation}
On the right branch, a near optimal location rule is paired with an amplitude
report of zero. Passing to the actual weights costs at most
$10L^2\bar\eta$ in the risk of the zero report and in the optimum.
Reveal $u$ and complete the rule by reporting zero on the complement of
the branch. Lemma~\ref{qcp-C-lem:uniform-amplitude-information}
with QFI $4a^2(1-u)$ bounds the full record amplitude improvement by
$(8/15)a^2L^4(1-u)$. Integrate over $[1-4d,1]$;
earlier right labels have probability at most
$4\epsilon^2e^{-(2d-w)/(4\epsilon^2)}$. Thus
\begin{equation}
 \pi_R\mathcal G_\lambda(\mathsf F_R)
 \le(1-\lambda)\left\{
 \frac{64}{15}a^2L^4d^2+
 4L^2\epsilon^2e^{-(2d-w)/(4\epsilon^2)}
 +10L^2\bar\eta_\epsilon\right\}.
 \label{qcp-C-eq:right-endpoint}
\end{equation}
\subsection{Return to the middle synthesis}
For $d=\epsilon^{3/2}$, the disturbance is $O(\epsilon^{5/2})$ and the
weight change is $O(\epsilon^3)$. The weighted left and right gaps are
$O(\epsilon^{9/4})$, $O(\epsilon^3)$. The conditional gaps are
nonnegative, so restoring the middle weights gives
\begin{equation}
 \mathcal G_\lambda(\mathsf E)
       =\mathcal G_\lambda(T_{\rm ph})+o(\epsilon^2).
 \label{qcp-C-eq:physical-middle-return}
\end{equation}
Here $T_{\rm ph}$ is C7's exact arrival synthesis with $k=f_M$;
source extension adds the waiting tail. The endpoint rules give
$o(\epsilon^2)$ excess for each task at fixed interior weights and
combine with the middle rule in C9, uniformly on compact positive
$\chi$ intervals.

\section{Comparison of marginal excess risks}
\label{qcp-C-proof-centered-location-regret-transfer}
\label{qcp-C-sec:centered-comparison}
Use $\mathcal F_\delta v=(2\pi\delta)^{-1/2}\int e^{-ipu/\delta}v(u)du$,
$P_\delta=-i\delta\partial_u$, $m_1(p)=m_{\chi,1}(p)>0$ from C2,
$f=\sqrt{m_1}$, and
$\gamma(p,q)=f(p)f(q)(1-\cos K(p-q))/[\pi K(p-q)^2]$,
$g(p)=\int\gamma(p,q)dq$.
Write $\langle P_\delta\rangle=(I+P_\delta^2)^{1/2}$, and let $M_E$
denote multiplication by the coefficient field $E(t)$.
For $f\asymp(1+|p|)^{-4}$,
\begin{equation}
 \Lambda_{\delta}=2(1-g/m_1)(P_{\delta})\otimes I,\qquad
 g(p)/m_1(p)\le C(1+p^2).
 \label{qcp-C-eq:source-dual}
\end{equation}
The corresponding positive source form is
\begin{equation}
 L_{x,\delta}=2\min\{1,|x-u|/(K\delta)\}-\Lambda_{\delta}\ge0
 \label{qcp-C-eq:positive-source-form}
\end{equation}
Conjugating by $f(P_\delta)$ and the report phase gives twice the
nonnegative graph Laplacian. The bound on $g/m_1$ follows from
convolution of $(1+|p|)^{-4}$ with $C(1+|p-q|)^{-2}$.
For $T\langle P_\delta\rangle$ Hilbert--Schmidt, define
\begin{equation}
 \mathfrak p_J(T)=r_{\loc}(J;T)-\operatorname{Tr}(T\Lambda_{\delta}T^{\dagger}),
 \qquad
 d_{\delta}(R)^2=2\|R\|_{\HS}^2+
              \|R(-\Lambda_{\delta})_+^{1/2}\|_{\HS}^2.
 \label{qcp-C-eq:centered-form}
\end{equation}
The positive source form writes $\mathfrak p_J(T)$ as a squared
direct integral row norm. Outcome partitions and domain approximation
extend it to every complete POVM. The triangle inequality and
$0\le\ell\le2$ give
\begin{equation}
 0\le \mathfrak p_J(R)\le d_{\delta}(R)^2,\qquad
 |\sqrt{\mathfrak p_J(T)}-\sqrt{\mathfrak p_J(S)}|\le d_{\delta}(T-S)
 \le C\|(T-S)\langle P_{\delta}\rangle\|_{\HS}.
 \label{qcp-C-eq:root-stability}
\end{equation}
Taking near minimizers in both directions gives, with $\mathfrak p^\star(T)=\inf_J\mathfrak p_J(T)$,
\begin{equation}
 |\sqrt{\mathfrak p^\star(T)}-\sqrt{\mathfrak p^\star(S)}|\le d_{\delta}(T-S),\qquad
 e_{\loc}(J;T)=\mathfrak p_J(T)-\mathfrak p^\star(T).
 \label{qcp-C-eq:actual-centered-regret}
\end{equation}
Thus each excess is centered at its own marginal optimum.
\begin{qcpClemma}[Centered quadratic comparison]\label{qcp-C-lem:centered}
Fix \(0<\lambda<1\). Let \(T_\epsilon,S_\epsilon\) have the same source targets,
losses and class of complete paired measurements, with their finite
unnormalized masses retained. Assume:
\begin{enumerate}
\item \(d_{\delta}(T_\epsilon-S_\epsilon)=o(\epsilon)\);
\item for a common reference \(T_0\) and a complete rule \(J_0\),
      \(\mathfrak p_{J_0}(T_0)=0\) and
      \(d_{\delta}(T_\epsilon-T_0)+d_{\delta}(S_\epsilon-T_0)=O(\epsilon)\);
\item uniformly over all paired rules with \(|Y|\le L\),
      \(|\sqrt{e_{\amp}(J;T_\epsilon)}
         -\sqrt{e_{\amp}(J;S_\epsilon)}|=o(\epsilon)\);
\item at least one model admits a complete common rule with weighted
      excess \(O(\epsilon^2)\).
\end{enumerate}
Then
\begin{equation}
 \mathcal G_\lambda(T_\epsilon)
       -\mathcal G_\lambda(S_\epsilon)=o(\epsilon^2).
 \label{qcp-C-eq:centered-gap}
\end{equation}
Uniform hypotheses and moduli on a known design set give the uniform
conclusion there.
\end{qcpClemma}
\begin{proof}
The reference gives $\mathfrak p^\star(T),\mathfrak p^\star(S)=O(\epsilon^2)$.
A feasible $O(\epsilon^2)$ gap at an interior weight bounds both excesses.
Root stability and hypothesis 3 transfer these bounds between the models.

For a gap optimizer within $o(\epsilon^2)$ of the infimum,
$d_\delta(T-S)(\sqrt{\mathfrak p_J(T)}+\sqrt{\mathfrak p_J(S)})=o(\epsilon^2)$
controls the location forms and, by the same estimate at marginal
minimizers, their infima. Squaring the amplitude root bound and
interchanging $T,S$ completes the comparison.
\end{proof}
Hypothesis 3 follows from a common bounded self-adjoint predictor.
With $\Omega_T=TT^\dagger$, $\Omega_S=SS^\dagger$, assume
\begin{equation}
 \|A_T-\widetilde A_\epsilon\|_{\Omega_T}
 +\|A_S-\widetilde A_\epsilon\|_{\Omega_S}=o(\epsilon),
 \quad
 \|T-S\|_{\HS}=o(\epsilon),\quad
 \sup_\epsilon\|\widetilde A_\epsilon\|<\infty,
 \label{qcp-C-eq:common-predictor}
\end{equation}
The square identity and triangle inequality give
\begin{equation}
 |\sqrt{e_{\amp}(J;T)}-\sqrt{e_{\amp}(J;S)}|
 \le o(\epsilon)+(L+\|\widetilde A_\epsilon\|)\|T-S\|_{\HS}
 =o(\epsilon).
 \label{qcp-C-eq:amplitude-root}
\end{equation}
C1's source defect bound verifies these estimates in C5 and C7.

\section{Scalar experiment and a common readout}
\label{qcp-C-proof-slow-residual-regret}
\label{qcp-C-proof-nonstationary-amplitude-benchmark}
\label{qcp-C-sec:scalar-experiment}
Use C3's middle filter $k$ with $d=\epsilon^{3/2}$.
Then $\|k\|_2^2=1-3d$, $\|k'\|_2^2=\pi^2/(4d)$, and
\begin{equation}
 c(t)=a\sqrt{1-\operatorname{clip}_{[d,1-d]}(t)-w},\qquad
 \|c'\|_2^2=\frac{a^2}{4}\log\frac{1-d-w}{d-w}.
 \label{qcp-C-eq:actual-suffix}
\end{equation}
The source is $L^2(du)\otimes L^2(\nu)$, $\nu=\Unif[-1/2,1/2]$;
$Z$ multiplies by $z$ and $D(t)=k(t)V_{c(t)}$.
Using the positive spectrum in \eqref{qcp-C-eq:packet-spectrum}, define
\begin{align}
 m_\tau(p)=m_{\chi,\tau}(p),\nonumber\\
 f=\sqrt{m_1},\qquad
 q=a\left.\partial_\tau\log m_\tau\right|_{\tau=1},
 \qquad f_z=\sqrt{m_{(1+a\epsilon z)^2}},\nonumber\\
 \Phi_\epsilon(z,p)=\frac{f_z(p)}{f(p)}>0.
 \label{qcp-C-eq:scalar-spectra}
\end{align}
The source multiplier $\Phi_\epsilon(Z,P_\delta)$ defines the real
synthesis and reference
\begin{equation}
 T_\epsilon^{\rm r}=T_0\Phi_\epsilon(Z,P_\delta),\qquad
 T_0=\sqrt{2\pi \delta}\,M_D f(P_{\delta}).
 \label{qcp-C-eq:real-synthesis}
\end{equation}
Put $\Omega_{\rm r}=T_\epsilon^{\rm r}(T_\epsilon^{\rm r})^\dagger$.
Both syntheses have the finite mass $m=\|k\|_2^2$ because $\int f_z^2=1$.
Their losses are $2\min(1,|x-u|/(K\delta))$ and $(y-z)^2$.

\subsection{Location optimality and amplitude prediction}
Coordinate readout followed by $A_{c(t)}$ is complete. Multiplying the
positive graph slack by $k(t)|c(t)z\rangle$ and integrating over $z$ gives
\begin{equation}
 \mathfrak p_{J_0}(T_0)=0,\qquad
 r_{\loc}^\star(T_0)=m\,r_F(\chi).
 \label{qcp-C-eq:ordered-optimality}
\end{equation}
Applied to each positive $f_z$, the same graph argument establishes exact
optimality of coordinate readout over all complete measurements.
For $\overline A=M_{A_c}$, $\overline K=M_{K_c}$, the reference defect
equals $\sqrt{2\pi\delta}M_D[\overline K,f(P_\delta)]$ and satisfies
\begin{equation}
 \|\overline A T_0-T_0Z-T_0\overline K\|_{\HS}
 \le 2(2L+\kappa_{\mathrm{src}})\frac{L}{\sqrt3}\,
        \delta\|k\|_\infty\|c'\|_2
       \left(\int|f'(p)|^2dp\right)^{1/2}
 =o(\epsilon).
 \label{qcp-C-eq:zero-source-defect}
\end{equation}
Indeed, the squared kernel norm integrates
$k(t)^2|\alpha(x)|^2\|V_{c(t)}(K_{c(t)}-K_{c(t-\delta x)})\|_{\HS}^2$,
$\alpha=\mathcal F^{-1}f$. Apply C1's weighted lift bound and
$\|c(\cdot)-c(\cdot-\delta x)\|_2\le\delta|x|\|c'\|_2$.

For $b(p)=\int e^{iyp}\mu(dy)$, the translation representation and
Minkowski's inequality give
\begin{equation}
 \sqrt{2\pi \delta}\,
 \|[b(P_{\delta}),M_E]f(P_{\delta})\|_{\HS}
 \le\|f\|_2\int
       \|E(\cdot+\delta y)-E(\cdot)\|_{L^2(\HS)}\,|\mu|(dy).
 \label{qcp-C-eq:fourier-reordering}
\end{equation}
For $E=D$, $\delta\|D'\|_2=O(\epsilon^{5/4})$.
For residual coefficients, remove an endpoint strip and truncate the
Fourier measure, then apply C1's weighted and strong continuity.
C10 makes these moduli uniform.
Intensity Taylor expansion gives the source multiplier
$\Phi_\epsilon(Z,P_\delta)=I+\epsilon Zq(P_\delta)+O(\epsilon^2)$ and hence
\begin{equation}
 T_\epsilon^{\rm r}=T_0+\epsilon T_1+O_{\HS}(\epsilon^2),
 \qquad T_1=\sqrt{2\pi \delta}\,M_{DZ}q(P_{\delta})f(P_{\delta}).
 \label{qcp-C-eq:scalar-expansion}
\end{equation}
Define the amplitude and source predictors by
\begin{equation}
 A_{\delta}=M_{A_c}+\epsilon\{q(P_{\delta}),M_{B_c}\},\qquad
 K_{\delta}=M_{K_c}+\frac\epsilon2\{q(P_{\delta}),M_{K_{1,c}}\}.
 \label{qcp-C-eq:scalar-predictors}
\end{equation}
Here $B_c$ is C1's response. The bounded predictors are self-adjoint
and skew-adjoint, respectively. The zeroth order defect is
\eqref{qcp-C-eq:zero-source-defect}; reordering the first order terms gives
$2B_cV_c+V_c[K_c,Z]-V_cK_{1,c}=0$. The remaining
$\epsilon o(1)+O(\epsilon^2)$ error and the source defect lemma imply
\begin{equation}
 \|A_{\delta}T_\epsilon^{\rm r}
        -T_\epsilon^{\rm r}Z-T_\epsilon^{\rm r}K_{\delta}\|_{\HS}
       =o(\epsilon),\qquad
 \|A_{\delta}-A_\epsilon\|_{\Omega_{\rm r}}
       =o(\epsilon).
 \label{qcp-C-eq:actual-amplitude-predictor}
\end{equation}

\subsection{A complete upper readout}
Coordinate readout followed by $A_c$ defines a location optimal rule
$M_0$ satisfying
$YJ=J\overline A$, $\overline A=M_{A_c}$. Its amplitude calibration is
\begin{equation}
 r_{\amp,\rm ref}(\epsilon)=\int k(t)^2\mathcal A_L(c(t))\,dt,\qquad
 \mathcal A_L(c)=L^2/3-\operatorname{Tr}(\Omega_cA_c^2).
 \label{qcp-C-eq:actual-amplitude-calibration}
\end{equation}
Boundedness of $B_c,q$ gives
$\|\overline A-A_\delta\|_{\Omega_{\rm r}}=O(\epsilon)$.
Equation~\eqref{qcp-C-eq:actual-amplitude-predictor} and the square identity
then give, uniformly on $I$,
\begin{equation}
 e_{\amp}(M_0;T_\epsilon^{\rm r})
 =\|\overline A-A_\epsilon\|_{\Omega_{\rm r}}^2
 \le C_{\amp}\epsilon^2+o(\epsilon^2).
 \label{qcp-C-eq:scalar-common-cost}
\end{equation}
The centering uses the finite-$\epsilon$ amplitude optimum.
The identity $\tr(\Omega_cB_c)=\mathcal A_L(c)$ gives C6's response.

\subsection{Finite positions and conditional weights}
Use $\mathcal F^{-1}b(x)=(2\pi)^{-1/2}\int_{\mathbb R}e^{ipx}b(p)\,dp$
and $\alpha_0=\mathcal F^{-1}f$. For the real synthesis, put
$\alpha_{\epsilon,z}=\mathcal F^{-1}f_z$; for the phase inclusive synthesis
in \eqref{qcp-C-eq:phase-inclusive-source}, use
$\alpha_{\epsilon,z}=\mathcal F^{-1}(f_ze^{i\epsilon zr})$ instead.
First use $u\in\mathbb R$ with measure $du\,\nu(dz)$ and extend $k$ by zero.
Let
$p_\epsilon(u,z)=\int_{\mathbb R}k(t)^2\delta^{-1}
|\alpha_{\epsilon,z}((t-u)/\delta)|^2\,dt$,
and define $p_0(u)$ by the same expression with $\alpha_0$.
These column masses weight the normalized vectors on $[0,1]$.
Uniformly square integrable momentum derivatives give packet moments
of every fixed order. Since $\operatorname{supp}k\subset[d,1-d]$,
the added mass is at most $C_r m(\delta/d)^{2r}$. With $r=3$, its
contribution to each bounded risk is $O(\epsilon^3)$.
Clip reports to $[0,1]$.
With $W_1$ denoting the first Wasserstein distance, changing conditional
weights while retaining normalized vectors costs
\begin{equation}
 \int_0^1|p_\epsilon(u,z)-p_0(u)|\,du
 \le\operatorname{TV}(k^2)\,\delta\,
  W_1(|\alpha_{\epsilon,z}|^2,|\alpha_0|^2)
 =O(\epsilon \delta).
 \label{qcp-C-eq:scalar-prior-fee}
\end{equation}
Equal packet masses and the weighted first moment derivative bound
give the $W_1$ estimate by coupling unmatched mass through zero;
$\operatorname{TV}(k^2)=2$. Retain the reference mass
$\pi_0=\int_0^1p_0$ at quadratic precision. These are the scalar column
weights; C3 uses the physical arrival posterior.

\section{Lower bound from rigidity and response}
\label{qcp-C-proof-nonstationary-full-class-lower}
\label{qcp-C-sec:scalar-lower}
\begin{qcpClemma}[Rigidity--response for the scalar synthesis]
\label{qcp-C-thm:core-rigidity}
Let $T=T_\epsilon^{\rm r}$, $\Omega=TT^\dagger$ and $m=\tr\Omega>0$, with actual
amplitude optimum $A$. Fix a known unitary $U$ on the record space and a known
measurable phase $g$, $|g(x)|=1$. Assume
\begin{equation}
 U\Omega U^\dagger\le C\Omega,\qquad
 \|(g(X)J-JU)T\|_{\HS}\le\kappa_{\rm rig}\sqrt{e_\loc(J)}
 \label{qcp-C-eq:core-rigidity-hypotheses}
\end{equation}
for every complete paired rule $J$. For any bounded Hermitian predictor
with $\|\widetilde A-A\|_{\Omega}\le\eta_\amp$, define
$\widetilde\Delta=\tr\Omega(\widetilde A-U^\dagger\widetilde A U)$.
For fixed $0<\lambda<1$, the optimized gap satisfies
\begin{equation}
 \begin{gathered}
 \alpha=2L\kappa_{\rm rig},\qquad \beta=1+\sqrt C,\qquad
 \rho=[|\widetilde\Delta|/\sqrt m-\beta\eta_\amp]_+,\\
 \mathcal G_\lambda\ge
 \frac{\rho^2}{\alpha^2/\lambda+\beta^2/(1-\lambda)}.
 \end{gathered}
 \label{qcp-C-eq:core-response-bound}
\end{equation}
\end{qcpClemma}
\begin{proof}
Set $\Delta=\tr\Omega(A-U^\dagger AU)$ and $O=J^\dagger YJ$.
The amplitude square bounds the error of replacing $A$ by $O$ in the
unshifted and shifted states by $\sqrt{me_\amp}$ and $\sqrt{mCe_\amp}$.
Since $\|g(X)JT\|_{\HS}=\|JUT\|_{\HS}=\sqrt m$ and
$g(X)^\dagger Yg(X)=Y$, the reported response is at most
$2L\sqrt m\,\|(g(X)J-JU)T\|_{\HS}$. Location rigidity therefore gives
\[
 |\Delta|/\sqrt m\le\alpha\sqrt{e_\loc}+\beta\sqrt{e_\amp}.
\]
The predictor error gives
$\rho\le\alpha\sqrt{e_\loc}+\beta\sqrt{e_\amp}$.
Weighted Cauchy--Schwarz bounds its square by
$[\alpha^2/\lambda+\beta^2/(1-\lambda)]
[\lambda e_\loc+(1-\lambda)e_\amp]$; take the infimum over paired rules.
\end{proof}

For $\mathcal D_\gamma(v)=\iint\gamma(p,q)\|v(p)-v(q)\|^2dpdq$, use
the normalized scalar spectrum $m_1$ from C5 and assume
\begin{equation}
 \mathcal D_\gamma(v)\ge c_D
  \int m_1(p)\|v(p)-\bar v\|^2dp,\qquad
 \bar v=\int m_1(p)v(p)\,dp ,
 \label{qcp-C-eq:graph-coercivity}
\end{equation}
and $m_1(p-\zeta)\le C_\zeta m_1(p)$. Subtracting $\bar v$, squaring the
triangle inequality and changing variable give
\begin{equation}
 \int m_1(p)\|v(p+\zeta)-v(p)\|^2dp
 \le2(1+C_\zeta)c_D^{-1}\mathcal D_\gamma(v).
 \label{qcp-C-eq:shift-control}
\end{equation}
For the real scalar synthesis, the location excess is the graph form
of the output rows after removing the report phase, at each latent $z$.
The momentum shift $U_\zeta$ multiplies record functions by
$e^{i\zeta t/\delta}$, with report phase $g(x)=e^{i\zeta x/\delta}$.
Summing over rows and $z$ gives
\begin{equation}
 \kappa_{\mathrm{rig}}^2=2(1+C_\zeta)/c_D .
 \label{qcp-C-eq:actual-location-rigidity}
\end{equation}
Coordinate optimality centers the excess at the marginal optimum for
all complete rules. Apply the graph and shift bounds to
$m_{\chi,(1+a\epsilon z)^2}$ with C2/C10's common constants, then multiply
by $M_{k|cz\rangle}$ and integrate to obtain average state domination.

Use C5's predictor $A_\delta$, with weighted error $o(\epsilon)$.
The zeroth order means of $A_c$ cancel because both shifted spectra
integrate to one. Evaluating the first order term at the reference
spectrum and reordering gives
\begin{align}
 \widetilde\Delta_{\epsilon,\zeta}
 &=2\epsilon\int k(t)^2
        \operatorname{Tr}(\Omega_{c(t)}B_{c(t)})\,dt
       \int[m_1(p)-m_1(p-\zeta)]q(p)\,dp+o(\epsilon)\nonumber\\
 &=-2\epsilon Q_\zeta(\chi)\,\overline{\mathcal A}_L+o(\epsilon),
 \quad
 Q_\zeta(\chi)=\int m_1(p)q(p+\zeta)\,dp .
 \label{qcp-C-eq:actual-scalar-response}
\end{align}
Normalization gives $\int m_1q=0$, and
$\tr(\Omega_cB_c)=\mathcal A_L(c)$ identifies Note J's positive suffix
average $\overline{\mathcal A}_L$. C10 gives one finite shift with
$\inf_IQ_\zeta>0$, so the rigidity bound yields
\begin{align}
 \liminf_{\epsilon\downarrow0}\epsilon^{-2}\mathcal G_\lambda(T_\epsilon^{\rm r})
 &\ge\frac{2Q_\zeta(\chi)^2\overline{\mathcal A}_L^2}
 {(1+C_\zeta)\{4L^2/(\lambda c_D)+1/(1-\lambda)\}}>0,
 \label{qcp-C-eq:scalar-lower}\\
 \limsup_{\epsilon\downarrow0}\epsilon^{-2}\mathcal G_\lambda(T_\epsilon^{\rm r})
 &\le(1-\lambda)C_{\amp}<\infty .
 \label{qcp-C-eq:scalar-upper}
\end{align}
Use $(1+\sqrt C)^2\le2(1+C)$ for the lower bound and C5's rule for
the upper bound, retaining the finite mass and marginal infima.

\section{Reduction of the retained packet}
\label{qcp-C-proof-physical-quadratic-upper}
\label{qcp-C-proof-nonstationary-normal-alignment}
\label{qcp-C-proof-physical-amplitude-response}
\label{qcp-C-sec:physical-return}
Keep the weights $du\,\nu(dz)$ of C3's middle branch and write its exact
packet as
\begin{equation}
 v_{\epsilon;u,z}(t)
 =k(t)\delta^{-1/2}\beta_{\epsilon,z}((t-u)/\delta)
                       \otimes|c(t)z\rangle,
 \quad
 \beta_{\epsilon,z}(y)=\sqrt{c_\chi}\,
                         B_{a\epsilon z}(c_\chi y),
 \quad c_\chi=\sqrt{1+\chi}.
 \label{qcp-C-eq:retained-middle}
\end{equation}
Extension to the full source line gives mass $m=\|k\|_2^2$;
negative positions require delay $d-u$ and have exponentially small mass.
For C2's frame $U_0$, the squared commutator norm is
\begin{equation}
 \int d\xi\,dp\,
 \|\widehat E(\xi)\|_{\HS}^2m_1(p)
 \|[U_0(p+\delta\xi)-U_0(p)](0,w_0(p))\|^2.
 \label{qcp-C-eq:ordinary-fourier-commutator}
\end{equation}
Ordinary derivatives bound it by $C_U^2\delta^2\|E'\|_2^2$.
For $E=D$,
\begin{equation}
 \|D'\|_2^2=\|k'\|_2^2+
             (L^2/3)\int k^2|c'|^2
 =O(d^{-1}+\log(1/d)).
 \label{qcp-C-eq:physical-commutator}
\end{equation}
For real normalized coherent states, the cross derivative vanishes.
The same estimate holds for $DZ$ and the window multipliers.

\subsection{Removal of the orthogonal tangent}
The packet has the weighted expansion
\begin{equation}
 U_0(p)(0,\widehat\beta_{\epsilon,z}(p))
 =f(p)\left[e_0+\epsilon z\{
            e_0(q(p)+ir(p))+a\widetilde w_1(p)\}\right]+O(\epsilon^2),
 \quad \widetilde w_1(p)\perp e_0 .
 \label{qcp-C-eq:framed-packet-expansion}
\end{equation}
Here $\widetilde w_1$ and its required derivatives are bounded, and
$r=-(ac_\chi/2)\partial_p\log m_1$.
Let $e(t)=1-\operatorname{clip}_{[d,1-d]}(t)-w$,
$D_1(p)=|\widetilde w_1(p)\rangle\langle e_0|$, and let $b$ denote
annihilation in the residual mode. Define the known skew-adjoint generator
\begin{equation}
 G_\epsilon=D_1(P_{\delta})M_{e^{-1/2}}\otimes b
          -M_{e^{-1/2}}D_1(P_{\delta})^{\dagger}\otimes b^{\dagger}.
 \label{qcp-C-eq:normal-generator}
\end{equation}
The generator preserves $b^\dagger b+\Pi_\perp$,
$\Pi_\perp=I-|e_0\rangle\langle e_0|$, and has sector norm at most
$C\sqrt N/\sqrt{d-w}$. Sectorwise exponentiation defines a unitary.
The suffix cancels the orthogonal tangent because
\begin{equation}
 b|c(t)z\rangle=a\sqrt{e(t)}\,z|c(t)z\rangle .
 \label{qcp-C-eq:coherent-eigenvector}
\end{equation}
For $T_0^{\rm out}=\sqrt{2\pi\delta}\,e_0\otimes M_Df(P_\delta)$,
commutator control and the full vector expansion give
\begin{equation}
 U_0(P_{\delta})T_{\rm ph}=T_0^{\rm out}+\epsilon T_1+R_\epsilon,
 \quad
 \|R_\epsilon\|_{\HS}\le C(\epsilon^2+\delta\sqrt{E_d}),
 \quad E_d=d^{-1}+\log(1/d).
 \label{qcp-C-eq:physical-expansion}
\end{equation}
The normal first order part differs from $G_\epsilon T_0^{\rm out}$
by at most $C\delta\sqrt{E_d}$, while the parallel part is $z(q+ir)$.
To control the exponential remainder on both sources, use
$W_\delta=\langle P_\delta\rangle$ and
\begin{equation}
 \|G_\epsilon^2T_0^{\rm out}W\|_{\HS}
       +\|G_\epsilon T_1W\|_{\HS}\le C(d-w)^{-1/2},
 \qquad W\in\{I,W_\delta\}.
 \label{qcp-C-eq:source-number-bound}
\end{equation}
Annihilation cancels the inverse strength on $e_0$; return through
the normal block leaves one inverse square root, with
$\|b^\dagger|cz\rangle\|^2=1+c^2z^2$.
For $W=W_\delta$, each scalar profile gains $\langle p\rangle$.
C2's $p^{-8}$ tail and vector derivative bounds preserve square
integrability. The commutator estimate with weight $(1+p^2)m_1(p)$
therefore gives $\|R_\epsilon W_\delta\|_{\HS}\le
C(\epsilon^2+\delta\sqrt{E_d})$.
On each conserved sector, Taylor's integral formula gives
\[
 \begin{aligned}
 &\bigl[(e^{-\epsilon G_\epsilon}-I+\epsilon G_\epsilon)T_0^{\rm out}
       +\epsilon(e^{-\epsilon G_\epsilon}-I)T_1\bigr]W_\delta\\
 &\qquad=\int_0^\epsilon e^{-tG_\epsilon}
       \bigl[(\epsilon-t)G_\epsilon^2T_0^{\rm out}
             -\epsilon G_\epsilon T_1\bigr]W_\delta\,dt.
 \end{aligned}
\]
Summing the sector Hilbert--Schmidt bounds and setting
$\mathcal U_\epsilon=e^{-\epsilon G_\epsilon}U_0(P_\delta)$ gives
\begin{equation}
 \|(\mathcal U_\epsilon T_{\rm ph}-T_\epsilon^{\rm p})
                   \langle P_{\delta}\rangle\|_{\HS}
 \le C\left(\epsilon^2+\frac{\epsilon^2}{\sqrt{d-w}}
                         +\delta\sqrt{E_d}\right)
 =O(\epsilon^{5/4})=o(\epsilon).
 \label{qcp-C-eq:weighted-normal-alignment}
\end{equation}
Here the synthesis retaining the scalar phase is
\begin{equation}
 T_\epsilon^{\rm p}
   =e_0\otimes T_\epsilon^{\rm r}R_\epsilon^{\rm src},
 \qquad
 R_\epsilon^{\rm src}=\exp(i\epsilon Zr(P_{\delta})).
 \label{qcp-C-eq:phase-inclusive-source}
\end{equation}

\subsection{Amplitude prediction and the gap return}
The source phase commutes with $Z$, so
\begin{equation}
 T_\epsilon^{\rm p}Z^j(T_\epsilon^{\rm p})^{\dagger}
 =e_0e_0^{\dagger}\otimes
       T_\epsilon^{\rm r}Z^j(T_\epsilon^{\rm r})^{\dagger}.
 \label{qcp-C-eq:phase-amplitude-moments}
\end{equation}
Put $T=\mathcal U_\epsilon T_{\rm ph}$, $S=T_\epsilon^{\rm p}$,
$\Omega_T=TT^\dagger$, $\Omega_S=SS^\dagger$, and let $A_T,A_S$ be their
amplitude optima. Extend $A_\delta$ boundedly and self-adjointly to the
unused directions and conjugate its source lift by the source phase.
The $o(\epsilon)$ synthesis error and C1's defect bound give
\begin{equation}
 \|A_T-A_{\delta}\|_{\Omega_T}
 +\|A_S-A_{\delta}\|_{\Omega_S}
 =o(\epsilon).
 \label{qcp-C-eq:physical-common-predictor}
\end{equation}
Both syntheses are within weighted $O(\epsilon)$ of $T_0^{\rm out}$.
Its coordinate rule has zero centered location form and
$o(\epsilon)$ amplitude defect; root stability gives both excesses
$O(\epsilon^2)$. Equations~\eqref{qcp-C-eq:common-predictor}--\eqref{qcp-C-eq:amplitude-root}
and \eqref{qcp-C-eq:weighted-normal-alignment} verify
Lemma~\ref{qcp-C-lem:centered}, yielding
\begin{equation}
 \mathcal G_\lambda(T_{\rm ph})
       -\mathcal G_\lambda(T_\epsilon^{\rm p})=o(\epsilon^2).
 \label{qcp-C-eq:physical-to-phase-gap}
\end{equation}
The same comparison transfers individual excesses. C9 removes the
phase and combines this bound with C3's endpoint return.

\section{Location bounds for Hermitian multipliers}
\label{qcp-C-proof-hermitian-multiplier-location-lemma}
\label{qcp-C-sec:hermitian-multiplier}
Let $D(t)$ be bounded Hilbert--Schmidt, $m=\int\|D(t)\|_{\HS}^2dt<\infty$,
$g_0=\int g<\infty$, and $C(p)=C(p)^\dagger$, $\|C\|\le B$.
Put $Q_\epsilon=I+\epsilon C$,
$T_\epsilon=\sqrt{2\pi\delta}M_Df(P_\delta)Q_\epsilon(P_\delta)$.
Write $m_f=f^2$; $m$ is the base mass.
Let $r_\omega(p)$ be the row of $JM_D\mathcal F_\delta^\dagger$ times
$e^{-ipx(\omega)/\delta}$, including multiplicities.
Put $\Delta r_\omega(p,q)=r_\omega(p)-r_\omega(q)$ and
\begin{equation}
 \mathfrak D_J=2\pi \delta
   \int\!\!\int\gamma(p,q)
      \int\|\Delta r_\omega(p,q)\|_{\HS}^2\,d\omega\,dp\,dq .
 \label{qcp-C-eq:dirichlet-rule}
\end{equation}
Completeness gives
\begin{equation}
 2\pi \delta\int g(p)\int\|r_\omega(p)\|_{\HS}^2d\omega\,dp=mg_0 .
 \label{qcp-C-eq:complete-row-mass}
\end{equation}
\begin{qcpClemma}[Fourth order reference error]\label{qcp-C-lem:hermitian}
For \(|\epsilon|B\le1/8\), let \(r_{\loc,\rm ref}\) be the coordinate readout
location risk on \(T_\epsilon\). Every complete paired rule \(J\) satisfies
\begin{align}
 |r_{\loc}(J;T_\epsilon)-r_{\loc,\rm ref}-\mathfrak D_J|
 &\le(2|\epsilon|B+\epsilon^2B^2)\mathfrak D_J
       +2\epsilon^2B^2\sqrt{mg_0\,\mathfrak D_J},
       \label{qcp-C-eq:hermitian-risk}\\
 0\le r_{\loc,\rm ref}-r_{\loc}^\star(T_\epsilon)
 &\le2B^4mg_0\,\epsilon^4,
       \label{qcp-C-eq:fourth-order-reference}\\
 \mathfrak D_J
 &\le4e_{\loc}(J;T_\epsilon)+16B^4mg_0\,\epsilon^4.
       \label{qcp-C-eq:near-opt-dirichlet}
\end{align}
For direct sums of labeled branches with a common \(B\), these bounds use
the sum of the base masses in place of \(m\).
\end{qcpClemma}
\begin{proof}
Let $S_\epsilon=\{Q_\epsilon(p),Q_\epsilon(q)\}/2$,
$\mathsf K_\epsilon=[Q_\epsilon(p),Q_\epsilon(q)]/2$.
Then $S_\epsilon\ge I/2$,
$\|S_\epsilon-I\|\le2|\epsilon|B+\epsilon^2B^2$ and
$\|\mathsf K_\epsilon\|\le\epsilon^2B^2$.
Expanding the triangular Fourier benefit and symmetrizing gives
\begin{align}
 r_{\loc}(J;T_\epsilon)-r_{\loc,\rm ref}
 ={}&2\pi \delta\int\gamma(p,q)\int
  (\Delta r_\omega(p,q))S_\epsilon(p,q)
  (\Delta r_\omega(p,q))^{\dagger}\,d\omega\,dp\,dq \nonumber\\
 &-4\pi \delta\,\operatorname{Re}\int\gamma(p,q)\int
 (\Delta r_\omega(p,q))\mathsf K_\epsilon(p,q)r_\omega(q)^{\dagger}
 \,d\omega\,dp\,dq .
 \label{qcp-C-eq:hermitian-exact}
\end{align}
The coordinate reference trace is
$4\pi\delta M_D[m_f(p)Q_\epsilon^2-\int\gamma S_\epsilon dq](P_\delta)M_D^\dagger$.
Subtract $r(q)\mathsf K_\epsilon r(q)^\dagger$, whose real part vanishes,
then apply Cauchy--Schwarz and \eqref{qcp-C-eq:complete-row-mass}.

In particular, $r_\loc(J)-r_{\loc,\rm ref}\ge\mathfrak D_J/2
-2\epsilon^2B^2\sqrt{mg_0\mathfrak D_J}\ge-2B^4mg_0\epsilon^4$.
Take near minimizers for the second bound. Young's inequality gives
$e_\loc(J)\ge\mathfrak D_J/4-4B^4mg_0\epsilon^4$.
Summing branches before Cauchy--Schwarz gives the direct sum version.
\end{proof}
For a positive diagonal multiplier $F_\epsilon(Z,p)=I+O(\epsilon)$,
the graph kernel is positive at each latent amplitude. Coordinate
readout is therefore optimal, with
\begin{equation}
 e_{\loc}^{\rm real}(J)=(1+O(\epsilon))\mathfrak D_J.
 \label{qcp-C-eq:real-multiplier}
\end{equation}
On either model's near optimal set, $e_\loc$ and $\mathfrak D_J$ are
$O(\epsilon^2)$, so
\begin{equation}
 e_{\loc}^{\rm Hermitian}(J)
       -e_{\loc}^{\rm real}(J)=o(\epsilon^2)
 \label{qcp-C-eq:hermitian-real-regret}
\end{equation}
for every rule in these sets.

\section{Removal of the scalar phase}
\label{qcp-C-proof-nonstationary-phase-gauge-return}
\label{qcp-C-sec:phase-return}
The phase can be removed at quadratic precision:
\begin{equation}
 \mathcal G_\lambda(T_\epsilon^{\rm p})
       -\mathcal G_\lambda(T_\epsilon^{\rm r})=o(\epsilon^2).
 \label{qcp-C-eq:phase-gap-return}
\end{equation}
Take a real sine and cosine partition $\sum_jf_j^2=1$ of support length at
most $2b$, $b=\epsilon^{1/2}$, with $\sum_j|f_j'|^2\le C/b^2$.
Apply the corresponding complete instrument to the output coordinate,
retaining the labels. The disturbance is bounded using the instrument
identity from C3 and uniform second and fourth packet moments:
\begin{equation}
 O(m\delta^2/b^2)=O(\epsilon^3)
 \label{qcp-C-eq:soft-label-disturbance}
\end{equation}
for both models. Applying or discarding the instrument compares the
three risks. The direct sum fields $D_j=kf_jV_c$ retain their weights and obey
\begin{equation}
 \sum_j\|D_j\|_{L^2(\HS)}^2=m,\qquad
 \sum_j\|D_j'\|_{L^2(\HS)}^2
  \le C(d^{-1}+b^{-2}+\log(1/d)).
 \label{qcp-C-eq:summed-label-fields}
\end{equation}
The cross terms vanish since $\sum_jf_jf_j'=0$.
Reorder the summed field before applying Cauchy--Schwarz.

\subsection{Output gauge and source lift}
Use the known bounded output gauge
\begin{equation}
 H_{\delta}=\tfrac12\{r(P_{\delta}),M_{A_c}\},\qquad
 U_\epsilon=e^{-i\epsilon H_{\delta}}.
 \label{qcp-C-eq:known-phase-gauge}
\end{equation}
The Taylor expansion, $A_cV_c=V_cZ+V_cK_c$ and Fourier reordering give
\begin{align}
 U_\epsilon T_{\epsilon,j}^{\rm p}
 =T_{0,j}
   &+\epsilon\sqrt{2\pi \delta}\,M_{D_jZ}q(P_{\delta})f(P_{\delta})\nonumber\\
   &-i\epsilon\sqrt{2\pi \delta}\,M_{D_jK_c}r(P_{\delta})f(P_{\delta})
     +o_{\HS,\langle P_{\delta}\rangle}(\epsilon).
 \label{qcp-C-eq:phase-gauged-expansion}
\end{align}
The identity $irD_jZ-irA_cD_j=-irD_jK_c$ cancels the parallel phase.
The summed derivatives and tight Fourier measures bound the remainder
by weighted $o(\epsilon)$, including C4's factor $\langle P_\delta\rangle$.
The exponential remainder is $O(\epsilon^2)$.
For one known point $t_j$ in each support, $K_j=K_{c(t_j)}$ satisfies
\begin{equation}
 \sum_j\int
 \|D_j(t)(K_{c(t)}-K_j)\|_{\HS}^2dt\le Cb m .
 \label{qcp-C-eq:freeze-source-only}
\end{equation}
C1's weighted lift bound and the square root modulus of $c$ give this
estimate, so freezing costs $O(\epsilon\sqrt b)=O(\epsilon^{5/4})$.
Set
\begin{equation}
 C_j(p)=q(p)Z-i r(p)K_j=C_j(p)^{\dagger},\qquad
 \widehat T_j=T_{0,j}[I+\epsilon C_j(P_{\delta})].
 \label{qcp-C-eq:hermitian-comparison-source}
\end{equation}
These multipliers have the common bound
$B=L\|q\|_\infty+\kappa_{\rm src}\|r\|_\infty$.
Suffixes, observables and masses still vary with the record coordinate.

\subsection{Comparing amplitude and location excesses}
The gauge conjugates the amplitude optimum. To verify C4's predictor
criterion for $A_\delta$, it remains to bound
\begin{equation}
 [H_{\delta},M_{A_c}]=\tfrac12[r(P_{\delta}),M_{A_c^2}].
 \label{qcp-C-eq:gauge-predictor-commutator}
\end{equation}
Expanding the difference of squares, with $A_cV_c=V_c(Z+K_c)$, gives
\begin{equation}
 \|(A_{c'}^2-A_c^2)V_c\|_{\HS}
 \le2L\|(A_{c'}-A_c)V_c\|_{\HS}
       +C\|V_{c'}-V_c\|_{\HS}
 \le C'|c-c'|.
 \label{qcp-C-eq:squared-predictor-difference}
\end{equation}
Fourier reordering makes this commutator $o(1)$ on the syntheses,
hence $o(\epsilon)$ with its prefactor. The first order predictor adds
$O(\epsilon^2)$. Transfer the defect through freezing and apply C1.
The coordinate optimal reference and C7's quadratic rule verify the
remaining hypotheses of Lemma~\ref{qcp-C-lem:centered}, giving
\begin{equation}
 \mathcal G_\lambda(T^{{\rm p},{\rm lab}})
       -\mathcal G_\lambda(\widehat T)=o(\epsilon^2).
 \label{qcp-C-eq:phase-to-Hermitian}
\end{equation}
The real labelled synthesis is $T_{0,j}\Phi_\epsilon(Z,P_\delta)$,
using C5's positive source multiplier. Let $T_0=\bigoplus_jT_{0,j}$,
$Q=\bigoplus_j(I+\epsilon C_j)$ and
$\Phi=\bigoplus_j\Phi_\epsilon(Z,P_\delta)$.
Both right factors have uniformly bounded inverses and
$\|Q-\Phi\|=O(\epsilon)$.
For a rule $J$ within $o(\epsilon^2)$ of either gap optimum, define
\[
 E_J=(YJ-JA_\delta)T_0.
\]
The amplitude root excess is $O(\epsilon)$. Replacing its optimum by
$A_\delta$ costs $o(\epsilon)$; removing the invertible right factor
therefore gives $\|E_J\|_{\HS}=O(\epsilon)$ and
\[
 \|E_JQ-E_J\Phi\|_{\HS}
 \le\|E_J\|_{\HS}\|Q-\Phi\|=O(\epsilon^2).
\]
Restoring each model's own amplitude optimum yields
\[
 \left|\sqrt{e_\amp(J;\widehat T)}
       -\sqrt{e_\amp(J;T^{{\rm r},{\rm lab}})}\right|
 =o(\epsilon).
\]
Both roots are $O(\epsilon)$, so the squared excesses differ by
$o(\epsilon^2)$. C8 gives the same precision for location, since the
summed Dirichlet form is $O(\epsilon^2)$. Comparing near optimizers in
both directions yields
\begin{equation}
 \mathcal G_\lambda(\widehat T)
       -\mathcal G_\lambda(T^{{\rm r},{\rm lab}})
       =o(\epsilon^2).
 \label{qcp-C-eq:Hermitian-to-real}
\end{equation}
Combine the $O(\epsilon^3)$ instrument disturbance with C3's endpoint
and posterior bounds and C7's packet comparison to obtain
\begin{equation}
 \begin{aligned}
 \mathcal G_\lambda(\mathsf E_{\rm original})
       &=\mathcal G_\lambda(T_{\rm ph})+o(\epsilon^2)\\
       &=\mathcal G_\lambda(T_\epsilon^{\rm p})+o(\epsilon^2)\\
       &=\mathcal G_\lambda(T_\epsilon^{\rm r})+o(\epsilon^2).
 \end{aligned}
 \label{qcp-C-eq:original-real-return}
\end{equation}
Pulling C5's rule through the gauges, labels and endpoint branches
gives excesses over the original marginal optima:
\begin{equation}
 e_{\loc}=o(\epsilon^2),\qquad
 e_{\amp}=O(\epsilon^2).
 \label{qcp-C-eq:original-common-rule}
\end{equation}

\section{Uniform inverse square bounds}
\label{qcp-C-proof-compact-scale-theorem}
\label{qcp-C-sec:uniform-design}
C2 gives $Q=\sup_{\chi,p}|q_\chi(p)|<\infty$ and
$R=\sup_{\chi,p}|r_\chi(p)|<\infty$ under the crossover assumptions.

\subsection{Coercivity and a uniform response witness}
For $R_0=\pi/K$, $P_K=2\pi/K$, use
\begin{equation}
 c_0=\frac{\sqrt{c_1/c_2}}
           {2\pi K(1+R_0^2)(1+R_0)^4},\quad
 G_0=\frac{Kc_2}{3\pi},\quad
 c_D=\frac{c_0^2P_K}{4G_0}>0 .
 \label{qcp-C-eq:uniform-coercivity-constant}
\end{equation}
For $|r|\le R_0$, $\gamma(p,r)\ge c_0m(p)(1-\cos K(p-r))$,
$g(r)\le G_0$. The kernel $\int_{-R_0}^{R_0}\gamma(p,r)\gamma(r,q)/g(r)dr$ is at least
$c_0^2P_Km(p)m(q)/(2G_0)$: the period integral of the two cosine factors
is at least $P_K/2$. Its Dirichlet form is at most four times the original,
whereas the $m\otimes m$ form is twice the variance. This proves C6's
Hilbert valued coercivity with $c_D$.
The tail comparison also gives
\begin{equation}
 m_{\chi,\tau}(p-\zeta)\le C_\zeta m_{\chi,\tau}(p),
 \qquad C_\zeta=(c_2/c_1)(1+|\zeta|)^8.
 \label{qcp-C-eq:uniform-shift-domination}
\end{equation}
It also gives the source domain bound $g/m\le C_I(1+p^2)$.
Differentiating the spectral tail yields
\begin{equation}
 q_\chi(p)=a\left\{1+
 \frac{D_\chi e^{-D_\chi}\cos(d_\chi p)}
             {1-e^{-D_\chi}\cos(d_\chi p)}\right\}
                  +v_\chi(p),\qquad
 v_\chi(p)=O_I(|p|^{-1}).
 \label{qcp-C-eq:uniform-score-tail}
\end{equation}
Since $181/924\le D_\chi\le181/231$, its periodic part is at least
$a[1-D_\chi/(e^{D_\chi}+1)]\ge2a/(2+D_\chi)
\ge\underline q=462a/643>0$.
Choose a common $P_0$ with $|v_\chi(p)|\le\underline q/2$ for $|p|\ge P_0$.
The common polynomial tail permits a known finite $\zeta>P_0$ such that
\begin{equation}
 2P_0c_2(1+\zeta-P_0)^{-8}
 \le\frac{\underline q}{4(\underline q/2+Q)}.
 \label{qcp-C-eq:one-response-shift}
\end{equation}
This bounds the mass of $|p+\zeta|<P_0$. The score is at least
$\underline q/2$ outside this region and at least $-Q$ everywhere, so
\begin{equation}
 Q_\zeta(\chi)=\int m_\chi(p)q_\chi(p+\zeta)\,dp
 \ge\underline q/4=:Q_*>0 \qquad(\chi\in I).
 \label{qcp-C-eq:uniform-positive-response}
\end{equation}

\subsection{Uniform reordering and return errors}
The phase score splits as
\begin{equation}
 r_\chi(p)=
 -\frac{ac_\chi e^{-D_\chi}d_\chi\sin(d_\chi p)}
          {2\{1-e^{-D_\chi}\cos(d_\chi p)\}}+w_\chi(p).
 \label{qcp-C-eq:uniform-phase-tail}
\end{equation}
The differentiated tails and positivity on compact sets give common
$H^1$ bounds for both remainders. Expanding $(1-E\cos)^{-1}$,
$E\le e^{-181/924}<1$, gives periodic Fourier measures with bounded
variation and exponential tails. For an $H^1$ remainder,
Cauchy--Schwarz bounds its inverse transform $L^1$ norm by its $H^1$ norm
and bounds its tail by $CY^{-1/2}\|v'\|_2$. Hence
\begin{equation}
 \sup_{\chi\in I}\|\mu_\chi\|_{\rm TV}<\infty,\qquad
 \sup_{\chi\in I}|\mu_\chi|(\{|y|>Y\})\longrightarrow0 .
 \label{qcp-C-eq:uniform-Fourier-tightness}
\end{equation}
The shifted score and moving atomic frequencies have the same bounds.
Remove an endpoint strip, choose a common Fourier cutoff, then let
$\epsilon\to0$. C1's weighted and strong continuity gives common
$o(1)$ moduli in this order.

After rescaling the window to $[0,1]$, C2's sector majorants and C3's
waiting bounds are uniform on $I$. The labelled fields satisfy
$\delta(\sum_j\|D_j'\|_2^2)^{1/2}=O_I(\epsilon^{5/4})$ and eventually
$d-w\ge d/2$, giving alignment and freezing errors $O_I(\epsilon^{5/4})$.
The phase and Hermitian generators are bounded by $LR$ and
$LQ+\kappa_{\rm src}R$. With the common
$2+\operatorname{Exp}(1/4)$ moment envelope,
Sections~\ref{qcp-C-sec:endpoints}, \ref{qcp-C-sec:scalar-experiment}
and \ref{qcp-C-sec:phase-return} give uniform $o_I(\epsilon^2)$ returns
for endpoints, posterior weights, source tails and soft labels. Hence
\begin{equation}
 \sup_{\chi\in I}
 |\mathcal G_\lambda(\mathsf E_{\rm original})
        -\mathcal G_\lambda(T_\epsilon^{\rm r})|
            =o(\epsilon^2).
 \label{qcp-C-eq:uniform-return}
\end{equation}

\subsection{Conclusion of the uniform bound}
The response expansion and dominated residual integrals are uniform.
Define
\begin{equation}
 a_-
 =\frac{2Q_*^2\overline{\mathcal A}_L^2}
 {(1+C_\zeta)\{4L^2/(\lambda c_D)+1/(1-\lambda)\}},
 \qquad
 a_+=(1-\lambda)C_{\amp} .
 \label{qcp-C-eq:uniform-order-constants}
\end{equation}
C5/C6 and the uniform return give
$a_- -r(s)\le s^2\G_\lambda(s,\chi/s^2)\le a_++r(s)$, where the common
$r(s)\to0$ bounds predictor, response, centered and residual integral
moduli together with $\epsilon^{1/4}+\epsilon$. Choose $s_0$ with
$r\le a_-/2$; then $c_-=a_-/2$, $c_+=a_++a_-/2$ prove
Theorem~\ref{qcp-M-thm:uniform-first-cost}(ii).
These comparisons transfer C5's rule and both excesses uniformly on $I$.
Continuity and compactness separate the leading risks from zero and
the prior risks. C11 transfers the bounds to finite records.

\section{Transfer to finite records}
\label{qcp-C-sec:finite-record-return}
Let $\phi$ be a known normalized $C^3$ curve in fixed finite dimension,
$q_0=\|(I-\rho(0))\phi'(0)\|^2>0$.
For a known monotone profile and fixed $L>0$, use
\begin{equation}
 \theta_i=\frac{s_n+az}{\sqrt{nq_0}}F((i/n-u)/w_n),
 \qquad u\sim\Unif[0,1],\quad z\sim\Unif[-L,L],
 \label{qcp-C-eq:core-site-parameters}
\end{equation}
Put $H_n=s_n+aL$, $b_n=1-|\langle\phi(0),\phi(s_n/\sqrt{nq_0})\rangle|^2$,
$\tau_n=nb_n$, $k_w(t)=t/\sqrt{1+wt}$. The physical and coherent position
coefficients are $k_{w_n}(\tau_n)/K$, $k_{w_n}(s_n^2)/K$.
The target is $u$, with reports in $\mathcal X_n$ and physical loss
\eqref{qcp-M-eq:physical-rise-loss}.

Apply Theorem~\ref{qcp-C-thm:channels}'s cell isometry and completed
inverse to each $\theta_i$. The cell error $C_\phi q_0\theta_i^2$
gives product fidelity error at most $C_\phi H_n^2/\sqrt n$.
For any monotone $f:[0,1]\to[0,1]$, its right step approximation obeys
$\|f_n-f\|_2^2\le\|f_n-f\|_1\le[f(1)-f(0)]/n\le1/n$ by telescoping.
With $H=H_n$, these curve and grid channels give
\begin{equation}
 \varepsilon_{n,H}
 \le\frac{H_n+C_\phi H_n^2}{\sqrt n}.
 \label{qcp-C-eq:core-physical-channel}
\end{equation}
The bound is uniform in $w_n>0$.
The curve expansion gives $\tau_n/s_n^2=1+O(s_n/\sqrt n)$, and
the logarithmic derivative of $k_w$ lies in $[1/2,1]$.
Clipping and grid rounding therefore give
\begin{equation}
 d_{\loc,n}
 =2\frac{|\tau_n-s_n^2|}{\max\{\tau_n,s_n^2\}}
       +\frac{2k_{w_n}(\tau_n)}{Kn}
 =O\!\left(\frac{s_n}{\sqrt n}+\frac{s_n^2}{n}\right),
 \qquad d_{\amp,n}=0.
 \label{qcp-C-eq:core-calibration}
\end{equation}
Apply Lemma~\ref{qcp-C-lem:transfer} with $B_\loc=2$, $B_\amp=4L^2$,
$B_\lambda=2\lambda+4L^2(1-\lambda)$ and
$E_n=B_\lambda\varepsilon_{n,H}+\lambda d_{\loc,n}$:
\begin{align}
 |r_{j,n}^\star-r_{j,\infty}^\star|
   &\le B_j\varepsilon_{n,H}+d_{j,n},\label{qcp-C-eq:core-marginal-return}\\
 |R_{\lambda,n}^\star-R_{\lambda,\infty}^\star|
   &\le E_n,
     \label{qcp-C-eq:core-joint-return}\\
 |\mathcal G_{\lambda,n}-\mathcal G_{\lambda}|
   &\le2E_n.
     \label{qcp-C-eq:core-gap-return}
\end{align}
Pulling back a common rule adds at most $2B_j\varepsilon_{n,H}+2d_{j,n}$
to each excess.
For $s_n\to\infty$, the channel, calibration and report errors are
$O(s_n^2/\sqrt n)$, $O(s_n/\sqrt n)$ and $O(s_n^2/n)$, uniformly in
$w_n>0$. They vanish for $s_n=o(n^{1/4})$, transferring J's leading
limits whenever $w_n\to0$.

Multiplication by $s_n^2$ gives $O(s_n^4/\sqrt n)$,
$O(s_n^3/\sqrt n)$ and $O(s_n^4/n)$, which vanish for
$s_n=o(n^{1/8})$. Thus the joint and marginal optima, the gap and both
common rule excesses transfer at $o(s_n^{-2})$ along every known path
$\chi_n\in[1/2,2]$, $w_n=\chi_n/s_n^2$.

\endgroup

\pdfbookmark[1]{IV. Fixed width limit and cubic width cost}{qcp-SM-part4}
\section*{IV. The fixed width limit, loss dependence and cubic width cost}
\begingroup
\qcpresetvolume{S}
\qcparticlelayout
\setlength{\emergencystretch}{3em}

\renewcommand{\thesection}{S}
\renewcommand{\theequation}{S.\arabic{equation}}
\section{Limits at fixed width and local quadratic compatibility}
Fix $0<w<1$, $L=1/2$ and the quintic profile from
\eqref{qcp-M-eq:rise-family}; in particular, $F(r)>0$ for $0<r\le1$.
Use the main priors, reports and calibration, with $\nu=\Unif[-L,L]$ and losses
\begin{equation}
 \delta_{s,w}=\frac{\sqrt{1+s^2w}}{s^2},\quad
 \ell_\loc^{(w)}(x,u)=2\min\{1,|x-u|/(K\delta_{s,w})\},\quad
 \ell_\amp(y,z)=(y-z)^2.\label{qcp-S-eq:loss}
\end{equation}
The loss bounds are $B_\loc=2$, $B_\amp=4L^2$ and
$B_\lambda=2\lambda+4L^2(1-\lambda)$. Write $\G$ for the full record gap
and $g$ for the covariant fiber gap. We use
$D(\ket f,\ket g)\le\|f-g\|_2$.

\subsection{Truncated tangents and global comparison}
For $u<1$, set $e=1-u$, $r=\min\{e/w,1\}$,
$I(r)=\int_0^rF'^2$, $J(r)=\int_0^rF^2$. Then
\begin{align}
 G_{11}(u)&=I(r),&G_{12}(u)&=-\tfrac{\sqrt w}{2}F(r)^2,&
 G_{22}(u)&=(e-w)_++wJ(r),\label{qcp-S-eq:gram}\\
 \kappa&=\sqrt{G_{11}},&d&=aG_{12}/\kappa,&c&=\sqrt{a^2G_{22}-d^2}.
 \label{qcp-S-eq:normal}
\end{align}
These are the Gram entries of $g_{\loc,u}=-w^{-1/2}F'((t-u)/w)$ and
$g_{\amp,u}=F((t-u)/w)$. Cauchy--Schwarz gives $IJ>F(r)^4/4$:
equality would make $F'$ proportional to $F$, contrary to $F(0)=0<F(r)$.
Thus $G_{11}>0$ and $\det G>0$, including truncated rises.
Gram--Schmidt and Fock factorization give
\begin{equation}
 \rho^u_{v,z}=\proj{\kappa(u)v+d(u)z}\otimes\proj{c(u)z}.
 \label{qcp-S-eq:fiber}
\end{equation}
The continuous coefficients satisfy $\kappa,c>0$ and $d<0$ on $[0,1)$.
The complete covariant gap is $g_\lambda(\kappa,d,c)$, defined as in
\eqref{qcp-M-eq:gap}; set
$g_{\lambda,w}(u)=g_\lambda(\kappa(u),d(u),c(u))$.

Partition $[0,1]$ into cells of length $b_s=1/N_s$, centered at $u_j$.
Good cells end at or before $1-\epsilon$; the others have mass at most
$\epsilon+b_s$. The comparison $\mathsf F_s^\epsilon$ emits the true good
label and $\ket{v g_{\loc,u_j}+azg_{\amp,u_j}}$, with
$v=s(u-u_j)/\sqrt w$, and emits vacuum with a bad label elsewhere.
Taylor and pilot attenuation errors are $T_s=O_w(sb_s^2+b_s)$ and
$H_s=O_w(t_s(sb_s+1))$. For $q_\epsilon=F(\epsilon/(2w))>0$, the label error is
\begin{equation}
 P_s\le\min\{1,C_{w,\epsilon}(N_s+1)
 (s^{-1}+(s\sqrt{t_s})^{-1})\}.
 \label{qcp-S-eq:label-error}
\end{equation}
The bounds are uniform over good cells and $|z|\le L$.
For the joint risk, take $r_j^\star=R_\lambda^\star$ and $B_j=B_\lambda$.
\begin{qcpStheorem}[Global comparison]\label{qcp-S-thm:channels}
Channels independent of $u,z$ in both directions have errors averaged
over the prior satisfying
\begin{equation}
 \varepsilon_{\to}\le\epsilon+b_s+P_s+T_s+H_s,\qquad
 \varepsilon_{\leftarrow}\le\epsilon+b_s+T_s,\label{qcp-S-eq:errors}
\end{equation}
with both errors capped at one. For each marginal and the joint risk,
$-B_j\varepsilon_{\leftarrow}\le r_j^\star(\mathsf E_s)-r_j^\star(\mathsf F_s^\epsilon)
\le B_j\varepsilon_{\to}$; hence
$|\G_\lambda(\mathsf E_s)-\G_\lambda(\mathsf F_s^\epsilon)|
\le B_\lambda(\varepsilon_{\to}+\varepsilon_{\leftarrow})$.
At fixed $\epsilon$, the choices $t_s=s^{-1/2}$ and
$N_s=\lceil s^{2/3}\rceil$ give $P_s=O_\epsilon(s^{-1/12})$,
$T_s=O(s^{-1/3})$, $H_s=O(s^{-1/6})$, while
\begin{equation}
 V_s=sb_s/(2\sqrt w)\longrightarrow\infty.\label{qcp-S-eq:wide-prior}
\end{equation}
\end{qcpStheorem}
\begin{proof}
Homodyne measurement of the pilot fraction gives
$X=\sqrt{t_s}(s+az)A_w(u)+Z/2$, where $Z\sim N(0,1)$,
$A_w(u)=\int_0^1F((t-u)/w)dt$ and $A_w'(u)=-F((1-u)/w)$.
Clip to the range on $[0,1-\epsilon/2]$ and invert $A_w$. On a good cell,
$|\widehat u-u|\le q_\epsilon^{-1}(aL/s+|Z|/(2s\sqrt{t_s}))$;
integration over cell boundary strips gives $P_s$.
The conditionally independent retained field is displaced by
$-\sqrt{1-t_s}sF((t-u_j)/w)$ at a reported good label.
Taylor's remainder and $1-\sqrt{1-t_s}\le t_s$ give $T_s,H_s$.
Wrong labels and bad positions contribute their probabilities.
The reverse channel restores the supplied displacement with error $T_s$
and assigns fixed states to bad labels. Pulling back complete POVMs gives
the three risk bounds; their difference gives the gap bound.
\end{proof}

\subsection{Wide priors and integration over position}
In a fixed fiber, take independent $v\sim\Unif[-V,V]$ and $z\sim\nu$,
with $c_K(x-v)=2\min\{1,|x-v|/K\}$. Translation by $t$ displaces the
first mode by $\kappa t$ and preserves amplitude reports.
Compactify the position reports by a failure point $\infty$ of cost two.
\begin{qcpSlemma}[Wide flat prior]\label{qcp-S-lem:cov}
As $V\to\infty$, the finite prior optimum converges to the complete
covariant risk $R_{\lambda,\mathrm{cov}}^\star(\kappa,d,c)$ for every
$0\le\lambda\le1$. Convergence is uniform on compact coefficient sets
with $\kappa,c>0$.
\end{qcpSlemma}
\begin{proof}
For the upper bound, replace covariant failure by a finite report and clip to
$[-V,V]$. For the lower bound, average $M_V$ with risk within $V^{-1}$ of the optimum:
$\overline M_V(B)=(2V)^{-1}\int_{-V}^VU_t^\dagger M_V(B+t)U_tdt$.
Its risk at zero equals the finite prior risk, and translation by $t$
changes a continuous test integral by at most $|t|\|f\|_\infty/V$.
Apply the compactness and translated-average argument of
Lemma~\ref{qcp-C-J-lem:canonical} to paired reports
$(\mathbb R\cup\{\infty\})\times[-L,L]$. It gives a complete covariant
limit with the limiting risk, since the cost field
$\lambda c_K(x)\Xi_0+(1-\lambda)(y^2\Xi_0-2y\Xi_1+\Xi_2)$
is trace-class valued and continuous on this compact space.
Here $\Xi_j=\int z^j\rho_{0,z}\nu(dz)$; the failure cost is included.
For uniformity, set $x=\kappa v$. A change in $\kappa$ changes the prior
by $|\kappa-\kappa'|/\max\{\kappa,\kappa'\}$ in total variation and
the loss by twice this amount. A change in $d,c$ changes the state
distance by at most $L\sqrt{(d-d')^2+(c-c')^2}$, uniformly in $x$.
Equicontinuity and a finite net give uniform convergence.
\end{proof}
\begin{qcpStheorem}[Global fixed width limit]\label{qcp-S-thm:global}
At fixed width, the original full priors and losses \eqref{qcp-S-eq:loss} give
\begin{align}
 \lim_{s\to\infty}R_{\lambda,s}^\star
 &=\int_0^1R_{\lambda,\mathrm{cov}}^\star(\kappa(u),d(u),c(u))du,
 \label{qcp-S-eq:global-risk}\\
 \lim_{s\to\infty}\G_{\lambda,s}&=\int_0^1g_{\lambda,w}(u)du.
 \label{qcp-S-eq:global-gap}
\end{align}
The first identity at $\lambda=1,0$ gives the marginal limits.
\end{qcpStheorem}
\begin{proof}
Clipping good cell reports gives local loss
$2\min\{1,q_s|\widehat v-v|/K\}$, where
$q_s=s\sqrt w/\sqrt{1+s^2w}\to1$; its distance from $c_K$ is at most
$2(1-q_s)$. Optimal risks add over the labels.
Lemma~\ref{qcp-S-lem:cov}, uniformly on $[0,1-\epsilon]$, turns these
sums into the fiber integral. The remaining contribution is bounded by
$B_\lambda(\epsilon+b_s)$. The channel bounds, followed by
$s\to\infty$ and $\epsilon\downarrow0$, prove the risk identity.
Its cases $\lambda=1,0$ give the marginal limits and hence the gap identity.
\end{proof}

\begin{qcpSproposition}[Residual amplitude optimum]\label{qcp-S-prop:amp}
The fiber amplitude optimum equals $\mathcal A_L(c(u))$, the Bayes risk of
the residual experiment
\begin{equation}
 z\longmapsto\ket{c(u)z},\qquad z\sim\nu.\label{qcp-S-eq:amp-mode}
\end{equation}
\end{qcpSproposition}
\begin{proof}
Averaging the first mode over $v$ gives the coherent mixture on
$[-\kappa V+dz,\kappa V+dz]$. Replacing it by the mixture centered at zero
costs at most $\min\{1,|dz|/(2\kappa V)\}$ in trace distance.
This replacement and discarding the first mode give
\begin{equation}
 0\le \mathcal A_L(c)-r_{\amp,V}^\star
 \le4L^2\mathbb E\min\{1,|dz|/(2\kappa V)\}
 \longrightarrow0.\label{qcp-S-eq:amp-finite-prior}
\end{equation}
Lemma~\ref{qcp-S-lem:cov} finishes the proof.
\end{proof}

\subsection{Complete covariant decisions and loss control}
\label{qcp-S-sec:covariant-decision}
On the support of $\Omega=\int\proj{cz}\nu(dz)$, set
$\psi_z=\ket{cz}$, $\Omega_j=\int z^j\proj{\psi_z}\nu(dz)$,
$\Omega A+A\Omega=2\Omega_1$ and
$\mathcal A=\tr\Omega_2-\tr(\Omega A^2)$, with $\|A\|\le L$.
Use $[Q,P]=i/2$ and $\varphi(p)=(2/\pi)^{1/4}e^{-p^2}$.
Covariant position POVMs admit an isometry field
\cite[Theorem~2.4.3 and Corollary~2.4.5]{qcp-S-toigo2005}.
For a complete covariant paired POVM with real location reports, take
a minimal joint Naimark dilation $(\widetilde{\mathcal H},\mathsf J,P)$.
Covariance makes
\[
 \widetilde U_tP(B\times C)\mathsf J\psi
 =P((B+t)\times C)\mathsf J U_t\psi
\]
an isometry on a dense span, defining a strongly continuous unitary
representation by compactly supported test functions.
The imprimitivity theorem identifies the dilation with
$L^2(\mathbb R,\mathcal K)$, with position multiplication and translation
shifts \cite[Theorem~1.2.2]{qcp-S-toigo2005}.
The amplitude PVM commutes with position and translations, so it is a
constant decomposable field $I\otimes F_\amp$.
Thus $Y=\int yF_\amp(dy)$ is self-adjoint with $\|Y\|\le L$, and the
momentum intertwiner is a measurable isometry field
$W(p):\mathcal H_\amp\to\mathcal K$. The location density is the squared norm of
\begin{equation}
 \sqrt{\kappa/\pi}\int e^{2i\kappa xp}\varphi(p)e^{-2idzp}W(p)\psi_zdp.
 \label{qcp-S-eq:cov-amplitude}
\end{equation}
Writing $\mathsf P_d(\xi)=\int e^{-2idz\xi}\proj{\psi_z}\nu(dz)$ and
$k_b(\xi)=b\operatorname{sinc}^2(b\xi)/\pi$, $b=\kappa K$, Fourier
transformation and Personick completion give
\begin{align}
 r_\loc(W)&=2-2\operatorname{Re}\int\varphi(p)\varphi(q)k_b(p-q)
 \tr[W(p)\mathsf P_d(p-q)W(q)^\dagger]dpdq,\label{qcp-S-eq:cov-loc-functional}\\
 r_\amp(W,Y)-\mathcal A&=\int\varphi(p)^2
 \|YW(p)-W(p)A\|_\Omega^2dp.
 \label{qcp-S-eq:cov-amp-excess}
\end{align}
These infima also include compactified failure reports. Disintegrate
their invariant effects as $F_p(dy)$ and partition $[-L,L]$ into bins
$C_j$ of diameter $\eta$, with $y_j\in C_j$.
The field $V_\eta(p)\psi=\bigoplus_jF_p(C_j)^{1/2}\psi$ and
$Y_\eta=\bigoplus_jy_jI$ preserve amplitude bin probabilities by
Plancherel. A real position report costs at most the failure cost two,
and amplitude rounding costs at most $4L\eta$.
Direct sum with the real report part and let $\eta\downarrow0$.

\begin{qcpStheorem}[Quadratic compatibility]\label{qcp-S-thm:quadratic-control}
For quadratic location loss $(\widehat v-v)^2$, the fiber optimum is
\begin{equation}
 r_2^\star=\frac{1/4+d^2\mathcal A_L(c)}{\kappa^2}.
 \label{qcp-S-eq:quadratic-value}
\end{equation}
Joint measurement of first mode $Q$ and residual mode $A$, reported as
$\widehat v=(X-dY)/\kappa$, $\widehat z=Y$, attains both marginal optima
in the complete covariant class, giving zero covariant gap for every weight.
\end{qcpStheorem}
\begin{proof}
By Plancherel, the quadratic risk equals
$(4\kappa^2)^{-1}\int\!\int\|\partial_p[\varphi e^{-2idzp}W\psi_z]\|^2\nu(dz)dp$.
Bounded $z$ and positivity of $\varphi$ on compact intervals give local
weighted Sobolev regularity for finite risk. Isometry cancels the
$\varphi'$ cross term, $\int|\varphi'|^2=1$, and the Sylvester equation gives
\begin{equation}
 \int\|(W'-2idzW)\psi_z\|^2\nu(dz)
 =4d^2\mathcal A+\|W'-2idWA\|_\Omega^2.
 \label{qcp-S-eq:kinetic-square}
\end{equation}
On spectral projections $P_N$ of $\Omega$, the compressed Sylvester
equation and skew-adjointness of $(WP_N)^\dagger(WP_N)'$ prove this identity.
Bounded $z,A$ control the cross terms in the Sobolev limit, yielding
\[
 r_2(W)-r_2^\star=\frac1{4\kappa^2}\int\varphi^2
       \|W'-2idWA\|_\Omega^2dp.
\]
The field $W=e^{2idpA}$ makes both defects zero and yields
$(Q-dA)/\kappa$. Here $X=\kappa v+dz+N$, with independent
$N\sim N(0,1/4)$, so the location error is $(N+d(z-Y))/\kappa$.
\end{proof}

\endgroup
\begingroup
\qcpresetvolume{P}
\qcparticlelayout
\setlength{\emergencystretch}{3em}

\renewcommand{\thesection}{P}
\renewcommand{\theequation}{P.\arabic{equation}}
\section{Local bounds needed for the cubic width cost}
\subsection{Scalar coercivity and loss response}
\label{qcp-P-app:scalar-loss}
Set $b=\kappa K$ and define
\[
 \varphi(p)=(2/\pi)^{1/4}e^{-p^2},\quad
 \mu(dp)=\varphi(p)^2dp,\quad
 k_b(x)=\frac b\pi\operatorname{sinc}^2(bx),\quad
 g_b(p,q)=\varphi(p)\varphi(q)k_b(p-q).
\]
For a real scalar field $\eta$, define
\begin{align}
 \mathcal D_b(\eta)&=\iint g_b(p,q)(\eta(p)-\eta(q))^2\,dp\,dq,\nonumber\\
 (\mathcal L_b\eta)(p)&=\int g_b(p,q)(\eta(p)-\eta(q))\,dq,\nonumber\\
 f_b(p)&=\int g_b(p,q)(p-q)^3\,dq.\label{qcp-P-eq:graph-definitions}
\end{align}
Then $\int f_b=0$. Replacing the squared difference by
$\|H(p)-H(q)\|_\Omega^2$ defines $\mathcal D_{b,\Omega}(H)$.
Choose representatives of mean zero under $\mu$ modulo constants.

\begin{qcpPlemma}[Gaussian graph coercivity]
\label{qcp-P-lem:graph-coercivity}
For every $b>0$, some $c_b>0$ satisfies
$\mathcal D_b(\eta)\ge c_b\operatorname{Var}_\mu \eta$.
The degree measure $m_b(p)dp$, with $m_b(p)=\int g_b(p,q)dq$,
also has a positive Poincare constant.
\end{qcpPlemma}
\begin{proof}
Define the two step symmetric kernel
\[
 g_b^{(2)}(p,q)=\int\frac{g_b(p,r')g_b(r',q)}{m_b(r')}\,dr'.
\]
This kernel has degree $m_b$ and form at most $4\mathcal D_b$.
With $r_b=\pi/(2b)$, boundedness of $k_b*\varphi$ and integration over
$[-r_b,r_b]$ give
\[
 g_b^{(2)}(p,q)\ge
 C_b
 \frac{\varphi(p)}{(|p|+r_b)^2}
 \frac{\varphi(q)}{(|q|+r_b)^2}.
\]
The sine square product averages to
$1/4+\cos(2b(p-q))/8\ge1/8$ over this period.
Boundedness of $\varphi(p)(|p|+r_b)^2$ makes this dominate
$\mu\otimes\mu$; its form is $2\operatorname{Var}_\mu\eta$.

The period estimate also gives $m_b(p)\gtrsim\varphi(p)/(1+p^2)$.
For the upper bound, split at $|q|=|p|/2$, use
$k_b(p-q)\le4/(\pi bp^2)$ on the inner part and Gaussian decay outside;
continuity covers bounded $p$. Hence $m_b(p)\asymp\varphi(p)/(1+p^2)$
and $\mu(dp)\le C m_b(p)dp$.
The same kernel bound dominates $m_b(p)m_b(q)$, proving degree coercivity.
\end{proof}

The equation $\mathcal L_b\chi_b=f_b$ admits a unique mean zero weak
solution in the Dirichlet space, characterized by
\begin{equation}
 \iint g_b(p,q)(\chi_b(p)-\chi_b(q))(\eta(p)-\eta(q))\,dp\,dq
 =2\int f_b(p)\eta(p)\,dp.
 \label{qcp-P-eq:response}
\end{equation}
The finite sixth kernel moment bounds the source in the Dirichlet norm.
Coercivity embeds the mean zero completion in $L^2(\mu)$, so Riesz
representation gives the solution.

\begin{qcpPlemma}[Positive response variance]
\label{qcp-P-lem:response-variance}
The response variance is finite, strictly positive and depends only on $b$:
\begin{equation}
 V_b:=\operatorname{Var}_\mu\chi_b>0.
 \label{qcp-P-eq:response-lower}
\end{equation}
\end{qcpPlemma}
\begin{proof}
Coercivity gives finite variance. The test field $\eta(p)=p$ belongs to the
Dirichlet space, and symmetry gives
\begin{equation}
 \int p f_b(p)dp
 =\tfrac12\iint g_b(p,q)(p-q)^4dpdq>0.
 \label{qcp-P-eq:Jfour}
\end{equation}
A constant $\chi_b$ would contradict this identity in the weak equation.
Thus $V_b>0$, and its defining objects depend only on $b$.
\end{proof}

\subsection{The complete local decision class}
\label{qcp-P-app:physical-loss}
\subsubsection{The uniform coherent residual}
Use $\nu=\Unif[-L,L]$, $L=1/2$, $\psi_z=|cz\rangle$, $0<c\le a$,
$\kappa=\kappa_F=\sqrt{10/7}$ and $K=9\log2$.
On the support of $\Omega=\int\rho_z\nu(dz)$, set
$\rho_z=\proj{\psi_z}$, $\Omega_j=\int z^j\rho_z\nu(dz)$ and
$\mathsf S_\Omega(X)=\{\Omega,X\}/2$.
C1 gives $\mathsf S_\Omega(A)=\Omega_1$ and $\|A\|\le L$.
For the complete covariant experiment
\begin{equation}
 |\kappa v+dz\rangle\otimes|cz\rangle,\qquad v\in\mathbb R,
 \label{qcp-P-eq:phys-experiment}
\end{equation}
the losses are $2\min\{1,|x-v|/K\}$ and $(y-z)^2$, $|y|\le L$.
Take $d\to0$ at fixed $c$ and write
$g_\lambda(d;c)=g_\lambda(\kappa_F,d,c)$.
Note S gives $\mathcal A=\tr\Omega_2-\tr\Omega A^2$, and
Proposition~\ref{qcp-C-prop:relative-residual} gives
\begin{equation}
 \mathsf B=\Omega^{-1/2}A\Omega^{1/2},\qquad \|\mathsf B\|\le R_A.
 \label{qcp-P-eq:relative-A}
\end{equation}
Define
\begin{align}
 \mathcal R_3&=2\int\{(zI-A)^3,\rho_z\}\nu(dz),\label{qcp-P-eq:phys-R3}\\
 C_3&=\mathsf S_\Omega^{-1}(\mathcal R_3),\qquad
 J_3=\|[A,C_3]\|_\Omega^2.\label{qcp-P-eq:phys-C3}
\end{align}
The moment form of the cubic source is
\[
 \mathcal R_3=4\Omega_3-6\{A,\Omega_2\}
 +\{A^3,\Omega\}+3A\{A,\Omega\}A.
\]
Here $\Omega_1=\{\Omega,A\}/2$ combines the ordered cubic products.
With $Z_j=\Omega^{-1/2}\Omega_j\Omega^{-1/2}$ and $\|Z_j\|\le L^j$, this gives
\begin{equation}
 \Omega^{-1/2}\mathcal R_3\Omega^{-1/2}
 =4Z_3-6(\mathsf{B}Z_2+Z_2\mathsf{B}^\dagger)
 +3(\mathsf{B}^2\mathsf{B}^\dagger+\mathsf{B}(\mathsf{B}^\dagger)^2)+\mathsf{B}^3+(\mathsf{B}^\dagger)^3.
 \label{qcp-P-eq:relative-R3}
\end{equation}
The relative source norm and $\|C_3\|$ are at most
$M_3=4(L^3+3R_AL^2+2R_A^3)$, by positivity of the inverse Sylvester map.
Complete bounded Hermitian operators in
$\langle X,Y\rangle_\Omega=\tr\Omega(XY+YX)/2$ to obtain
$\mathfrak h_\Omega$. Spectral compressions are dense, so
\begin{equation}
 \|[A,X]\|_\Omega\le(L+R_A)\|X\|_\Omega
 \label{qcp-P-eq:bounded-commutator}
\end{equation}
extends the commutator continuously. For $b=\kappa K$, set
\[
 H_*=-\chi_b C_3/3,\qquad
 \mathcal E_A(H)=\int\|[A,H(p)-\bar H]\|_\Omega^2\mu(dp),
 \quad\bar H=\int H\,d\mu.
\]
The vector-valued Dirichlet form is denoted by $\mathcal D_{b,\Omega}$,
with fields identified modulo constants.

\subsubsection{A common space for varying measurement dilations}
\label{qcp-P-sec:weighted-dilation}
Normalize $m_b(p)dp$ to a probability measure $\widehat m_b$.
For fields $Z:\mathcal H\to\mathcal K$, write
\[
 \|Z\|_{\mu,\Omega}^2
 =\int\|Z(p)\|_\Omega^2\mu(dp),\qquad
 \mathcal D_{b,\Omega}(Z)
 =\iint g_b(p,q)\|Z(p)-Z(q)\|_\Omega^2dpdq.
\]
Complete these fields in $L^2(\widehat m_b)$ and the Dirichlet seminorm.
Scalar coercivity, applied coordinatewise, gives the Hilbert--Schmidt
bounds and the embedding into $L^2(\mu)$.

\begin{qcpPlemma}[Weighted dilation compactness]
\label{qcp-P-lem:weighted-dilation-compactness}
Let $\Omega$ be faithful and trace class, with $\tr\Omega=1$.
Let $A=A^\dagger$, $\|A\|\le L$, and suppose
$\Omega^{-1/2}A\Omega^{1/2}$ has norm at most $R$.
Let $t_k\ne0$ tend to zero and let $V_k(p):\mathcal H\to\mathcal K_k$
be measurable isometries with
$\mathcal D_{b,\Omega}(V_k)=O(t_k^2)$.
After padding the separable output spaces and applying constant output
unitaries, a subsequence and a common isometry $J:\mathcal H\to\mathcal K$ satisfy
\[
 X_k=(V_k-J)/t_k\ \rightharpoonup\ X=JiH+N,
 \qquad J^\dagger N=0,
\]
in the ambient weighted space, with $H$ in the Hermitian weighted
completion and
\begin{equation}
 \liminf_k t_k^{-2}\mathcal D_{b,\Omega}(V_k)
 \ge\mathcal D_{b,\Omega}(H)+\mathcal D_{b,\Omega}(N).
 \label{qcp-P-eq:dilation-dirichlet-liminf}
\end{equation}
If, in addition, $Y_k=Y_k^\dagger$, $\|Y_k\|\le L$, and
$\|Y_kV_k-V_kA\|_{\mu,\Omega}=O(|t_k|)$, then a further subsequence
admits a constant Hermitian weighted correction $Y_\infty$ satisfying
\begin{equation}
 \liminf_k\frac{\|Y_kV_k-V_kA\|_{\mu,\Omega}^2}{t_k^2}
 \ge\int\|Y_\infty+i[A,H(p)]\|_\Omega^2\mu(dp)
 \ge\mathcal E_A(H).
 \label{qcp-P-eq:dilation-amplitude-liminf}
\end{equation}
Both lower bounds hold jointly for every nonnegative weighted sum.
\end{qcpPlemma}

\begin{proof}
Let $\overline V_k=\int V_k\,d\widehat m_b$. Since each $V_k$ is an
isometry, $0\preceq|\overline V_k|\preceq I$, and
\[
 \int\|V_k-\overline V_k\|_\Omega^2d\widehat m_b
 =\tr\Omega(I-\overline V_k^\dagger\overline V_k)=O(t_k^2).
\]
The estimate follows from graph coercivity. Extend the polar part of
$\overline V_k$ to an isometry $J_k$; then
\[
 \|J_k-\overline V_k\|_\Omega^2
 =\tr\Omega(I-|\overline V_k|)^2
 \le\tr\Omega(I-|\overline V_k|^2)=O(t_k^2).
\]
After padding output complements to countably infinite dimension,
constant unitaries identify $J_k$ with $J$ and conjugate $Y_k$.
The resulting bounded fields $X_k$ have a weak limit $JiH+N$,
with $J^\dagger N=0$.

To check that $H$ is Hermitian, use $V_k^\dagger V_k=I$:
\[
 J^\dagger X_k+X_k^\dagger J=-t_kX_k^\dagger X_k.
\]
For fixed eigenvectors of $\Omega$, the columns are bounded in
$L^2(\widehat m_b)$ and the right side tends to zero in $L^1$.
Testing against bounded scalar functions makes $J^\dagger X=iH$
skew-adjoint. Weak convergence of Dirichlet differences and orthogonality
of the output components give \eqref{qcp-P-eq:dilation-dirichlet-liminf}.

Right multiplication by $A$ is continuous in the weighted norm,
because $A\Omega^{1/2}=\Omega^{1/2}\mathsf{B}$ with $\|\mathsf{B}\|\le R$.
The triangle inequality, $\|Y_k\|\le L$, and the $L^2(\mu)$ bound
on $X_k$ therefore give
\[
 \|Y_kJ-JA\|_\Omega=O(|t_k|).
\]
Hence $Z_k=(J^\dagger Y_kJ-A)/t_k$ has a Hermitian weak limit
$Y_\infty\in\mathfrak h_\Omega$. Uniform operator bounds extend
$Y_kJ-JA\to0$ from the dense range of $\Omega^{1/2}$ to the whole space.

Set $F_k=J^\dagger Y_kJ-A$ and $\mathsf{E}_k=(I-JJ^\dagger)Y_kJ$.
Both operators converge strongly to zero and are uniformly bounded, with $F_k$
self-adjoint. Projection of the scaled amplitude defect onto
$J\mathcal H$ gives
\[
 Z_k+A J^\dagger X_k-J^\dagger X_k A
       +F_kJ^\dagger X_k+\mathsf{E}_k^\dagger(I-JJ^\dagger)X_k.
\]
The last two terms vanish weakly by moving $F_k,\mathsf E_k$ onto fixed
left vectors in finite rank tests, then using strong convergence, uniform
$L^2$ bounds and density. Projection therefore gives $Y_\infty+i[A,H]$
and the first amplitude inequality. The minimizing constant is
$-i[A,\int H\,d\mu]$ by \eqref{qcp-P-eq:bounded-commutator}.
Lower semicontinuity in the direct sum gives both bounds jointly.
\end{proof}

\begin{qcpPlemma}[Attainable location correction]
\label{qcp-P-lem:dense-dilation-fields}
Bounded Hermitian fields $H_n=-\chi_nC_3/3$ approximate
$H_*=-\chi_bC_3/3$ in the Dirichlet norm. The complete
readouts $W=e^{id^3H_n}e^{2idpA}$ attain $\Phi(H_n)$ at order $d^6$.
\end{qcpPlemma}
\begin{proof}
Coercivity places $\chi_b$ in $L^2(m_bdp)$ modulo constants.
Bounded smooth approximations converge in energy since
$\mathcal D_b(h)\le4\int m_bh^2$. Bounded $C_3$ gives Hermitian $H_n$
and continuity of the source functional. Gaussian moment bounds justify
Taylor expansion for each $H_n$. Taking the approximation before
$d\to0$ gives the benchmark upper limit by a diagonal choice.
\end{proof}

\begin{qcpPlemma}[Sixth order local lower bound]
\label{qcp-P-thm:physical-principle}
For each fixed $0<c\le a$ and $0<\lambda<1$, the gap over all
complete covariant paired measurements satisfies
\begin{equation}
 \liminf_{d\to0}\frac{g_\lambda(d;c)}{d^6}
 \ge q_\lambda(c):=
 \frac{\lambda(1-\lambda)c_bV_bJ_3(c)}
 {9[\lambda c_b+(1-\lambda)(L+R_A)^2]}.
 \label{qcp-P-eq:physical-coefficient}
\end{equation}
Here $b=\kappa_FK$, and the positive constants $c_b$ and
$V_b=\operatorname{Var}_\mu\chi_b>0$ depend only on $b$.
Whenever $J_3(c)>0$, the sixth order lower coefficient is strictly positive.
\end{qcpPlemma}
\begin{proof}
Use \eqref{qcp-S-eq:cov-loc-functional}--\eqref{qcp-S-eq:cov-amp-excess}
and complete failure reports at mesh $\eta(d)=o(|d|^6)$.
Write $\mathsf P_d(\xi)=\int e^{-2idz\xi}\rho_z\nu(dz)$ and
$e_\amp=r_\amp-\mathcal A$.
The reference $W_0(p)=e^{2idpA}$, $Y=A$, has zero amplitude excess
and location risk $r_0(d)$.
For $W=VW_0$, set $G_d=W_0(p)\mathsf{P}_d(p-q)W_0(q)^\dagger$,
$S_d=(G_d+G_d^\dagger)/2$, $T_d=(G_d-G_d^\dagger)/2$.
The difference in location risk is exactly
\begin{align}
 r_\loc(VW_0)-r_0(d)
 &=\iint g_b\tr[(V(p)-V(q))S_d(V(p)-V(q))^\dagger]dpdq\notag\\
 &\quad-2\operatorname{Re}\iint g_b
 \tr[V(q)^\dagger(V(p)-V(q))T_d]dpdq.
 \label{qcp-P-eq:physical-exact-difference}
\end{align}

The unital positive moment map and
$A\Omega^{1/2}=\Omega^{1/2}\mathsf B$ give relative kernel
$e^{2idp\mathsf B}Z_d(p-q)e^{-2idq\mathsf B^\dagger}$, where
$Z_d=\Omega^{-1/2}\mathsf P_d\Omega^{-1/2}$ has derivatives bounded by
$(2L|p-q|)^j$. On $|p|,|q|\le P(d)=\sqrt{40\log(1/|d|)}$, Taylor expansion gives
\begin{align}
 (1-\eta_d)\Omega&\preceq S_d\preceq(1+\eta_d)\Omega,
 &\eta_d&=O(|d|P(d))=o(1),\label{qcp-P-eq:uniform-relative-S}\\
 T_d&=\frac{id^3(p-q)^3}{3}\mathcal R_3
                  +\Omega^{1/2}E_d\Omega^{1/2},
 &\|E_d\|&=O(d^4(1+|p|+|q|)^4).
 \label{qcp-P-eq:uniform-relative-T}
\end{align}
The Personick equation cancels the first two anti-Hermitian orders.
The anti-Hermitian part of
$e^{itA}\int e^{-itz}\rho_z\nu(dz)$ is
$it^3\mathcal R_3/24+O(t^4)$.
Conjugating by $e^{2idqA}$ and using $t=2d(p-q)$ gives the stated
$id^3(p-q)^3\mathcal R_3/3$ term. Spectral tests and the bounded relative
series justify the expansion. Gaussian tails, including all relative
factors, contribute $o(d^{12})$. Cauchy--Schwarz then gives
\[
 r_\loc(VW_0)-r_0(d)\ge(1-o(1))\mathcal D_{b,\Omega}(V)
 -C|d|^3\mathcal D_{b,\Omega}(V)^{1/2}-o(d^{12}).
\]
For a joint rule within $o(d^6)$ of the optimum, comparison with the
reference and positivity of both weights give
$\mathcal D_{b,\Omega}(V)=O(d^6)$ and $e_\amp=O(d^6)$.

Apply Lemma~\ref{qcp-P-lem:weighted-dilation-compactness} with $t=d^3$.
For the amplitude defect, set $D(p)=YV(p)-V(p)A$. Since $W_0$ commutes with $A$,
\[
 (YW-WA)\Omega^{1/2}
 =D\Omega^{1/2}e^{2idp\mathsf B}.
\]
The last factor and its inverse are $I+o(1)$ on the square;
$\|D\|\le2L$ controls the Gaussian exterior.
In the source term, replacing $V(q)^\dagger$ by the aligned $J^\dagger$
has error, for bounded relative $Z$, at most
\[
 \|Z\|\|V(q)-J\|_\Omega
             \|V(p)-V(q)\|_\Omega.
\]
Its integral on the square is $O(d^6P(d)^3)$; the cubic prefactor
therefore gives $o(d^6)$, as does the Gaussian exterior.
The relative bound and finite Gaussian moments make the source continuous.
Joint lower semicontinuity gives
\begin{equation}
 \liminf\frac{r_\loc(VW_0)-r_0(d)}{d^6}\ge
 \Phi(H)=\mathcal D_{b,\Omega}(H)+\frac43\int f_b(p)\tr[H(p)\mathcal R_3]dp,
 \quad\liminf\frac{e_\amp}{d^6}\ge\mathcal E_A(H).
 \label{qcp-P-eq:physical-Phi}
\end{equation}
Completing the square with $\mathsf S_\Omega(C_3)=\mathcal R_3$ gives
$\Phi(H)=\mathcal D_{b,\Omega}(H-H_*)-\mathcal D_{b,\Omega}(H_*)$.
Lemma~\ref{qcp-P-lem:dense-dilation-fields} supplies bounded
approximations to $H_*$ and complete readouts $W=e^{id^3H}W_0$, giving
\begin{equation}
 \limsup_{d\to0}\frac{r_\loc^\star(d)-r_0(d)}{d^6}
 \le-\mathcal D_{b,\Omega}(H_*).
 \label{qcp-P-eq:physical-three-risk-limit}
\end{equation}
Subtract the location benchmark and exact amplitude optimum from
the joint bound in \eqref{qcp-P-eq:physical-Phi} to obtain
\begin{equation}
 \liminf_{d\to0}\frac{g_\lambda(d;c)}{d^6}
 \ge\inf_H\{\lambda\mathcal D_{b,\Omega}(H-H_*)
                    +(1-\lambda)\mathcal E_A(H)\}.
 \label{qcp-P-eq:physical-gap}
\end{equation}
Let $\mathsf T_A$ be the positive bounded operator associated with the commutator
form, so that $\|\mathsf T_A\|\le(L+R_A)^2$. Graph coercivity gives
\[
 \inf_H\{\lambda\mathcal D_{b,\Omega}(H-H_*)+(1-\lambda)\mathcal E_A(H)\}
 \ge\frac{\operatorname{Var}_\mu\chi_b}{9}
 \left\langle C_3,
 \frac{\lambda c_b(1-\lambda)\mathsf T_A}
 {\lambda c_b+(1-\lambda)\mathsf T_A}C_3\right\rangle_\Omega.
\]
Since $\operatorname{Var}_\mu\chi_b=V_b$, bounded functional calculus gives
\eqref{qcp-P-eq:physical-coefficient}.
\end{proof}

\subsubsection{Uniform prior relative bounds and finite section return}
\begin{qcpPlemma}[Relative bound for the cubic observable]
\label{qcp-P-thm:coherent-class}
Uniformly for $0\le c\le a$, the full residual and every unnormalized
finite Fock section satisfy C1's relative $A$ bound and
\begin{equation}
 \|\Omega^{-1/2}C_3\Omega^{1/2}\|
 \le R_C:=8(S_3+3R_AS_2+2R_A^3).
 \label{qcp-P-eq:uniform-RC}
\end{equation}
The constants $S_j,R_A$ are those of
\eqref{qcp-C-eq:relative-constants}.
\end{qcpPlemma}
\begin{proof}
Equation~\eqref{qcp-C-eq:q-dependent-Schur} of C1 gives
$\|Z_j\|_{\rm Sch}\le S_j$, $\|\mathsf B\|_{\rm Sch}\le R_A$.
Submultiplicativity in \eqref{qcp-P-eq:relative-R3} bounds the relative
cubic source by $4S_3+12R_AS_2+8R_A^3$ in this same norm.
If $\omega_i>0$ are the support eigenvalues, the relative Sylvester
solution has entries
\[
 (\Omega^{-1/2}C_3\Omega^{1/2})_{ij}
 =\frac{2\omega_j}{\omega_i+\omega_j}
 (\Omega^{-1/2}\mathcal R_3\Omega^{-1/2})_{ij}.
\]
The multiplier is at most two. Applying the same estimate to the compressed
moments $P_N\mathsf J^jP_N$ gives every unnormalized Fock section.
At $c=0$, $A=C_3=0$ on the vacuum support.
\end{proof}

\subsubsection{Returning the witness to the full residual}
Let $P_N$ project onto the first $N$ Fock levels. Construct $A_N,C_{3,N}$
from the unnormalized moments $P_N\Omega_jP_N$ and extend them by zero;
write $\Omega_N=P_N\Omega P_N$. Constants are uniform for $c$ near zero.
\begin{qcpPproposition}[Weighted witness return]
\label{qcp-P-prop:witness-return}
For each fixed Fock cutoff $N$, the full and finite witnesses satisfy,
as $c\downarrow0$,
\begin{equation}
 |\sqrt{J_3(c)}-\sqrt{J_{3,N}(c)}|=O(c^{N/4}),\qquad
 J_{3,N}=\|[A_N,C_{3,N}]\|_{\Omega_N}^2.
 \label{qcp-P-eq:witness-return}
\end{equation}
\end{qcpPproposition}
\begin{proof}
The coherent tail gives
$\|\Omega_j-P_N\Omega_jP_N\|_1=O(c^N)$ for $j\le3$.
The full Sylvester residual of $A_N$ is therefore $O(c^N)$ in trace
norm; pairing it with $A_N-A$, whose norm is at most $2L$, gives
$\|A_N-A\|_\Omega=O(c^{N/2})$.
The ordered cubic source, the uniform relative bounds above and
\[
 \|A_N^j-A^j\|_\Omega
 \le\sum_{k=0}^{j-1}L^{j-1-k}R_A^k
       \|A_N-A\|_\Omega,\qquad j=2,3,
\]
give a full cubic Sylvester residual $O(c^{N/2})$.
Since both cubic observables have norm at most $M_3$, the same energy
pairing gives $\|C_{3,N}-C_3\|_\Omega=O(c^{N/4})$.
The commutator difference expands as
\[
 A_N(C_{3,N}-C_3)+(A_N-A)C_3
 -C_{3,N}(A_N-A)-(C_{3,N}-C_3)A.
\]
The ordinary and relative bounds give error $O(c^{N/4})$.
The finite commutator has identical $\Omega$ and $\Omega_N$ weights,
so the triangle inequality proves \eqref{qcp-P-eq:witness-return}.
\end{proof}

\subsection{Reference risks for the width law}
\begin{qcpPlemma}[Frequency envelope for the Personick reference]
\label{qcp-P-lem:frequency-envelope}
Let $\mathsf P(t)=\int e^{-2itz}\rho_z\nu(dz)$ and
$q_A(t)=\operatorname{Re}\tr(e^{2itA}\mathsf P(t))$.
The complete field $W_A(p)=e^{2idpA}$ has zero amplitude excess, and
\begin{align}
 r_\loc^\star&\ge2-2\int k_b(\xi)e^{-\xi^2/2}
                    \|\mathsf P(d\xi)\|_1d\xi,
 \label{qcp-P-eq:frequency-lower}\\
 r_\loc(W_A)&=2-2\int k_b(\xi)e^{-\xi^2/2}q_A(d\xi)d\xi.
 \label{qcp-P-eq:frequency-reference}
\end{align}
With $\delta_A(t)=\|\mathsf P(t)\|_1-q_A(t)$,
\[
 0\le\delta_A(t)\le2\mathcal A_L(c)t^2,\qquad
 0\le e_\loc(W_A)\le
 2\int k_b(\xi)e^{-\xi^2/2}\delta_A(d\xi)d\xi .
\]
\end{qcpPlemma}
\begin{proof}
In \eqref{qcp-S-eq:cov-loc-functional}, $W(q)^\dagger W(p)$ is a
contraction. Trace duality and
$\int\varphi(p)\varphi(p-\xi)dp=e^{-\xi^2/2}$ give the two risk
bounds. Spectral calculus yields
$1-q_A(t)\le2t^2\int\|(A-z)\psi_z\|^2\nu(dz)
=2t^2\mathcal A_L(c)$.
Use $\|\mathsf P(t)\|_1\le1$ for the envelope, and subtract the
location lower bound. The amplitude square in
\eqref{qcp-S-eq:cov-amp-excess} vanishes because $W_A$ commutes with $A$.
\end{proof}

\endgroup
\begingroup
\qcpresetvolume{B}
\qcparticlelayout
\setlength{\emergencystretch}{3em}

\renewcommand{\thesection}{B}
\renewcommand{\theequation}{B.\arabic{equation}}
\section{The uniform prior witness for the cubic width law}
\subsection{A noncommuting cubic response at weak residual sensitivity}

For $\nu=\Unif[-1/2,1/2]$, the moments are
\[
 m_j=\int z^j\nu(dz),\qquad
 m_2=\frac1{12},\quad m_4=\frac1{80},\quad m_6=\frac1{448}.
\]
Use Note P's operators $\Omega_j(c),A(c),C_3(c)$ on the average state support.

\begin{qcpBtheorem}[Uniform prior cubic witness]
\label{qcp-B-thm:prior-moment-test}
For the uniform residual prior, as $c\downarrow0$,
\[
 J_3(c):=\|[A(c),C_3(c)]\|_{\Omega(c)}^2
       =\frac{128}{1575^2}c^4+O(c^6).
\]
For every fixed $0<\lambda<1$, Note P's local lower coefficient
$q_\lambda(c)$ is strictly positive for all sufficiently small $c>0$.
\end{qcpBtheorem}

\begin{proof}
In the unnormalized $N=16$ Fock section, the moment matrices are analytic
and $\Omega_N>0$ for $c>0$: a nonzero polynomial has positive squared
integral on the prior interval. Note P's uniform bounds on $A_N,C_{3,N}$
remove their meromorphic Sylvester singularities at zero.

Prior parity makes $A_N,C_{3,N}$ odd in $c$ and off-diagonal in Fock
parity. Write $A_N=cA^{(1)}+O(c^3)$, $C_{3,N}=cC_3^{(1)}+O(c^3)$, and set
$a_{ij}=A^{(1)}_{ij}$, $c_{ij}=(C_3^{(1)})_{ij}$.
The moment expansion is
\[
 (\Omega_j)_{kl}
 =\frac{c^{k+l}}{\sqrt{k!l!}}
   \sum_{r=0}^\infty\frac{(-c^2)^r}{r!}m_{k+l+j+2r}.
\]
At order $c$ the vacuum-row equations give
$a_{01}=2m_2$, $c_{01}=8(m_4-3m_2^2)$, and all
$a_{0j}=c_{0j}=0$ for $j\ge3$.
The $(1,2)$ Personick equation at order $c^3$ is
\[
 \tfrac12\bigl(m_2 a_{12}+m_2 a_{01}/\sqrt2\bigr)
 =m_4/\sqrt2.
\]
The remaining $a_{1j}$, $j\ge3$, vanish by the same equation and parity.

If $r_{12}^{(3)}$ is the $c^3$ coefficient of the cubic right side,
retaining its operator order gives
\[
 r_{12}^{(3)}
 =\frac{4m_6}{\sqrt2}
  -6\left(\frac{a_{01}m_4}{\sqrt2}+m_4a_{12}\right)
  +6m_2 a_{01}a_{12}.
\]
The term involving $A^3\Omega+\Omega A^3$ starts at order $c^5$
in this entry. The cubic Sylvester equation gives
\[
 \tfrac12\bigl(m_2 c_{12}+m_2 c_{01}/\sqrt2\bigr)
 =r_{12}^{(3)}.
\]
Substitution gives
\[
 a_{01}=\tfrac16,\qquad a_{12}=\tfrac{\sqrt2}{15},\qquad
 c_{01}=-\tfrac1{15},\qquad c_{12}=\tfrac{2\sqrt2}{525}.
\]
The vanishing $a_{1j}$ also gives $c_{1j}=0$ for $j\ge4$.

The two adjacent Fock transitions give
\[
 \begin{aligned}
 [A_N,C_{3,N}]_{02}
 &=c^2\bigl(a_{01}c_{12}-c_{01}a_{12}\bigr)+O(c^4)\\
 &=\frac{8\sqrt2}{1575}c^2+O(c^4).
 \end{aligned}
\]
These exhaust the leading vacuum transitions. The vacuum weight is
$1+O(c^2)$ and the other weights are $O(c^2)$, so parity gives
\[
 J_{3,N}(c)=\frac{128}{1575^2}c^4+O(c^6).
\]

The uniform relative bounds in Lemma~\ref{qcp-P-thm:coherent-class} and
Proposition~\ref{qcp-P-prop:witness-return} give
$|\sqrt{J_3}-\sqrt{J_{3,16}}|=O(c^4)$.
As $\sqrt{J_{3,16}}=O(c^2)$, this changes the witness by $O(c^6)$.
Note P then gives $q_\lambda(c)>0$ for fixed $0<\lambda<1$ and all
sufficiently small $c>0$.

\end{proof}

\endgroup
\begingroup
\qcpresetvolume{U}
\qcparticlelayout
\setlength{\emergencystretch}{3em}

\renewcommand{\thesection}{U}
\renewcommand{\theequation}{U.\arabic{equation}}
\section{Proof of the cubic width law under the full prior}
Note S gives the limit as $s\to\infty$ at fixed $0<w<1$.
We match the lower bound from P and B with a complete upper readout,
then let $w\downarrow0$.
For complete rises write $g_\lambda(d;c)=g_\lambda(\kappa_F,d,c)$,
$\kappa_F=\sqrt{10/7}$; $g_{\lambda,w}(u)$ includes truncated rises.
\subsection{A complete sixth order upper readout}

Use Proposition~\ref{qcp-C-prop:relative-residual}'s bound
$\|\Omega^{-1/2}A\Omega^{1/2}\|\le R_A$ on $c^2L^2\le\log2$,
including $c=0$, and write $\mathsf P(t)=\int e^{-2itz}\rho_z\nu(dz)$.

\begin{qcpUlemma}[Uniform sixth order reference regret]
For $L=1/2$ and $0\le c\le a$, a constant $C=C(L,R_A)$ satisfies
\[
 \begin{aligned}
 \delta_A(t)&=\|\mathsf P(t)\|_1-\operatorname{Re}\tr(e^{2itA}\mathsf P(t)),\\
 0&\le\delta_A(t)\le C t^6\qquad(t\in\mathbb R).
 \end{aligned}
\]
The complete reference field $W_A=e^{2idpA}$ has
$e_\amp=0$ and $e_\loc\le2CJ_6(\kappa K)d^6$, where
$J_6(b)=\int k_b(\xi)e^{-\xi^2/2}\xi^6d\xi\le16/\pi$.
\end{qcpUlemma}
\begin{proof}
Set $H_A=2(R_A+L)$, $v_A=L^2/3$ and
$T=e^{2itA}\mathsf P(t)=H_R+iH_I$, with $H_R,H_I$ self-adjoint.
With $x=H_A|t|$, relative Taylor expansion and Personick cancellation give
$H_R\succeq(2-e^x)\Omega$ and
$\|\Omega^{-1/2}H_I\Omega^{-1/2}\|\le x^3e^x/6$.
For $H_R\succeq\eta_0\Omega$ and $H_I=\Omega^{1/2}Z\Omega^{1/2}$,
a contraction $V$ has excess reward at most
$-\eta_0\epsilon_V^2/2+\epsilon_V\|Z\|$, with
$\epsilon_V=\|V-I\|_\Omega$. Trace duality gives
$\|H_R+iH_I\|_1-\tr H_R\le\|Z\|^2/(2\eta_0)$.
Thus $\delta_A(t)\le H_A^6t^6/16$ for $|t|\le\log(3/2)/H_A$.
Elsewhere use $\delta_A(t)\le2v_At^2$ from
Lemma~\ref{qcp-P-lem:frequency-envelope}. One may take
$C=\max\{H_A^6/16,2v_AH_A^4/\log(3/2)^4\}$.
Integrating against Note P's frequency envelope gives the regret bound.
The kernel bound follows from
$k_b(\xi)\le1/(\pi|\xi|)$ and
$\int|\xi|^5e^{-\xi^2/2}d\xi=16$, uniformly for every $b>0$.
\end{proof}

\subsection{Integrating complete and terminal rises under the full prior}
For the quintic profile, set
\[
 \begin{gathered}
 I_F=\int_0^1F^2=181/462,\quad\kappa_F^2=10/7,\\
 \gamma_F=\frac{a^2}{4\kappa_F^2}=\frac{7\log2}{10},\qquad
 \vartheta=1-I_F+\frac1{4\kappa_F^2}=\frac{7237}{9240}.
 \end{gathered}
\]
\begin{qcpUtheorem}[Complete upper readout under the full prior]
\label{qcp-U-thm:effective-full-prior}
At every sufficiently small fixed $w>0$, the strong signal limit admits
a complete common readout with full prior excesses
$e_\amp=0$ and $e_\loc=O(w^3)$, including terminally truncated rises.
\end{qcpUtheorem}
\begin{proof}
Pair $W_A=e^{2idpA}$ with the exact Personick report in each fiber.
For complete rises, $u\le1-w$,
\[
 \kappa=\kappa_F,\qquad d=-\sqrt{\gamma_Fw},\qquad
 c^2=a^2(1-u-\vartheta w),\qquad c^2/4\le\log2.
\]
The uniform $R_A$ bound gives
$e_\loc\le C_0d^6=C_0\gamma_F^3w^3$, with $C_0$ independent of $u,w$.
For truncated rises, set $\tau=(1-u)/w$. Orthogonal decomposition gives
\[
 c^2\le a^2wI_F,\qquad
 d^2=\frac{a^2wF(\tau)^4}{4\int_0^{\tau}(F')^2}\le a^2wI_F.
\]
The last inequality is Cauchy--Schwarz applied to
$F(\tau)^2/2=\int_0^\tau FF'$.
For small $w$, the uniform $R_A$ and $J_6\le16/\pi$ bounds give terminal
regret at most $C_0(a^2I_F)^3w^3$, also as $\kappa\downarrow0$.
Its region has mass $w$, hence integrated regret $O(w^4)$.
The endpoint $u=1$ has zero mass; Note S integrates the complete fields
or their measurable approximations for all three risks.

\end{proof}

\begin{qcpUcorollary}[Ordered cubic width law under the full prior]
\label{qcp-U-cor:width}
For fixed $0<\lambda<1$, the full prior gap satisfies
$\overline{\G}_\lambda(w)=\Theta_\lambda(w^3)$ as $w\downarrow0$.
\end{qcpUcorollary}
\begin{proof}
Note B, Theorem~\ref{qcp-B-thm:prior-moment-test}, gives
$J_3(c)=128c^4/1575^2+O(c^6)$.
Choose $c_0>0$ such that $J_3(c)\ge64c^4/1575^2$ for
$0<c\le c_0$, and fix
$0<c_-<c_+<\min\{c_0,a\}$. Set $\mathcal C=[c_-,c_+]$.
For all sufficiently small $w$, the positions
\[
 u_w(c)=1-c^2/a^2-\vartheta w,\qquad c\in\mathcal C,
\]
lie in $(0,1-w)$. This change of variables in the uniform position
integral fixes the residual at $c$ and has Jacobian $2c/a^2$.
Note P applies to either sign of $d$, giving
\[
 \liminf_{w\downarrow0}\frac{g_\lambda(-\sqrt{\gamma_F w};c)}{w^3}
 \ge\gamma_F^3q_\lambda(c)>0.
\]
For fixed $b=\kappa_FK$, the constants $c_b,V_b$ are fixed.
The uniform $R_A$ bound and $J_3(c)\ge64c_-^4/1575^2$ in
\eqref{qcp-P-eq:physical-coefficient} give $q_\lambda(c)\ge q_*>0$
on $\mathcal C$. By Note S,
$\overline{\G}_\lambda(w)=\int_0^1g_{\lambda,w}(u)du$.
Nonnegativity and Fatou's lemma give
\[
 \liminf_{w\downarrow0}\frac{\overline{\G}_\lambda(w)}{w^3}
 \ge\gamma_F^3\int_{\mathcal C}q_\lambda(c)\frac{2c}{a^2}dc>0.
\]
The complete upper readout gives
$\overline{\G}_\lambda\le\lambda e_\loc=O(w^3)$.
\end{proof}

\endgroup

\pdfbookmark[1]{V. The finite phase record certificate}{qcp-SM-part5}
\section*{V. The finite phase record certificate}
\begingroup
\qcpresetvolume{V}
\qcparticlelayout
\setlength{\emergencystretch}{3em}

\renewcommand{\thesection}{V}
\renewcommand{\theequation}{V.\arabic{equation}}
\section{Certificate for the original record of 16 qubits}
Fix $s=s_{\mathrm c}$, equal task weights and known $L\in[7/16,9/16]$,
with the main text calibration and losses.
\subsection{Physical primal and dual bounds}
\label{qcp-V-app:finite-width}
The record lies in the span of vacuum and the orthonormal vectors
$\ket{e_{j,l}}=\ket0^{\otimes(j-1)}\ket1\ket{D_l^{n-j}}$,
$1\le j\le n$, $0\le l\le n-j$, where $D_l^m$ is the normalized
Dicke vector with $l$ excitations. For $n=16$, the dimension is
$1+n(n+1)/2=137$.
In the phase model, $q_0=1/4$ and
$\alpha=\theta/2=(s+az)/\sqrt n$. Grouping by the first excitation gives
\[
 \langle0^n|\Psi_{k,\alpha}\rangle=(\cos\alpha)^{n-k+1},\quad
 \langle e_{j,l}|\Psi_{k,\alpha}\rangle=
 \begin{cases}\sqrt{\binom{n-j}{l}}(\cos\alpha)^{n-k-l}(\sin\alpha)^{l+1},&j\ge k,\\
 0,&j<k.\end{cases}
\]
Zeroth powers equal one.

For normalized Gauss quadrature with $q=32$, put
$\omega=2a\sqrt n$ and $C_q=(q!)^4/[(2q+1)((2q)!)^3]$.
The generator norm is at most $a\sqrt n$, hence
$\|\partial_z^m\rho\|_1\le\omega^m$ and, for $r=0,1,2$,
\[
 e_r=C_q(2L)^{2q}\sum_{j=0}^r\binom{2q}{j}
 \frac{r!L^{r-j}\omega^{2q-j}}{(r-j)!},\quad
 E_\loc=2e_0,\quad E_\amp=L^2e_0+2Le_1+e_2.
\]
For a normalized dyadic contraction $\mathsf K$ into dimension $D$,
let $\mathsf X$ be its successful weighted coordinates. If
$\|\mathsf X-\mathsf X_{\rm c}\|_F\le e$, $|z-z_c|\le e_z$ and
the location table error is $e_\ell$, the cost allowances are
\[
 \zeta_\loc=2e(2+e)+e_\ell(1+e)^2,\qquad
 \zeta_\amp=4L^2e(2+e)+4Le_z(1+e)^2.
\]
The lower risks include $D\zeta_j+B_j\sqrt{2\delta}+E_j$ and the
exact dual repairs, where $\delta$ is compression failure mass.
Joint fees are the weighted task fees.

Set $S_r=\mathsf X_{\rm c}\operatorname{diag}(z_c^r)\mathsf X_{\rm c}^\dagger$
and $E_A=S_0\widehat A+\widehat A S_0-2S_1$ for a dyadic Hermitian
$\widehat A$. The amplitude dual
$Z=S_2-\widehat A S_0\widehat A-L\|E_A\|I$ is feasible for $|y|\le L$:
its slack is
$(yI-\widehat A)S_0(yI-\widehat A)+yE_A+L\|E_A\|I\succeq0$.
Exact integer residuals certify the joint dual at integer locations
and 129 amplitude nodes. Subtracting
$(1-\lambda)\Delta_y^2S_0/4$, $\Delta_y=2L/128$, extends feasibility
to all amplitude reports. The position cost is concave between integer
reports, so those actions suffice.

Normalize the upper frame by its exact Gram bound, pull it through
$\mathsf K$ and add the positive complement with fallback amplitude zero.
Starting from prior risk $L^2/3$, the allowance for effect improvements
is $2D\zeta_j+E_j$. The gap lower bound is
$\underline R-\lambda\overline r_\loc-(1-\lambda)\overline r_\amp$.
The node endpoints use exact dyadic widths and actions and ball
enclosures at 256 bit precision.

\subsection{Interpolation in prior width}
On $I_L=[7/16,9/16]$, put $b_n=\sin^2(s_{\mathrm c}/\sqrt n)$,
$\xi=z/L$ and $y=L\eta$. Fixed POVM effects and normalized reports
$\eta$ give losses $2\min\{1,b_n|\widehat k-k|/K\}$ and
$L^2(\eta-\xi)^2$ on a common report space.
For $N=n-k+1$ changed sites and
$\alpha=(s_{\mathrm c}+aL\xi)/\sqrt n$, the Pauli generator moments give
$\|\psi'\|^2=a^2\xi^2N/n$ and
$\|\psi''\|^2=a^4\xi^4(3N^2-2N)/n^2$. Thus
\begin{equation}
 \|\rho'\|_1\le2a|\xi|\sqrt{N/n},\qquad
 \|\rho''\|_1\le2(1+\sqrt3)a^2\xi^2N/n.
 \label{qcp-V-eq:width-state-derivatives}
\end{equation}
Primes denote $L$ derivatives, which have trace zero.
Set $\beta_1=n^{-1}\sum_{N=1}^n\sqrt{N/n}$,
$\beta_2=(n+1)/(2n)$, $C_2=(1+\sqrt3)a^2\beta_2/3$ and $L_{\max}=9/16$.
The risk curvature bounds are
\begin{equation}
 M_\loc=2C_2,\quad M_\amp=8+8L_{\max}a\beta_1+4L_{\max}^2C_2,\quad
 M_\lambda=\lambda M_\loc+(1-\lambda)M_\amp.
 \label{qcp-V-eq:width-curvatures}
\end{equation}
\begin{qcpVlemma}[Curvature and interpolation]
Every fixed complete paired POVM satisfies $|r_j''(L)|\le M_j$ on $I_L$.
For $[L_-,L_+]$, put $\Delta_L=L_+-L_-$,
$\tau=(L-L_-)/\Delta_L$ and $\epsilon_j=M_j\Delta_L^2/8$. Then
\begin{align}
 r_j^\star(L)&\ge(1-\tau)\underline r_j(L_-)+\tau\underline r_j(L_+)-\epsilon_j,
 \label{qcp-V-eq:width-optimum-lower}\\
 r_j^\star(L)&\le r_j(M;L)\le(1-\tau)\overline r_j(M;L_-)
                         +\tau\overline r_j(M;L_+)+\epsilon_j.
 \label{qcp-V-eq:width-primal-upper}
\end{align}
For the joint risk use $M_\lambda$; minimize the upper bound over
complete strategies held fixed at both endpoints.
\end{qcpVlemma}
\begin{proof}
For Hermitian $X$ of trace zero and $0\preceq Q\preceq BI$,
$|\tr XQ|\le B\|X\|_1/2$. Taking $B=2$ and
$\mathbb E\xi^2=1/3$ gives $M_\loc$.
For amplitude, put $Q_\xi=\int(\eta-\xi)^2M\preceq4I$ and
$f=\tr\rho Q_\xi$. Equation~\eqref{qcp-V-eq:width-state-derivatives}
gives $f\le4$, $|f'|\le4a|\xi|\sqrt{N/n}$ and
$|f''|\le4(1+\sqrt3)a^2\xi^2N/n$.
Average $(L^2f)''=2f+4Lf'+L^2f''$, using $\mathbb E|\xi|=1/2$,
to obtain $M_\amp$. These uniform bounds justify differentiation.
Linear interpolation has error at most
$M_j(L-L_-)(L_+-L)/2\le\epsilon_j$.
Take infima after this estimate for the lower bound; fix a strategy
for the upper bound.
\end{proof}
The nine nodes $L_i=(28+i)/64$, $i=0,\ldots,8$, give eight cells.
Subtract marginal envelopes from the joint envelope; divide by the
positive marginal baseline for relative gaps. Comparisons at endpoints and affine crossings bound each cell. At a new width, keep an endpoint POVM,
rescale its reports by known $L/L_i$, retain its complement and fallback
zero, and apply the quadrature allowance at that width.

\subsection{Numerical endpoints used by the width proof}
The following node enclosures are the separately certified risks in the exact
physical record. All decimal endpoints are rounded outward; the accompanying
numerical certificate archive retains the exact integer inputs, continuous
prior fees, complete upper-rule frames and all-report dual repairs.
\begin{center}\small
\begin{tabular}{lccc}
\toprule
$L$ & $r_\loc^\star$ & $r_\amp^\star$ & $R_{1/2}^\star$\\
\midrule
$7/16$ & $[0.090182574,\,0.090184177]$ & $[0.048585258,\,0.048585730]$ & $[0.070338846,\,0.070347021]$\\
$29/64$ & $[0.090866566,\,0.090869583]$ & $[0.051282099,\,0.051282688]$ & $[0.072120947,\,0.072130336]$\\
$15/32$ & $[0.091572917,\,0.091574446]$ & $[0.053991553,\,0.053992135]$ & $[0.073921146,\,0.073931016]$\\
$31/64$ & $[0.092297122,\,0.092298894]$ & $[0.056710168,\,0.056710800]$ & $[0.075737028,\,0.075746892]$\\
$1/2$ & $[0.093041191,\,0.093042611]$ & $[0.059435032,\,0.059435706]$ & $[0.077565023,\,0.077575534]$\\
$33/64$ & $[0.093804044,\,0.093805435]$ & $[0.062163472,\,0.062164157]$ & $[0.079403967,\,0.079415285]$\\
$17/32$ & $[0.094585665,\,0.094587198]$ & $[0.064892969,\,0.064893713]$ & $[0.081249595,\,0.081263592]$\\
$35/64$ & $[0.095386330,\,0.095387690]$ & $[0.067621431,\,0.067622180]$ & $[0.083107292,\,0.083119188]$\\
$9/16$ & $[0.096205257,\,0.096206805]$ & $[0.070346728,\,0.070347580]$ & $[0.084968096,\,0.084980452]$\\
\bottomrule
\end{tabular}
\end{center}
The complete-readout curvature bounds in \eqref{qcp-V-eq:width-curvatures}
are enclosed above by $M_\loc<2.682760529$ and
$M_\amp<14.904048650$; $M_{1/2}=(M_\loc+M_\amp)/2$.
For cell length $1/64$, these give respective interpolation allowances
$M_j/32768$. Applying the lower and fixed-strategy upper envelopes gives:
\begin{center}\small
\begin{tabular}{lcc}
\toprule
Prior-width cell & Certified $10^3\mathcal G_{1/2}$ & Certified relative cost (\%)\\
\midrule
$[7/16,\,29/64]$ & $[0.417186692,\,1.592709437]$ & $[0.58474971,\,2.30441504]$\\
$[29/64,\,15/32]$ & $[0.508106019,\,1.685486461]$ & $[0.69554358,\,2.38042955]$\\
$[15/32,\,31/64]$ & $[0.601150130,\,1.779952428]$ & $[0.80396469,\,2.45463679]$\\
$[31/64,\,1/2]$ & $[0.695474849,\,1.874127805]$ & $[0.90902819,\,2.52457794]$\\
$[1/2,\,33/64]$ & $[0.789158802,\,1.968233299]$ & $[1.00846907,\,2.59081157]$\\
$[33/64,\,17/32]$ & $[0.882465532,\,2.060980344]$ & $[1.10296047,\,2.65195857]$\\
$[17/32,\,35/64]$ & $[0.972433266,\,2.152013640]$ & $[1.18918205,\,2.70792442]$\\
$[35/64,\,9/16]$ & $[1.065651117,\,2.241164557]$ & $[1.27553314,\,2.75884781]$\\
\bottomrule
\end{tabular}
\end{center}
The upper envelopes use the $75$ complete physical strategy cross-evaluations
at neighboring nodes, rather than a width-dependent mixture of incomplete
success branches. Their exact affine coefficients and crossings are supplied
in the archive's \texttt{width\_interval\_certificate.json}; each endpoint
risk includes the fallback probability and a fresh continuous-prior fee.

\begin{qcpVtheorem}[Conflict across prior widths]
For every $L\in[7/16,9/16]$ in the stated $n=16$ experiment,
\begin{equation}
 4.17\times10^{-4}<\mathcal G_{1/2}(L)<2.242\times10^{-3},\qquad
 0.5847\%<\frac{\mathcal G_{1/2}(L)}{(r_\loc^\star(L)+r_\amp^\star(L))/2}<2.759\%.
 \label{qcp-V-eq:width-main-numbers}
\end{equation}
Both marginals have positive risk and improve on their priors:
\[
 r_\loc^\star>0.0901,\quad r_\amp^\star>0.0481,\quad
 r_\loc^{\rm no\ data}-r_\loc^\star>0.342,\quad L^2/3-r_\amp^\star>0.0147.
\]
\end{qcpVtheorem}
\begin{proof}
The exact cell envelopes give the gap bounds, and the marginal
envelopes give the remaining inequalities.
The amplitude prior risk is $L^2/3$. For position, $b_n\le K/16$
makes the loss unclipped, and the median absolute error of the uniform
$16$-site prior is four, giving prior risk $8b_n/K$.
The risk bounds include compression, rounding, quadrature and dual
repairs, with positive complements in every upper measurement.
\end{proof}

\FloatBarrier

\endgroup
\pdfbookmark[1]{Supplementary references}{qcp-SM-references}

\end{document}